\pdfoutput=1
\def\VersionForArXiV{}

\ifdefined\VersionForArXiV%
	\newcommand{\LongVersion}[1]{#1}
	\newcommand{\ACMVersion}[1]{}
\else
	\newcommand{\LongVersion}[1]{}
	\newcommand{\ACMVersion}[1]{#1}
\fi

\newcommand{\LongVersionActuallyNot}[1]{#1}

\ACMVersion{
  \newcommand{\myparagraph}[1]{\noindent\textbf{#1}\quad}
}\LongVersion{
  \newcommand{\myparagraph}[1]{\paragraph{#1}}
}

\ifdefined\VersionForArXiV
	\PassOptionsToPackage{svgnames,table}{xcolor}
	\documentclass[]{article}
\else
	\documentclass[acmsmall]{acmart}
\fi

\ACMVersion{
\AtBeginDocument{%
  \providecommand\BibTeX{{%
    \normalfont B\kern-0.5em{\scshape i\kern-0.25em b}\kern-0.8em\TeX}}}

\setcopyright{acmcopyright}
\acmJournal{TCPS}
\acmYear{2022} \acmVolume{1} \acmNumber{1} \acmArticle{1} \acmMonth{1} \acmPrice{15.00}\acmDOI{10.1145/3529095}
}

\usepackage[utf8]{inputenc}
\usepackage[english]{babel}
\usepackage{csquotes}

\usepackage[ruled,lined,linesnumbered,noend]{algorithm2e}
\SetKwInOut{Input}{input}
\SetKwInOut{Output}{output}
\SetKw{KwPush}{push}
\SetKw{KwPop}{pop}
\SetKw{KwFrom}{from}
\SetKw{KwCompute}{compute}
\SetAlgorithmName{Procedure}{procedure}{List of Procedures}

\usepackage{subcaption}
\usepackage{paralist} %
\usepackage{xspace}

\newenvironment{ienumeration}
	{\ifdefined\VersionForArXiV\begin{enumerate}\else\begin{inparaenum}[\itshape i\upshape)]\fi}
	{\ifdefined\VersionForArXiV\end{enumerate}\else\end{inparaenum}\fi}

\LongVersion{\usepackage{amsthm}}
\usepackage{amsmath} %
\LongVersion{\usepackage{amssymb}} %
\usepackage{multirow}
\usepackage{stmaryrd}

\usepackage[misc,geometry]{ifsym} %

\usepackage{eurosym} %

\usepackage{listings}
\usepackage{color}

	\definecolor{mygreen}{rgb}{0,0.6,0}
	\definecolor{mygray}{rgb}{0.5,0.5,0.5}
	\definecolor{mymauve}{rgb}{0.58,0,0.82}
	\definecolor{weborange}{RGB}{255,165,0}

\lstdefinestyle{log}{
	backgroundcolor=\color{white},   %
	basicstyle=\scriptsize,        %
	breakatwhitespace=false,         %
	breaklines=true,                 %
	captionpos=b,                    %
	commentstyle=\color{mygreen},    %
	deletekeywords={...},            %
	escapeinside={\%*}{*)},          %
	extendedchars=true,              %
	frame=single,	                   %
	keepspaces=true,                 %
	keywordstyle=\color{red!70!black}\bfseries,       %
	morekeywords={@, open, close, update},            %
	numbers=left,                    %
	numbersep=5pt,                   %
	numberstyle=\tiny\color{mygray}, %
	rulecolor=\color{black},         %
	showspaces=false,                %
	showstringspaces=false,          %
	showtabs=false,                  %
	stepnumber=1,                    %
	stringstyle=\color{mymauve},     %
	tabsize=2,	                   %
	classoffset=1, %
	otherkeywords={@},
	morekeywords={@},
	keywordstyle=\color{weborange},
	classoffset=0,
}

\ifdefined\VersionForArXiV
	\usepackage[backend=biber,backref=true,style=alphabetic,url=false,doi=true,defernumbers=true,sorting=anyt,maxnames=99]{biblatex} %
	\DeclareSourcemap{
	\maps[datatype=bibtex]{
		\map{
			\step[fieldsource=NOdoi, fieldtarget=doi]
			\step[fieldsource=NOeditor, fieldtarget=editor]
			\step[fieldsource=NOlocation, fieldtarget=location]
			\step[fieldsource=NOmonth, fieldtarget=month]
		}
	}
	}

	\renewbibmacro*{doi+eprint+url}{%
		\iftoggle{bbx:doi}
			{\color{black!40}\footnotesize\printfield{doi}}
			{}%
		\newunit\newblock
		\iftoggle{bbx:eprint}
			{\usebibmacro{eprint}}
			{}%
		\newunit\newblock
		\iftoggle{bbx:url}
			{\usebibmacro{url+urldate}}
			{}%
	}

\fi
\ifdefined\VersionWithComments%
	\usepackage{draftwatermark}
	\SetWatermarkText{draft}
	\SetWatermarkScale{3}
	\SetWatermarkColor[gray]{0.9}
\fi
\usepackage{xcolor}
\usepackage{colortbl}

\LongVersion{%
	\definecolor{USPNcobalt}{HTML}{293358}
	\definecolor{USPNocre}{HTML}{8b7d6d}
	\definecolor{USPNblanc}{HTML}{ffffff}
	\definecolor{USPNceruleen}{HTML}{354878}
	\definecolor{USPNsable}{HTML}{ad947e}

	\usepackage[%
			pdfauthor={Masaki Waga, Etienne Andre, and Ichiro Hasuo},
			pdftitle={Model-Bounded Monitoring of Hybrid Systems},
			breaklinks  = true,
			colorlinks=true,
			citecolor   = USPNsable,
			linkcolor   = USPNocre,
			urlcolor    = USPNceruleen,
		]{hyperref}
}

\usepackage[capitalise,english,nameinlink]{cleveref} %
\crefname{line}{\text{line}}{\text{lines}} %
\crefname{item}{\text{item}}{\text{items}} %
\crefname{example}{\text{Example}}{\text{Examples}} %
\crefname{assumption}{\text{Assumption}}{\text{Assumptions}} %
\crefname{algorithm}{\text{Procedure}}{\text{Procedures}}

\usepackage{wrapfig}
\usepackage{tikz}
\usetikzlibrary{arrows,automata,positioning}
\tikzstyle{every node}=[initial text=]
	\tikzstyle{location}=[circle, minimum size=12pt, draw=black, fill=blue!10, inner sep=1pt] %
\tikzstyle{invariant}=[draw=black, dotted, inner sep=1pt] %
\tikzstyle{final}=[double distance=3pt]
\tikzstyle{accepting}=[final]
\tikzstyle{PTPMOPT}=[,dashed,color=red,semithick]

\newcommand{\styleact}[1]{\ensuremath{\textcolor{coloract!80!black}{\mathrm{#1}}}}
\newcommand{\stylebenchmark}[1]{\textcolor{red!30!black}{\textsc{#1}}}
\newcommand{\stylevar}[1]{\ensuremath{\textcolor{colorclock!80!black}{#1}}}
	\definecolor{coloract}{rgb}{0.50, 0.70, 0.30}
	\definecolor{colorclock}{rgb}{0.4, 0.4, 1}
	\definecolor{colorconst}{rgb}{0.50, 0.20, 0.00}
	\definecolor{colordisc}{rgb}{1, 0, 1}
	\definecolor{colorloc}{rgb}{0.4, 0.4, 0.65}
	\definecolor{colorparam}{rgb}{1, 0.6, 0.0}
	\definecolor{colorvar}{rgb}{0.6, 0.7, 1}
	\definecolor{colorlvar}{rgb}{0.4, 0.4, .5}
	\definecolor{colordparam}{rgb}{.9, 0.8, 0.0}
\usepackage{wasysym}

\usepackage{gnuplot-lua-tikz}
\usepackage{pgfplotstable}
\pgfplotsset{compat=1.12}
\usepackage{booktabs}

\newcommand{\init}{_0}

\newcommand{\setQ}{{\mathbb Q}}

\newcommand{\setZ}{{\mathbb Z}}
 \newcommand{\interval}{\mathcal{I}}

\newcommand{\Variables}{\mathbb{X}} %
\newcommand{\variable}{x} %
\newcommand{\variablei}[1]{\variable_{#1}} %
\newcommand{\variabledoti}[1]{\ensuremath{\dot{\variable}_{#1}}} %
\newcommand{\VariablesCard}{M}

\newcommand{\update}[2]{\ensuremath{[#1]_{#2}}}
\newcommand{\updates}{\mu}

\newcommand{\Guard}{\Phi}
\newcommand{\GuardP}{\Guard(\mathbf{x})}
\newcommand{\GuardD}{\Guard(\dot{\mathbf{x}})}
\newcommand{\guard}{g}

\newcommand{\A}{\ensuremath{\mathcal{M}}}
\newcommand{\signal}{\sigma}

\newcommand{\assign}{\leftarrow}
\newcommand{\clockabs}{\ensuremath{t_\mathit{abs}}} %
\newcommand{\clockrel}{\ensuremath{t_\mathit{rel}}} %
\newcommand{\varval}{v} %
\newcommand{\wordVal}{u} %
\newcommand{\completeWordVal}{\varval} %

\newcommand{\compOp}{\bowtie}

\newcommand{\Val}{({\Rgeqzero})^{\Variables}}

\newcommand{\edge}{e}
\newcommand{\Edges}{E}
\newcommand{\longueflecheRel}[1]{\stackrel{#1}{\mapsto}}

\newcommand{\flecheRel}{{\rightarrow}}

\newcommand{\invariant}{\mathrm{Inv}}
\newcommand{\Duration}{\mathrm{Dur}}
\newcommand{\Flow}{\mathcal{F}}
\newcommand{\flow}{f}
\newcommand{\K}{K}

\newcommand{\Lmon}{\mathcal{L}_{\mathrm{mon}}}
\newcommand{\Lg}{\Lmon}
\newcommand{\partialLg}{\mathcal{L}^{\mathrm{p}}_{\mathrm{mon}}}
\newcommand{\loc}{\ell} %

\newcommand{\Loc}{L} %
\newcommand{\LocFinal}{F}

\newcommand{\product}{\mathrel{||}}
\newcommand{\R}{{\mathbb{R}}}
\newcommand{\Rgeqzero}{\R_{\geq 0}}

\newcommand{\sinit}{s\init} %
\newcommand{\Sinit}{S\init} %
\newcommand{\States}{S} %

\newcommand{\VarValInit}{\mathrm{Init}}

\newcommand{\wloc}{w\loc}
\newcommand{\word}{\textcolor{colorok}{w}}

\newcommand{\DPValuate}[4]
	{\ensuremath{\reset{#3}{\valuate{\valuate{#1}{#2}}{#4}}}}  %

\newcommand{\partfun}{\nrightarrow}
\newcommand{\domain}[1]{\mathrm{dom}(#1)}

\newcommand{\reset}[2]{\ensuremath{[#1]_{#2}}}
\newcommand{\valuate}[2]{\ensuremath{#2(#1)}}
\newcommand{\Rp}{{\mathbb{R}_{>0}}}
\newcommand{\Rnn}{{\Rgeqzero}}

\newcommand{\disjointUnion}{\sqcup}

\newcommand{\matching}{\mathsf{C}}

\usepackage{pifont}%

\newcommand{\Init}{\VarValInit}
\newcommand{\run}{\rho}

\newcommand{\Resulti}[1]{\mathit{Result}_{#1}}

\newcommand{\Conf}{\mathit{State}}

\newcommand{\tbcolor}{\cellcolor{green!25}}

\newcommand{\TO}{\cellcolor{red!25} T.O.}

\newcommand{\defProblem}[3] {%
\noindent\fcolorbox{black}{blue!15}{ %
	\smallskip
	
	\begin{minipage}{.95\columnwidth}
		\textbf{#1 problem:}\\
		\textsc{Input}: #2\\
		\textsc{Problem}: #3
	\end{minipage}
}
	
	\medskip
	
}

\let\oldquote\quote
\renewcommand\quote{\em\oldquote}

\definecolor{vertfonce}{rgb}{0.0, 0.5, 0.0}
\definecolor{rougefonce}{rgb}{1, 0.0, 0.0}

	\definecolor{cellcolor}{rgb}{.8, .8, 1}
\theoremstyle{plain}
\newtheorem{theorem}{Theorem} %
\newtheorem{lemma}[theorem]{Lemma}

\theoremstyle{definition}
\newtheorem{definition}[theorem]{Definition}
\newtheorem{example}[theorem]{Example}

\theoremstyle{remark}
\newcommand{\stylealgo}[1]{\ensuremath{\textsf{#1}}}

\newcommand{\TransWord}{\stylealgo{TQW2LHA}}

\newcommand{\phaver}{PHAVer}
\newcommand{\phaverlite}{PHAVerLite}

\newcommand{\masakiTool}{\textsc{HAMoni}}

\newcommand{\Accc}{\stylebenchmark{ACCC}}
\newcommand{\Acci}{\stylebenchmark{ACCI}}
\newcommand{\Accd}{\stylebenchmark{ACCD}}

\newcommand{\AccdSafety}{\stylebenchmark{\Accd}}
\newcommand{\NAV}{\stylebenchmark{NAV}}
\newcommand{\UnsafeNAV}{\stylebenchmark{UnsafeNAV}}
\newcommand{\SharedGasBurner}{\stylebenchmark{GasBurner}}

\ifdefined\VersionWithComments%
 	\definecolor{colorok}{RGB}{80,80,150}
\else
	\definecolor{colorok}{RGB}{0,0,0}
\fi

\newcommand{\eg}{\textcolor{colorok}{e.\,g.,}\xspace}
\newcommand{\ie}{\textcolor{colorok}{i.\,e.,}\xspace}

\newcommand{\wrt}{\textcolor{colorok}{w.r.t.}\xspace}

\usepackage{graphicx}	%
\usepackage{version}	%
\excludeversion{PartialObservationVersion}

\newcommand{\ouracks}{%
	This work is partially supported
	by JST ACT-X Grant No.\ JPMJAX200U, 
	by JST ERATO HASUO Metamathematics for Systems Design Project (No.\ JPMJER1603),
	by JSPS Grant-in-Aid No.\ 18J22498,
	and
	by ANR-NRF ProMiS (ANR-19-CE25-0015 / 2019 ANR NRF 0092).
}

\newcommand{\ourkeywords}{monitoring, cyber-physical systems, hybrid automata}

\newcommand{\ourabstract}{%
	Monitoring of hybrid systems attracts both scientific and practical attention. However, monitoring algorithms suffer from the methodological difficulty of only observing sampled discrete-time signals, while real behaviors are continuous-time signals. To mitigate this problem of sampling uncertainties, we introduce a \emph{model-bounded monitoring} scheme, where we use prior knowledge about the target system to prune interpolation candidates. Technically, we express such prior knowledge by linear hybrid automata (LHAs)---the LHAs are called \emph{bounding models}. We introduce a novel notion of \emph{monitored language} of LHAs, and we reduce the monitoring problem to the membership problem of the monitored language. We present two partial algorithms---one is via reduction to reachability in LHAs and the other is a direct one using polyhedra---and show that these methods, and thus the proposed model-bounded monitoring scheme, are efficient and practically relevant.
}

\LongVersion{
	\newcommand{\orcidID}[1]{\href{https://orcid.org/#1}{\includegraphics[height=1em]{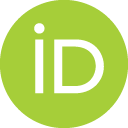}}}
}

\begin{document}

\title{Model-Bounded Monitoring of Hybrid Systems\LongVersion{\thanks{%
	This is the author version of the manuscript of the same name published in the ACM Transactions on Cyber-Physical Systems.
	The final version is available at \href{https://doi.org/10.1145/3529095}{\nolinkurl{10.1145/3529095}}.
	\ouracks{}
}}
}

\ifdefined\VersionForArXiV\else
\author{Masaki Waga}
\email{mwaga@fos.kuis.kyoto-u.ac.jp}
\orcid{0000-0001-9360-7490}
\affiliation{%
  \institution{Kyoto University}
  \streetaddress{Yoshida-honmachi, Sakyo-ku}
  \city{Kyoto}
  \postcode{606-8501}
  \country{Japan}
}
\author{\'Etienne Andr\'e}
\orcid{0000-0001-8473-9555}
\affiliation{%
  \institution{Université de Lorraine, CNRS, Inria, LORIA}
  \streetaddress{54\,506 Vandœuvre-lès-Nancy}
  \city{Nancy}
  \country{France}
}

\author{Ichiro Hasuo}
\email{i.hasuo@acm.org}
\orcid{0000-0002-8300-4650}
\affiliation{%
  \institution{National Institute of Informatics}
  \streetaddress{2-1-2 Hitotsubashi}
  \city{Tokyo}
  \postcode{101-8430}
  \country{Japan}}
\additionalaffiliation{the Graduate University for Advanced Studies}

\renewcommand{\shortauthors}{Masaki Waga, \'Etienne Andr\'e, and Ichiro Hasuo}

\begin{abstract}
	\ourabstract{}
\end{abstract}


\fi
\LongVersion{
    \author{}
	\date{}
}

\maketitle

\LongVersion{%
	\newcommand{\keywords}[1]
	{%
		\small\textbf{\textit{Keywords---}} #1
	}

	\noindent{}\textbf{Masaki Waga\orcidID{0000-0001-9360-7490}}
	\\
	{\em\small{}Kyoto University, Japan}
	\\
	{\em\small{}National Institute of Informatics, Japan}

	\smallskip

	\noindent{}\textbf{Étienne André\orcidID{0000-0001-8473-9555}}
	\\
	{\em\small{}Université de Lorraine, CNRS, Inria, LORIA, Nancy, France}

	\smallskip

	\noindent{}\textbf{Ichiro Hasuo\orcidID{0000-0002-8300-4650}}
	\\
	{\em\small{}National Institute of Informatics, Japan}
	\\
	{\em\small{}The Graduate University for Advanced Studies, Japan}

	\begin{abstract}
		\ourabstract{}
	\end{abstract}

}
\section{Introduction}\label{section:introduction}
\myparagraph{Monitoring} 
Pervasiveness and safety criticalness of \emph{cyber-physical systems (CPSs)}---where physical dynamics are controlled by software---pose their quality assurance as a pressing industrial and social problem. A number of research efforts have aimed at their correctness proofs, with software science and control theory collaborating hand-in-hand.

However, such \emph{exhaustive verification} is often very hard with real-world examples. This is because \emph{white-box models} of real-world CPSs are hard to find---the difficulty can be because of 1) the systems' complexity,
2) uncertainties in their operation environments, and 3) third-party black-box components.
Mathematically, formal verification is to give a \emph{proof} that a system is correct, and a white-box model is a \emph{definition} of the system. Without a white-box model, there is no definition to build a proof on.

\emph{Monitoring} is attracting attention as a light-weight yet feasible alternative in quality assurance of CPSs. Monitoring consists in checking whether a sequence of data 
satisfies  a specification expressed using some formalism. It can be used offline (\eg{} for extracting interesting parts from a huge log) and online (\eg{} for alerting to unsafe phenomena). See the related work in this section for references.

\begin{figure}[tbp]
\centering
 \includegraphics[bb=102 312 893 550,clip,width=\ACMVersion{.8\linewidth}\LongVersion{\linewidth}]{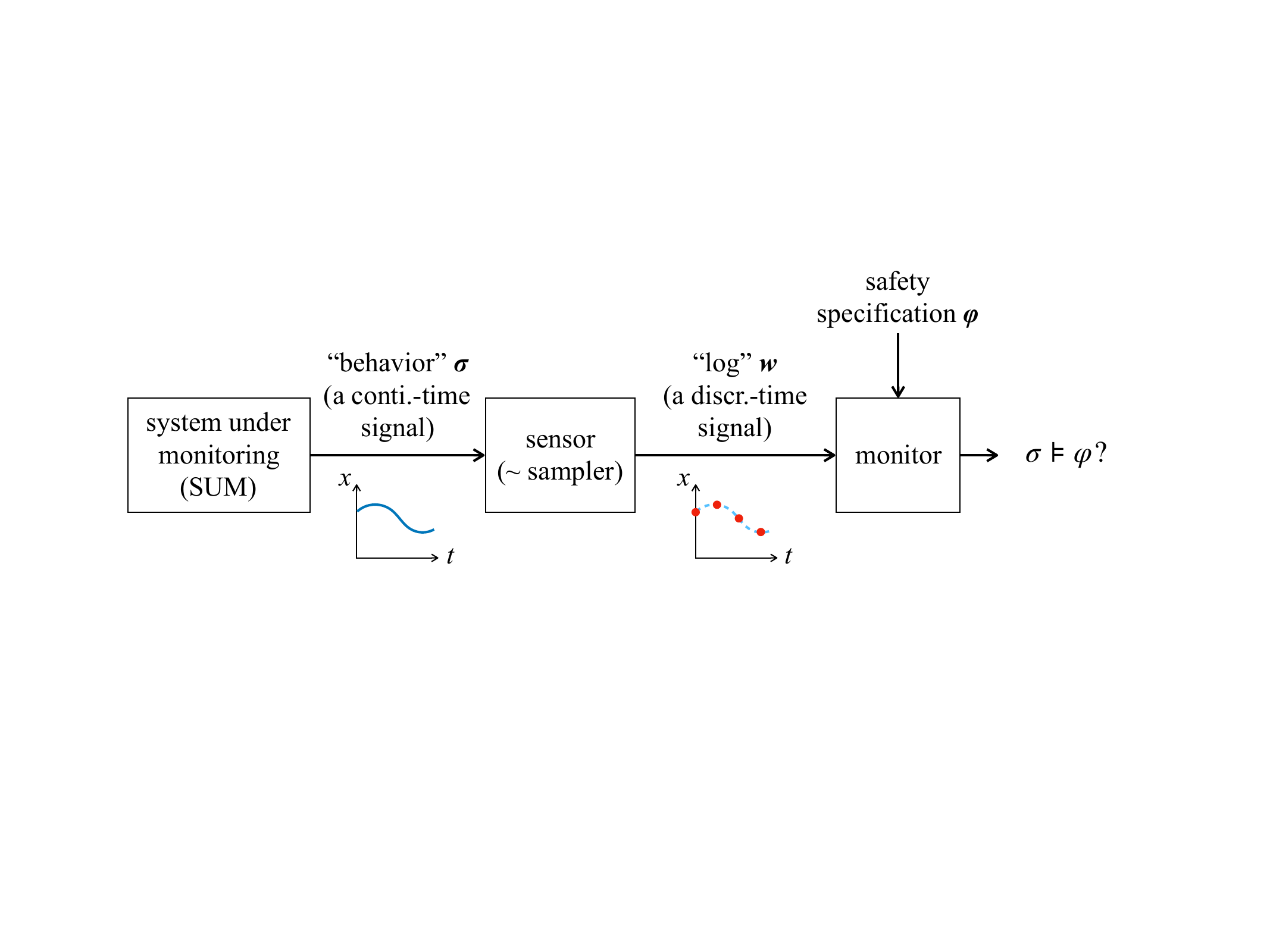}
\caption{Hybrid system monitoring and sampling uncertainties}
\label{figure:hybridSysMonitoring}
\end{figure}

\ACMVersion{\vspace*{.2em}}
\myparagraph{Hybrid System Monitoring} 
In this paper, we study monitoring of CPSs, with a particular emphasis on their \emph{hybrid} aspect (\ie{} the  interplay between continuous and discrete worlds).

We sketch the workflow of hybrid system monitoring in \cref{figure:hybridSysMonitoring}. We are given a specific behavior $\sigma$ of the system under monitoring (SUM) and a specification $\varphi$ (in this paper, we focus on safety specifications). The problem is to decide whether  $\sigma$ is safe or not, in the sense that $\sigma$ satisfies $\varphi$. We assume a computer solves this problem. Therefore, an input to a monitor must be a discrete-time signal $\word$, obtained from the continuous-time signal $\sigma$ via sampling. 
We shall call such $\word$ a \emph{log} of the SUM induced by the behavior $\sigma$.

\vspace*{.2em}
\begin{wrapfigure}[4]{r}{0pt}
\raisebox{-.5\height}{\includegraphics[bb=412 527 612 601,clip,width=.3\linewidth]{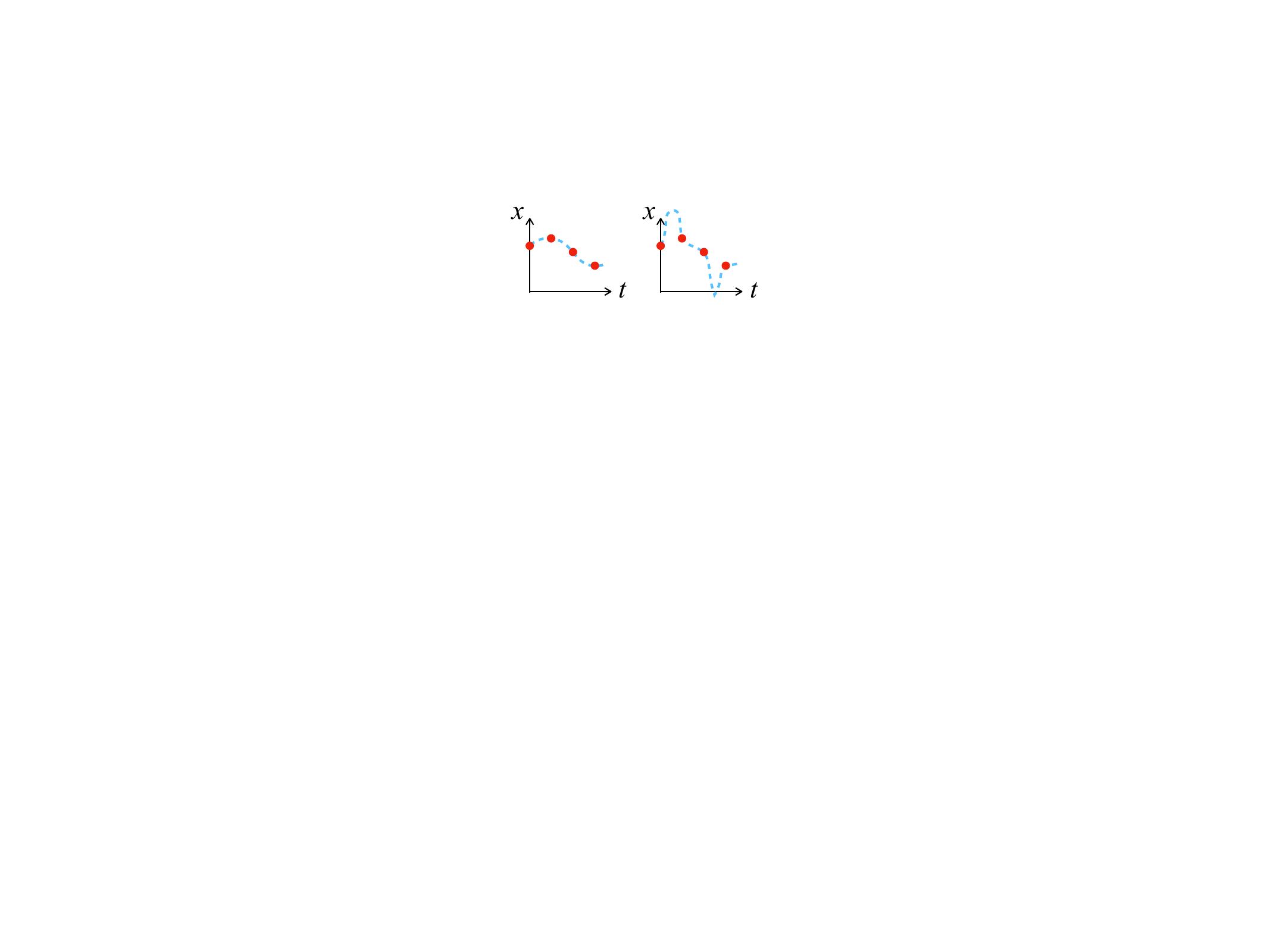}}
\caption{$\word$ and $\sigma$}
\label{figure:samplingUncertainties}
\end{wrapfigure}
\myparagraph{Sampling Uncertainties in Hybrid System Monitoring} There is a methodological difficulty  already  in the high-level schematics in \cref{figure:hybridSysMonitoring}: 
\begin{quote}
 By looking only at a sampled log $\word$, how can a monitor conclude anything about the real behavior $\sigma$? 
\end{quote}
The same log $\word$ can result from different behaviors $\sigma$. \cref{figure:samplingUncertainties} shows 
an example, where we cannot decide if a safety property ``$x$ is always nonnegative'' is satisfied by $\sigma$. In other words, the way we interpolate  the log $\word$ and recover $\sigma$ is totally arbitrary. Thus, we cannot exclude potential violations of any safety specification unless the specification happens to talk only about values at sampling instants.

\begin{wrapfigure}{r}{0pt}
\includegraphics[bb=411 527 612 601,clip,width=.3\linewidth]{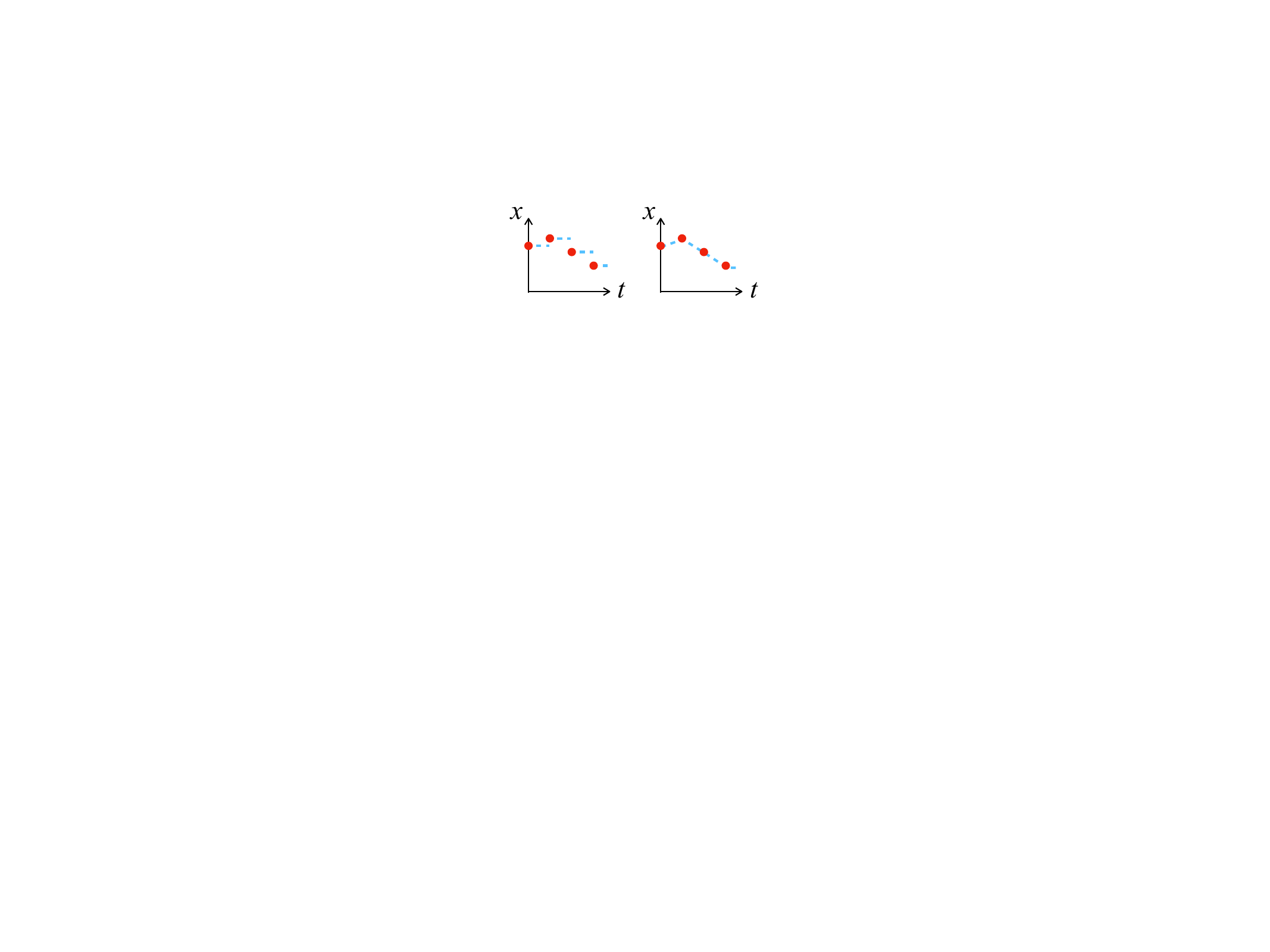}
\end{wrapfigure}
This issue of \emph{sampling uncertainties} is often ignored in the hybrid system monitoring literature. They typically employ heuristic interpolation methods, such as  piecewise-constant and piecewise-linear interpolation (above).
Use of these heuristic interpolation methods is often justified, typically when the sampling rate is large enough.
However, in \emph{networked monitoring} scenarios where a sensor and a monitor are separated by, \eg{} a wireless network, the sampling rate is small, and the interpolation of a log becomes a real issue. Network monitoring is increasingly common in IoT applications, and smaller sampling rates (\ie{} longer sampling intervals) are preferred for energy efficiency. %
\begin{example}
 [automotive platooning]
 \label{example:automotivePlatooningLog}
Consider a situation where two vehicles drive one after the other, with their distance kept small. Such \emph{automotive platooning} attracts interest as a measure for enhanced road capacity as well as for fuel efficiency (by reducing air resistance). 

Assume that the monitoring is conducted on a remote server.
Each vehicle intermittently sends its position to the server via the Internet.
Thus, the remote monitor only has a coarse-grained log.
Concretely, a log $\word$ is given in \cref{figure:automotivePlatooningLog}, by the position $\variablei{1}, \variablei{2}$ (meters) of each of the two vehicles, sampled at time $t=0, 10, 20$ (seconds). 
\begin{figure}[tbp]
\begin{subfigure}[c]{.49\linewidth}
\includegraphics[bb=127 335 474 418,clip,width=\linewidth]{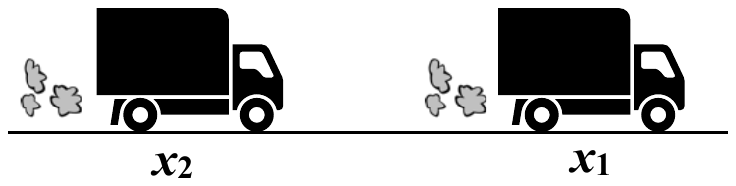}
\caption{Automotive platooning}
\label{figure:platooning}
\end{subfigure}
\begin{subfigure}[c]{.49\linewidth}
 	\centering
		\footnotesize
		
		\scalebox{1}{
		\begin{tikzpicture}[shorten >=1pt, scale=.4, yscale=.6, xscale=0.7, auto]
		\draw[->] (0, 0) --++ (0, 7.0);
		\draw[->] (0, 0) --++ (11.0, 0) node[anchor=north]{$t$};
                 \draw (0, 0) -- (0, -.2)node [below] (t1) {\scriptsize{$0$}};
                 \draw (5.0, 0) -- (5.0, -.2)node [below] (t2) {\scriptsize{$10$}};
                 \draw (10.0, 0) -- (10.0, -.2)node [below] (t3) {\scriptsize{$20$}};

                 \fill[color=red] (0, 1.5) coordinate (x1_1) circle[radius=7pt] node[color=black,left=5pt] {\scriptsize{40}};
                 \node[color=blue%
		 ] at (0, 0.75) (x2_1) {\pgfuseplotmark{triangle*}} node[color=black,left=5pt] {\scriptsize{35}};

                 \fill[color=red] (5.0, 4.5) coordinate (x1_2) circle[radius=7pt];
                 \draw[thick, dashed] (x1_2) -- (-.2, 4.5) node [left=0.5pt] {\scriptsize{$123$}};

                 \node[color=blue%
		 ] at (5.0, 4.0) (x2_2) {\pgfuseplotmark{triangle*}};
                 \draw[thick, dashed] (x2_2) -- (-.2, 4.0);
                 \node [below left=0.2pt] at (-.2, 4.2)  {\scriptsize{$117$}};

                 \fill[color=red] (10.0, 6.5) coordinate (x1_3) circle[radius=7pt];
                 \draw[thick, dashed] (x1_3) -- (-.2, 6.5) node [left=0.5pt] {\scriptsize{$203$}};

                 \node[color=blue%
		 ] at (10.0, 6.05) (x2_3) {\pgfuseplotmark{triangle*}};
                 \draw[thick, dashed] (x2_3) -- (-.2, 6.05);
                 \node [below left=0.2pt] at (-.2, 6.3)  {\scriptsize{$201$}};

                 \draw[thick, dashed] (x1_2) -- (t2);
                 \draw[thick, dashed] (x1_3) -- (t3);

                 \fill[color=red] (5.0, 4.5) coordinate circle[radius=7pt];
                 \fill[color=red] (10.0, 6.5) coordinate circle[radius=7pt];
                 \node[color=blue%
		 ] at (5.0, 4.0) {\pgfuseplotmark{triangle*}};
                 \node[color=blue%
		 ] at (10.0, 6.05) {\pgfuseplotmark{triangle*}};

		\end{tikzpicture}
		}
	 \caption{The log  $\word$. The red circles are $\variablei{1}$ and the blue triangles are
 $\variablei{2}$.}
	 \label{figure:automotivePlatooningLog}
\end{subfigure}
\caption{A leading example: automotive platooning}
\label{figure:leading_example}
\end{figure}

Let us now ask this question: \emph{have the two vehicles touched each other}? Physical contact of the vehicles is not observed in \cref{figure:automotivePlatooningLog}, but we cannot be sure what happened between the sampling instants. 
The piecewise-constant and piecewise-linear interpolation can only answer this question approximately. Moreover, such approximation is not of much help in the current example where sampling intervals are long.
\end{example}

\vspace*{.2em}
\myparagraph{Interpolation Assisted by System Knowledge} 
The following idea underpins the current work.
\begin{quote}
Prior knowledge about a system 
 is a powerful tool to bound sampling uncertainties.
\end{quote}
The latter means excluding some candidates when we recover a behavior $\sigma$ from a word $w$ by interpolation (cf.\ \cref{figure:samplingUncertainties}). 
For the log in \cref{figure:automotivePlatooningLog}, for example, we can say $\variablei{1}$ never reached $10^4$, knowing that the vehicle cannot accelerate that quickly. 

 Putting this idea to actual use requires a careful choice of a knowledge representation formalism. 
\begin{itemize}
 \item For one, it is desired to be \emph{expressive}. The above ``acceleration rate'' argument can be formulated in terms of  Lipschitz constants, but it is nice to also include mode switching---an important feature of hybrid systems. 
 \item For another, a formalism should be \emph{computationally tractable}. Monitoring is a practice-oriented method that often tries to process a large amount of data with limited computing resources (especially in embedded applications). Therefore, inference over knowledge represented in the chosen formalism should better be efficient. 
\end{itemize}
Note that these two concerns---expressivity and computational tractability---are in a trade-off.

\vspace*{.2em}
\noindent\myparagraph{Bounding Models Given by LHAs}
In this paper, we express such prior knowledge about a system using a \emph{linear hybrid automaton (LHA)}~\cite{HPR94}. 
This LHA
is called a \emph{bounding model}, and serves as an overapproximation of the target system.

    LHA is one of the well-known subclasses of hybrid automata (HA); an example is in \cref{figure:LHAsystemmodelExample}.  LHA's notable simplifying feature is that  flow dynamics is restricted to a conjunction of linear (in)equalities over the derivatives $\variabledoti{1},\variabledoti{2},\cdots,  \variabledoti{\VariablesCard}$. Its expressivity is limited---for example, a flow specification $\dot{\bf x}= {\bf A}{\bf x}+{\bf b}$ is not allowed since the variables $\bf x$  occur there. Differential inclusions are allowed, nevertheless (such as $\variabledoti{1} \in [7.5,8.5]$ and $\variabledoti{1}-\variabledoti{2}\le 1$); these are useful in expressing known safety envelopes, as  in \cref{figure:LHAsystemmodelExample}.
Most importantly, analysis of LHAs is tractable, with convex polyhedra providing an efficient means to study the reachability problem. 

\vspace*{.2em}
\myparagraph{Model-Bounded Hybrid System Monitoring} 
Our proposal is a  scheme that we call \emph{model-bounded  monitoring} of hybrid systems.
Its workflow is in \cref{figure:modelAwareMonitoring};
its features are as follows.

\ACMVersion{%
  \begin{figure*}[tbp]
  \begin{minipage}[c]{.65\linewidth}
}
\LongVersion{
\begin{figure}
}
 \centering
 \includegraphics[bb=48 179 788 421,clip,width=\LongVersion{\linewidth}\ACMVersion{.9\linewidth}]{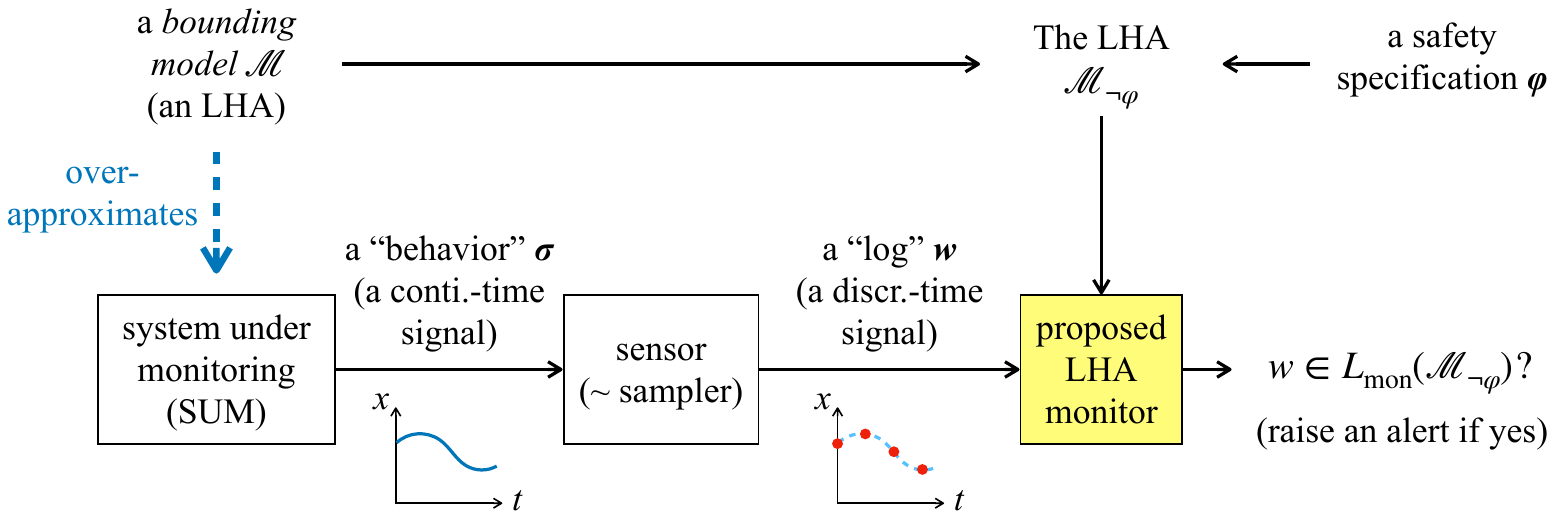}
 \caption{Model-bounded monitoring of hybrid systems}%
 \label{figure:modelAwareMonitoring}
\ACMVersion{
  \end{minipage}
  \hfill
  \begin{minipage}[c]{.33\linewidth}
}
\LongVersion{\end{figure}
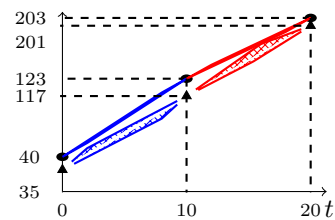
\begin{figure}
}
\centering
          \begin{tikzpicture}[shorten >=1pt, scale=.33, yscale=.7, auto]
		\draw[->] (0, 0) --++ (0, 10.5);
		\draw[->] (0, 0) --++ (10.75, 0) node[anchor=north]{$t$};
                 \draw (0, 0) -- (0, -.2)node [below] (t1) {\scriptsize{$0$}};
                 \draw (5.0, 0) -- (5.0, -.2)node [below] (t2) {\scriptsize{$10$}};
                 \draw (10.0, 0) -- (10.0, -.2)node [below] (t3) {\scriptsize{$20$}};

                 \fill[color=black] (0, 2.0) coordinate (x1_1) circle[radius=7pt] node[color=black,left=5pt] {\scriptsize{40}};
                 \node[color=black] at (0, 1.25) (x2_1) {\pgfuseplotmark{triangle*}} node[color=black,left=5pt] {\scriptsize{35}};

                 \fill[color=black] (5.0, 6.5) coordinate (x1_2) circle[radius=7pt];
                 \draw[thick, dashed] (x1_2) -- (-.2, 6.5) node [left=0.5pt] {\scriptsize{$123$}};

                 \draw[thick, blue, pattern=north west lines, pattern color=blue] (x1_1) -- (1.0, 2.8) -- (x1_2) -- (4.0, 5.7) -- (x1_1);

                 \node[color=black] at (5.0, 5.5) (x2_2) {\pgfuseplotmark{triangle*}};
                 \draw[thick, dashed] (x2_2) -- (-.2, 5.5) node [left=0.5pt] {\scriptsize{$117$}};

                 \draw[thick, blue, pattern=crosshatch, pattern color=blue] (x2_1) -- (1.0, 2.5) -- (x2_2) -- (4.0, 4.2) -- (x2_1);

                 \fill[color=black] (10.0, 10.0) coordinate (x1_3) circle[radius=7pt];
                 \draw[thick, dashed] (x1_3) -- (-.2, 10.0) node [left=0.5pt] {\scriptsize{$203$}};
                 \draw[thick, red, pattern=north west lines, pattern color=red] (x1_2) -- (6.0, 7.3) -- (x1_3) -- (9.0, 9.2) -- (x1_2);

                 \node[color=black] at (10.0, 9.55) (x2_3) {\pgfuseplotmark{triangle*}};
                 \draw[thick, dashed] (x2_3) -- (-.2, 9.55) node [below left=0.5pt] {\scriptsize{$201$}};
                 \draw[thick, red, pattern=crosshatch, pattern color=red] (x2_2) -- (8.0, 8.5) -- (x2_3) -- (7.7, 7.5) -- (x2_2);

                 \draw[thick, dashed] (x1_2) -- (t2);
                 \draw[thick, dashed] (x1_3) -- (t3);

         \end{tikzpicture}
	\caption{Model-bounded monitoring of the log $w$ in \cref{figure:automotivePlatooningLog}. The bounding model $\A$ in \cref{figure:LHAsystemmodelExample} confines  interpolation  to the hatched area. Thus no collision in $t\in [0,10]$; potential collision in $t\in [10, 20]$. 
   }
 \label{figure:examples:acc:signal}
\ACMVersion{\end{minipage}}
\LongVersion{
\end{figure}
\begin{figure*}
}

\centering
\begin{subfigure}[c]{.55\linewidth}
\centering
  \scalebox{.75}{
  \begin{tikzpicture}[shorten >=1pt, scale=2, yscale=1, auto] %
		\node[location] (l0) at (0,0) [align=center]{$\loc_0$\\ $\variabledoti{1} \in [7.5,8.5]$\\ $\variabledoti{2} \in [8.0,9.0]$};
                 \draw[<-] (l0) -- node[above, align=center] {$\variablei{1} = 40$ \\ $\variablei{2} = 35$} ++(-1.5cm,0);
		\node[location] (l1) at (2.25,0) [align=center]{$\loc_1$\\ $\variabledoti{1} \in [11.0,13.0]$\\ $\variabledoti{2} \in [9.0,11.0]$};
		\path[->] 
		(l0) edge [bend left=10] node {$\variablei{1} - \variablei{2} \leq 4$} (l1)
		(l1) edge [bend left=10] node {$\variablei{1} - \variablei{2} \geq 4$} (l0)
	;
  \end{tikzpicture}
  }
 \caption{A bounding model $\A$ for the  platooning example, expressed as an LHA}
 \label{figure:LHAsystemmodelExample}
\end{subfigure}
\hfill
\begin{subfigure}[c]{.4\linewidth}
\centering
 \scalebox{.6}{
 \begin{tikzpicture}[shorten >=1pt, scale=2, yscale=1.3, auto] %
  \node[location] (l0) at (0,0) [align=center]{$\loc_0$\\ $\variabledoti{1} \in [7.5,8.5]$\\ $\variabledoti{2} \in [8.0,9.0]$};
  \draw[<-] (l0) -- node[above, align=center] {$\variablei{1} = 40$ \\ $\variablei{2} = 35$} ++(-1.5cm,0);
  \node[location] (l1) at (2.5,0) [align=center]{$\loc_1$\\ $\variabledoti{1} \in [11.0,13.0]$\\ $\variabledoti{2} \in [9.0,11.0]$};
  \node[location,accepting] (l0bad) at (0,-1.40) [align=center] {$\loc_2$\\ $\variabledoti{1} \in [7.5,8.5]$\\ $\variabledoti{2} \in [8.0,9.0]$};
  \node[location,accepting] (l1bad) at (2.5,-1.40) [align=center] {$\loc_3$\\ $\variabledoti{1} \in [11.0,13.0]$\\ $\variabledoti{2} \in [9.0,11.0]$};

  \path[->] 
  (l0) edge [bend left=10] node {$\variablei{1} - \variablei{2} \leq 4$} (l1)
  (l1) edge [bend left=10] node {$\variablei{1} - \variablei{2} \geq 4$} (l0)
  (l0bad) edge [bend left=10] node {$\variablei{1} - \variablei{2} \leq 4$} (l1bad)
  (l1bad) edge [bend left=10] node {$\variablei{1} - \variablei{2} \geq 4$} (l0bad)
  (l0) edge [] node {$\variablei{1} - \variablei{2} \leq 0$} (l0bad)
  (l1) edge [] node {$\variablei{1} - \variablei{2} \leq 0$} (l1bad)
  ;
 \end{tikzpicture}}
 \caption{The LHA $\A_{\neg\varphi}$ for 
 $\varphi=(x_{1}-x_{2}>0)$}
 \label{figure:collision:product}
\end{subfigure}
\caption{LHAs for the automotive platooning example}
\label{figure:LHAExamplesOriginalAndProduct}
\end{figure*}
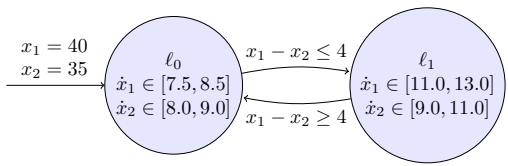
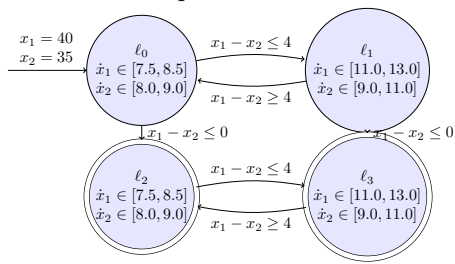

\begin{enumerate}
 \item We use our prior knowledge about the SUM in order to reduce sampling uncertainties. The knowledge is expressed by an LHA;\@ it is called a \emph{bounding model}.
 \item 
We restrict to a safety specification $\varphi$ given by a conjunction of linear (in)equalities.%
\footnote{This restriction is for the ease of presentation.  Extension to LTL specifications should not be hard: an LTL formula can be translated to an automaton; and it can then be combined with a bounding model $\mathcal{M}$. This is future work.}
 We interpret $\varphi$ globally (``$\sigma(t)$  satisfies $\varphi$ at any time $t$'').  
We obtain an LHA $\mathcal{M}_{\lnot \varphi}$ by taking the synchronized product of the bounding model $\mathcal{M}$ and the automaton to monitor the violation of $\varphi$. 
 \item We introduce the notion of \emph{monitored language} $\Lmon$ of an LHA. Roughly speaking, it is  the set of ``logs which have a corresponding signal accepted by the LHA.''
The notion differs from known  language notions for LHA, in that mode switches in an LHA need not be visible in a log (modes may  change between sampling instants). 
 \item\label{item:featureCorrectness} We show the following  meta-level correctness result: $w\in \Lmon(\mathcal{M}_{\lnot\varphi})$ if and only if there exists a continuous-time signal $\sigma$ such that
       \begin{enumerate}
	\item\label{item:featureCorrectness:a} $\sigma$ induces $\word$ by sampling, 
	\item\label{item:featureCorrectness:b} $\sigma$ conforms with the bounding model $\mathcal{M}$, and
	\item\label{item:featureCorrectness:c} $\sigma$ violates the safety specification $\varphi$. 
       \end{enumerate}
\end{enumerate}
\noindent Our main technical contribution consists of 
\begin{ienumeration}
\item
 the introduction of the new language notion $\Lmon$, 
\item
 the use of $\Lmon$ in the proposed model-bounded monitoring scheme, and
\item
(partial) algorithms that solve    $\Lmon$ membership. 
\end{ienumeration}
Used in the scheme in \cref{figure:modelAwareMonitoring}, these algorithms check if the given log $\word$ belongs to $\Lmon(\mathcal{M}_{\lnot \varphi})$, whose answer is then used for the safety analysis of the (unknown) actual behavior $\sigma$. The last point is discussed  in the next paragraph about usage scenarios.

 We present two (partial) algorithms: one reduces the $\Lmon$ membership problem to the reachability problem of LHAs, translating a log $\word$ into an LHA. The other is a direct algorithm that relies on polyhedra computation. These algorithms  are necessarily partial since  $\Lmon$ membership  is undecidable (\cref{theorem:undecidability}). However, their positive and negative answers are guaranteed to be correct. Moreover, we observe that the latter direct algorithm terminates in most benchmarks, especially when a bounding model's dimensionality is not too large. 

\begin{example}\label{example:automotivePlatooningMLnotVarphi}
We continue \cref{example:automotivePlatooningLog}. %
For the log $w$ in \cref{figure:automotivePlatooningLog}, the bounding model $\A$ in \cref{figure:LHAsystemmodelExample} confines potential interpolation between the samples to the hatched areas in \cref{figure:examples:acc:signal}. The two areas are separate in $t\in[0,10]$, which means the two cars were safe in the period. For $t\in[10,20]$, the two areas overlap, suggesting potential collision. 

The above analysis is automated by our automata-theoretic framework in 
 \cref{figure:modelAwareMonitoring}. We shall sketch its workflow.

Let $\varphi$ be the safety specification  $x_{1}-x_{2}>0$ (``no physical contact'')\footnote{For simplicity, we modeled the cars as points. It is straightforward to use a more realistic model, \eg{} a car model by a rectangle.}. The formal construction of $\A_{\lnot\varphi}$ (\cref{def:Mlnotvarphi}) yields the LHA in \cref{figure:collision:product}. In $\A_{\lnot\varphi}$,  the original LHA $\A$ (\cref{figure:LHAsystemmodelExample}) is duplicated, and  once $\varphi$ is violated, the execution can move from the first copy (the top two states in \cref{figure:collision:product}) to the second (the bottom states). The bottom states are accepting---they detect violation of $\varphi$. 

Now we use one of our algorithms to solve the membership problem, \ie{} if the log $w$ belongs to $\Lmon(\A_{\lnot\varphi})$. Solving this membership problem amounts to computing the hatched areas in \cref{figure:examples:acc:signal}---it is done relying on polyhedra computation---and checking if the specification $\varphi$ is violated.

\end{example}

\vspace*{.2em}
\myparagraph{Usage Scenarios}
The 
 scheme in \cref{figure:modelAwareMonitoring} is used as follows. As a basic prerequisite, we assume that the bounding model $\mathcal{M}$ \emph{overapproximates} the  SUM: for each continuous-time signal $\sigma$, 
\begin{quote}
 \textbf{(soundness of a bounding model)}\; 

 $\sigma$ is a behavior of the SUM $\Longrightarrow$  $\sigma$ is a run of~$\mathcal{M}$.
\end{quote}
 We do not require the other implication. Due to the limited expressivity of LHAs (that is the price for computational tractability), $\mathcal{M}$ would not tightly describe  the SUM.\@

 Assume first that our monitor did \emph{not} raise an alert (\ie{} $\word\not\in \Lmon(\A_{\lnot\varphi})$). Let  $\sigma_{0}$ be the (unknown) actual behavior of the SUM that is behind the log $\word$.
By the feature~\ref{item:featureCorrectness} of the scheme, we conclude that  $\sigma_{0}$ was safe.
Indeed,  $\sigma_{0}$  satisfies \cref{item:featureCorrectness:a} by definition.
It comes from the SUM, and thus by the  soundness assumption, $\sigma_{0}$ satisfies \cref{item:featureCorrectness:b}.
Hence \cref{item:featureCorrectness:c} must fail. 

Let us turn to the case where our monitor \emph{did} raise an alert ($w\in \Lmon(\mathcal{M}_{\lnot\varphi})$). This can be a false alarm. For one, the existence of unsafe $\sigma$ (as in the feature~\ref{item:featureCorrectness}) does not imply that the actual behavior $\sigma_{0}$ was unsafe.
For another, \cref{item:featureCorrectness:b} does not guarantee that $\sigma$ is indeed a possible behavior of the SUM, since we only assume \emph{soundness} of the bounding model. Nevertheless, a positive answer of our monitor comes with a reachability witness (a trace) in $\mathcal{M}_{\lnot\varphi}$, which serves as a useful clue for further examination. 

Summarizing, our monitor's alert can be false, while the absence of an alert proves safety.  We can thus say our model-bounded monitoring scheme is \emph{sound}.

\vspace*{.2em}
\myparagraph{Bounding Models}
We note that the roles of bounding models are different from  common roles played by  \emph{system models}. A system model aims to describe the system's behaviors in a \emph{sound} and \emph{complete} manner. In contrast, bounding models focus on overapproximation, trading completeness for computational tractability that is needed in monitoring applications. 
 The \emph{overapproximating}  nature of a bounding model  is less of a problem in monitoring, compared to other exhaustive applications such as model checking. In the latter, approximation errors accumulate over time, leading to increasingly loose overapproximation. In contrast, in our usage, a bounding model is used to interpolate between samples (\cref{figure:examples:acc:signal}). Here overapproximation errors get reset to zero by new samples.
This observation also leads that even if the bounding model is loose, it can be useful if the sampling interval is short.

Because we assume that we have an overapproximation of the actual system, we can indeed formally verify the safety of the system by the reachability analysis.
However, due to the overapproximation, the given LHA usually contain unsafe behaviors, \ie{} the unsafe locations are reachable in the LHA.\@ Most of the benchmarks in our experiments are certainly unsafe.
In contrast, in model-bounded monitoring, even if the unsafe locations are reachable in the LHA, the monitored behavior can be safe, \ie{} the unsafe locations are unreachable from the current sample. Therefore, monitoring is still useful even if we have a model overapproximating the actual system.

Bounding models can arise in different ways, including:
\begin{itemize}
 \item \textbf{(Adding margins to a system model)}
       If  a system model is given as an LHA,  we can use it as a bounding model. A more realistic scenario is to add some margins to address potential perception and actuation errors. LHAs' feature that they allow differential \emph{inclusions} is particularly useful here. An example is in \cref{figure:addMargins}, where perception and actuation  uncertainties are addressed by the additional margins in the transition guards and flow dynamics, respectively.
 \item \textbf{(LHA approximation of a system model)}
       LHA is one of the subclasses of HA for which reachability is attackable (it is hopeless for general HA). Consequently,  tools have been proposed for analyzing LHA, including \phaverlite{}~\cite{BZ19} and its predecessor \phaver{}~\cite{Frehse2008}.
       Moreover, for their application, overapproximation of other dynamics by LHAs has been studied and tool-supported.
       See \eg{}~\cite[Section~3.2]{Frehse2008}. These  techniques can be used to obtain an LHA bounding model  from a more complex model. 

	\item \textbf{(From a third-party vendor)}
	HA is a well-accepted formalism in academia and industry.  It is conceivable that a system vendor provides an LHA as the system's ``safety specification''.
	It serves as a bounding model. 
\end{itemize}

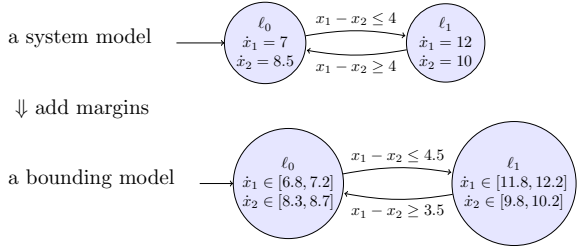
\begin{figure}[tbp]
\centering
               \scalebox{.8}{a system model}
              \raisebox{-.5\height}{		
               \scalebox{.6}{
		\begin{tikzpicture}[shorten >=1pt, scale=2, yscale=1, auto] %
		\node[location] (l0) at (0,0) [align=center]{$\loc_0$\\ $\variabledoti{1} =7$\\ $\variabledoti{2} =8.5$};
                 \draw[<-] (l0) -- node[above, align=center] {}
                 ++(-1cm,0);
		\node[location] (l1) at (2,0) [align=center]{$\loc_1$\\ $\variabledoti{1} =12$\\ $\variabledoti{2} =10$};
		\path[->] 
		(l0) edge [bend left=10] node {$\variablei{1} - \variablei{2} \leq 4$} (l1)
		(l1) edge [bend left=10] node {$\variablei{1} - \variablei{2} \geq 4$} (l0)
	;
		\end{tikzpicture}
		}}

        \scalebox{.8}{ $\Downarrow$ add margins\phantom{add margins}\phantom{add margins}\phantom{add margins}}

               \scalebox{.8}{a bounding model}
        \vspace{.3em}
              \raisebox{-.5\height}{		
               \scalebox{.6}{
		\begin{tikzpicture}[shorten >=1pt, scale=2, yscale=1, auto] %
		\node[location] (l0) at (0,0) [align=center]{$\loc_0$\\ $\variabledoti{1} \in [6.8, 7.2]$\\ $\variabledoti{2} \in[8.3, 8.7]$};
                 \draw[<-] (l0) -- node[above, align=center] {}
                 ++(-1cm,0);
		\node[location] (l1) at (2.5,0) [align=center]{$\loc_1$\\ $\variabledoti{1} \in[11.8,12.2]$\\ $\variabledoti{2} \in[9.8, 10.2]$};
		\path[->] 
		(l0) edge [bend left=10] node {$\variablei{1} - \variablei{2} \leq 4.5$} (l1)
		(l1) edge [bend left=10] node {$\variablei{1} - \variablei{2} \geq 3.5$} (l0)
	;
		\end{tikzpicture}
		}}	

        \caption{\emph{Adding margins} to obtain bounding models. The top model gets loosened by perception uncertainties (margin $0.5$) and actuation uncertainties (margin 0.2)}
	\label{figure:addMargins}
\end{figure}%

\vspace*{.2em}
\myparagraph{Contributions} We summarize our main contributions.
\begin{itemize}
 \item We tackle the issue of sampling uncertainties in hybrid system monitoring, proposing the  \emph{model-bounded monitoring} scheme (\cref{figure:modelAwareMonitoring}) as a countermeasure. The scheme uses LHAs as \emph{bounding models}. 
 \item We introduce the novel technical notion of \emph{monitored language} $\Lmon$ for LHAs. In $\Lmon$, unlike in other language notions, input words and mode switches do not necessarily synchronize.  We show that  $\Lmon$ membership is undecidable, yet we introduce two partial algorithms. 
 \item We establish soundness of our model-bounded monitoring scheme: absence of an alert guarantees that every possible behavior $\sigma$ behind the log $w$ is safe.
 \item The practical relevance and algorithmic scalability is demonstrated by experiments, using benchmarks that are mainly taken from automotive platooning scenarios. 
\end{itemize}

\vspace*{.2em}
\myparagraph{Related Work}\label{paragraph:related}
In 
the 
IoT
 applications~\cite{GBMP13},  energy efficiency is of paramount importance. Energy efficiency demands longer sampling and communication intervals; the current work presents an automatic and sound method to  mitigate the  uncertainties that result from those longer intervals. 

In the context of quality assurance of CPSs, monitoring of \emph{digital} (\ie{} discrete-valued) or \emph{analog} (\ie{} continuous-valued) signals takes an important role.
There have been many works on signal monitoring using various logic, \eg{} \emph{signal temporal logic (STL)}~\cite{MN04,FP09}, \emph{timed regular expressions (TREs)}~\cite{UFAM14}, \emph{timed automata}~\cite{BFNMA18}, or \emph{timed symbolic weighted automata (TSWAs)}~\cite{Waga19}.
However, in most of the existing works, interpolation of the sampled signals is limited to only piecewise-constant or piecewise-linear.

There are a few works on monitoring utilizing  system models.
In~\cite{ZLD12}, a set of predictive words are generated through a static analysis of the monitored program and monitored against LTL (linear temporal logic).
In~\cite{PJTFMP17}, the system model and the monitored property are given as timed automata to construct a monitor predicting the satisfaction (or violation) of the monitored property.
In~\cite{BGF18}, the stochastic system model is trained as a hidden Markov model and utilized for monitoring against a deterministic finite automaton.\@
In~\cite{QD20}, statistical models on the monitored signals are utilized to predict the future signals for robust monitoring against~STL.

Overall, prediction (\ie{} \emph{extrapolation}) of the future behaviors is the main purpose of the existing model-based monitoring works~\cite{ZLD12,PJTFMP17,BGF18,QD20} to the best of our knowledge.
Our approach utilizes system knowledge for \emph{interpolation} of the infrequently sampled signals. %
One of the exceptions is monitoring of timed distributed traces using symbolic overapproximation by an SMT solver, where the system model is given as a rectangular hybrid automaton~\cite{DJZM12}.
Although their problem setting is close to ours, their SMT-based algorithm is not suitable to online monitoring. 
In contrast, our automata-based approach makes an application to online monitoring straightforward, especially for the procedure in \cref{section:algorithm2}.
Another exception is an algorithm to give a bound of continuous-time robust semantics for metric temporal logic formulas from discrete-time samples~\cite[Section 4]{FP09}.
Although their problem setting is also close to ours, they use only one function to bound the signal behavior while we use LHAs.
Thanks to the expressivity of LHAs, our approach may bound the signal behavior more finely, which may cause less false alarms.
In monitoring, a system model can also be used to monitor the compliance of the monitored behavior with the system model to validate the verification result obtained from the %
model.
In~\cite{MP16,MP18}, a monitor is constructed from a system model in differential dynamic logic~\cite{DBLP:conf/lics/Platzer12a}.
The largest difference between their method and ours is the role of the system model:
they do not assume the compliance of the actual behavior with the system model but they check the compliance at runtime;
we assume the model compliance and utilize it for interpolation of the infrequently sampled signals.

There are also several papers that use an approximate assumption of the actual system behavior in the safety assurance of CPS.\@
For example, \cite{PC07} utilizes Lipschitz constants to allow a numeric method for model checking.

There are existing language notions for LHAs~\cite{AKV98}.
These are different from the notion $\Lmon$ that we introduce; hence the results in~\cite{AKV98} do not subsume ours. The key difference is if the input word and mode switches must synchronize; see \cref{example:timed_quantitative_word} and the preceding discussions.
A recent line of work is that of timed pattern matching~\cite{UFAM14,WHS17,BFNMA18,UM18,AHW18,WA19,WAH19,Waga19}, that takes as input a log and a specification, and decides \emph{where} in the log the specification is satisfied or violated.
Through the construction of the \emph{matching automata}~\cite{BFNMA18,Waga19}, our monitoring problem can also decide where in the log the specification is satisfied or violated.
Thus, our work can also be seen as an extension of timed pattern matching concerning the system models.

\vspace*{.2em}
\myparagraph{About this manuscript}
This manuscript is an extension of the article with the same title published in the proceedings of ICCPS 2021~\cite{WAH21}.
We significantly improved the content, by notably
	discussing model-bounded monitoring with an imprecise model,
        adding missing proofs and a new lemma,
        extending the related work and perspectives,
	adding new examples and figures for illustration, and
	adding new benchmarks \UnsafeNAV{} and \SharedGasBurner{} to evaluate the precision.

\vspace*{.2em}
\myparagraph{Outline}
We recall LHAs in \cref{section:preliminaries}.
After we introduce monitored languages $\Lmon$ for LHAs in \cref{section:monitoredLang}, model-bounded monitoring is formalized in \cref{section:monitoring} and we prove its correctness.
We show that $\Lmon$ membership is undecidable in \cref{sec:symbInt}.
We present two partial algorithms:
\begin{ienumeration}%
	\item the one in \cref{section:algorithm1} 
     relies on an existing model checker \phaverlite{} via suitable translation and
	\item the one in \cref{section:algorithm2} is a dedicated algorithm.
\end{ienumeration}%
We perform extensive experiments in \cref{section:experiments} and conclude in \cref{section:conclusion,section:perspectives}. 

\section{Preliminaries: Linear Hybrid Automata}\label{section:preliminaries}
Let $\interval(\setQ)$ be the set of closed intervals on~$\setQ$, \ie{} of the form $[a,b]$, where $a,b\in \setQ$ and $a \leq b$.
For a partial function $f\colon X \partfun Y$,  the domain $\{x \in X \mid f(x) \text{ is defined}\}$ is denoted by~$\domain{f}$.
We fix a set~$\Variables = \{ \variablei{1}, \dots, \variablei{\VariablesCard} \} $ of real-valued \emph{variables}.
A \emph{(variable) }valuation is a function
$\varval : \Variables \rightarrow \R$.
When $\Variables$ is clear from the context,  a valuation $\varval$ is expressed by the tuple $(\varval(\variablei{1}),\varval(\variablei{2}),\dots,\varval(\variablei{\VariablesCard}))$.
Given $\updates\colon \Variables \partfun \interval(\setQ)$, we define the \emph{update} of a valuation~$\varval$, written $\reset{\varval}{\updates}$, as follows: $\update{\varval}{\updates}(\variable) \in \updates(\variable)$ if $\variable \in \domain{\updates}$, and $\update{\varval}{\updates}(\variable)=\varval(\variable)$ otherwise.

We assume ${\compOp} \in \{\leq, =, \geq\}$.
Let $\GuardP$ be the set of \emph{linear systems} over $\Variables$ defined by a finite conjunction of inequalities of the form $a_1\variablei{1} + a_2\variablei{2} + \dots + a_{\VariablesCard}\variablei{\VariablesCard} \compOp d$, with $d, a_1, a_2, \dots, a_{\VariablesCard}\in \setZ$. 
We let $\top = \bigwedge \emptyset$, and let $\bot$ be the contradiction.
The set $\GuardD$ is defined similarly; it consists of constraints over derivatives $\variabledoti{1},\dotsc,\variabledoti{\VariablesCard}$. 
\subsection{Syntax}
\begin{definition}
 [linear hybrid automata (LHA)~\cite{HPR94}]
 \label{def:LHA}
 An \emph{LHA} is a tuple
 $\A = (\Loc, \LocFinal, \Variables, \VarValInit, \Flow, \invariant, \Edges)$, where:
 \begin{enumerate}
  \item $\Loc$ is a finite set of locations,
  \item $\LocFinal \subseteq \Loc$ is the set of accepting locations,
  \item $\Variables$ is a finite set of variables,
  \item $\VarValInit\colon \Loc \to \GuardP$ is the initial variable valuation for each location, %
  \item $\Flow\colon \Loc \to \GuardD$ is the \emph{flow}, assigning to each $\loc \in \Loc$ the set of the derivatives (``rates'') %
        $\{(\variabledoti{1},\variabledoti{2},\cdots,  \variabledoti{\VariablesCard}) \mid (\variabledoti{1},\variabledoti{2},\cdots, \variabledoti{\VariablesCard}) \models \Flow(\loc)\}$,
  \item $\invariant \colon \Loc \to \GuardP$ is the \emph{invariant} for each location,
  \item $\Edges$ is a finite set of \emph{edges} $\edge = (\loc, \guard, \updates, \loc')$
        where
        \begin{ienumeration}
         \item $\loc, \loc'\in \Loc$ are the source and target locations,
         \item $\guard \in \GuardP$ is the guard, and
         \item $\updates\colon \Variables \partfun \interval(\setQ)$ is the update function.
        \end{ienumeration}
 \end{enumerate}
\end{definition}

\noindent
Note that \cref{def:LHA} allows for non-deterministic initial locations.  A location~$\loc$ that cannot be initial is  such that $\VarValInit(\loc) = \bot$. %

LHAs can be \emph{composed} using synchronized product (see \eg{}~\cite[Definition~4]{Raskin05}) in a way similar to finite-state automata.
The synchronized product of two LHAs is known to be an LHA~\cite{HPR94}.
Of importance is that, in a composed location, the global flow constraint is the \emph{intersection} of the local component flow constraints.

\begin{example}
	Consider the LHA in \cref{figure:LHAsystemmodelExample}, where $\VarValInit$ is such that $\VarValInit(\loc_0) = \{ \variablei{1} = 40 \land \variablei{2} = 35 \}$ and $\VarValInit(\loc_1) = \bot$.
	This LHA, giving a bounding model for an automotive platooning system (\cref{example:automotivePlatooningLog}), contains 2 locations and 2 variables $\Variables = \{ \variablei{1} , \variablei{2} \}$.
	This LHA features no invariant (\ie{} all invariants are $\top$).
	Note that this LHA fits into a subclass in which the derivatives for the flows are all in bounded, constant intervals.

       In this LHA (\cref{figure:LHAsystemmodelExample}), 
		$\variablei{1}$ (resp.~$\variablei{2}$) denotes the position 
of Vehicle~1 (resp.~2), initially 40 and~35 respectively.
		In~$\loc_0$, both vehicles run roughly at the same speed, although Vehicle~2 can be slightly faster (\eg{} due to smaller air resistance, as it follows Vehicle~1).
		When the distance between both vehicles becomes less than~4, they enter mode~$\loc_1$, where Vehicle~1 drives faster than in $\loc_0$.

In the LHA $\A_{\lnot\varphi}$ in \cref{figure:collision:product}, the vertical edges are enabled once the specification $x_{1}-x_{2}>0$ is violated, that is, once the two vehicles touch each other. 
\end{example}
\begin{example}
	Consider the LHA in \cref{figure:translation:word}, where
	\begin{ienumeration}
		\item $\Variables = \{ \variablei{1} , \variablei{2}, \clockabs, \clockrel \}$,
		\item $\VarValInit$ is such that $\VarValInit(\loc_0) = \{ \variablei{1} = 40 \land \variablei{2} = 35 \land \clockabs = 0 \land \clockrel = 0 \}$ and $\VarValInit(\wloc_i) = \bot$ for $1 \leq i \leq 3$, and
		\item $\variabledoti{1} = \variabledoti{2} = \dot{\clockabs} = \dot{\clockrel} = 1$ in all locations (not depicted in \cref{figure:translation:word}).
	\end{ienumeration}%
	We depict invariants using a box under the location.
\end{example}

\LongVersionActuallyNot{
\subsection{Semantics}

}

We recall the standard  semantics of LHAs, called concrete semantics. It is formulated as a timed transition system~\cite{HMP91}. 

\begin{definition}[concrete semantics of an LHA]\label{def:concreteSem}
 Given an LHA $\A = (\Loc, \LocFinal, \Variables, \VarValInit, \Flow, \invariant, \Edges)$, %
 the \emph{concrete semantics} of $\A$ is given by the timed transition system (TTS) $(\States, \Sinit, \flecheRel)$, with
 \begin{itemize}
  \item $\States = \big\{ (\loc, \varval) \in \Loc \times \R^\VariablesCard \mid \varval \models \invariant(\loc) \big\}$,%
  \item $\Sinit = \big\{(\loc,\varval) \mid \loc \in \Loc, \varval \in \VarValInit(\loc) \big\} \cap \States$,
  \item  $\flecheRel$ consists of the discrete and continuous transition relations:
         \begin{enumerate}
          \item discrete transitions: $(\loc, \varval) \longueflecheRel{\edge} (\loc', \varval')$, %
                if
                there exists
                $\edge = (\loc, \guard, \updates, \loc') \in \Edges$
                such that
                $\varval \models \guard$,
                $\varval' \in \update{\varval}{\updates}$.
          \item continuous transitions: $(\loc, \varval) \longueflecheRel{d,\flow} (\loc, \varval')$, with
                the delay $d \in \Rgeqzero$ and
                the flow $\flow\colon \Variables \to \R$ satisfying, 
                $\flow \models \Flow(\loc)$,
                $\forall d' \in [0, d], (\loc, \varval+d' \flow ) \in \States$, and
                $\varval' = \varval+d \flow$, where 
                $\varval+d' \flow$ is the valuation satisfying $\valuate{\variable}{(\varval+d' \flow)} = \valuate{\variable}{\varval} +d' \valuate{\variable}{\flow}$ for any $\variable \in \Variables$.
         \end{enumerate}
 \end{itemize}
\end{definition}

\begin{definition}[(accepting) run]\label{def:runAndAcceptingRun}
 Given an LHA~$\A$ with concrete semantics $(\States, \Sinit, \flecheRel)$, we refer to the states of~$\States$ as the \emph{concrete states} of~$\A$.
 A \emph{run} of~$\A$ is an alternating sequence $\run = s_0, \flecheRel_1, s_1, \dots, \flecheRel_n, s_n$ of concrete states $s_i \in \States$ and transitions $\flecheRel_i \in \flecheRel$ satisfying  $\sinit \in \Sinit$ and 
 $s_0 \to_1 s_1 \dots \to_n s_n$.
 For a run $\run$, the \emph{duration} $\Duration(\run) \in \Rnn$ is the sum of the delays in $\run$.
 We denote the $i$-th prefix $s_0 \to_1 s_1 \dots \to_i s_i$ of $\run$ by $\run[i]$.
 A run is \emph{accepting} if its last state $(\loc, \varval)$ satisfies $\loc \in \LocFinal$.
\end{definition}

\begin{example}
 \label{example:run}
 Let $\A$ be the LHA in \cref{figure:LHAsystemmodelExample}.
 The sequence 
 \begin{math}
  \run = \bigl(\loc_0, \varval_0\bigr) 
  \longueflecheRel{10, (8.3,8.2)} \bigl(\loc_0, \varval_1\bigr)
  \longueflecheRel{\frac{4}{3}, (7.5,9)} \bigl(\loc_0, \varval_2\bigr)
  \longueflecheRel{\edge_1} \bigl(\loc_1, \varval_2\bigr)
  \longueflecheRel{\frac{2}{3}, (12,9)} \bigl(\loc_1, \varval_3\bigr)
  \longueflecheRel{\edge_2} \bigl(\loc_0, \varval_3\bigr)
  \longueflecheRel{8, (7.75,8.25)} \bigl(\loc_0, \varval_4\bigr)
 \end{math}
 is a run of $\A$, where
 $\varval_0 = (40,35)$,
 $\varval_1 = (123, 117)$,
 $\varval_2 = (133, 129)$,
 $\varval_3 = (141, 135)$,
 $\varval_4 = (203, 201)$, and
 $\edge_1$ and $\edge_2$ are the edges from $\loc_0$ and $\loc_1$, respectively.
\end{example}
\section{Monitored Languages of LHAs}\label{section:monitoredLang}
We introduce another semantics for LHAs besides concrete semantics (\cref{def:concreteSem}); it is called the \emph{monitored language}. The two semantics are used in \cref{figure:modelAwareMonitoring} in the following way: 
\begin{ienumeration}
 \item concrete semantics is (roughly) about whether a continuous-time signal $\signal$ (``behavior'') conforms with the LHA $\A$ and
 \item the monitored language $\Lmon(\A)$ is about whether  a discrete-time signal $\word$ (``log'') conforms with $\A$.
\end{ienumeration}

\begin{definition}
 [timed quantitative words]%
 \label{def:timed_quantitative_word}
 A \emph{timed quantitative word} $\word$ is a sequence $(\wordVal_1, \tau_1), (\wordVal_2, \tau_2), \dots,(\wordVal_m, \tau_m)$ of pairs $(\wordVal_i, \tau_i)$ of a valuation $\wordVal_i\colon \Variables \to \R$ and a timestamp $\tau_i \in \Rnn$ satisfying $\tau_i \leq \tau_{i+1}$ for each $i \in \{1,2,\dots,m-1\}$.

 For a timed quantitative word $\word = (\wordVal_1, \tau_1),\dots, (\wordVal_m, \tau_m)$, we let $|\word| = m$ and 
 for any $i \in \{1,\dots,m\}$, we let $\word[i] = (\wordVal_1, \tau_1),\dots, (\wordVal_i, \tau_i)$.
\end{definition}
\noindent We sometimes refer to pairs $(\wordVal_i, \tau_i)$ as \emph{samples}---these are the  red dots in \cref{figure:modelAwareMonitoring}. %

\begin{definition}[monitored language $\Lmon(\A)$]%
\label{def:monitoredLang}
Let $\run = s_0 \to_1 s_1 \to_2 \dots \to_n s_n$ be a run of an LHA $\A$ (\cref{def:concreteSem}), and
 $\word = (\wordVal_1,\tau_1),(\wordVal_2,\tau_2), \dots, (\wordVal_m,\tau_m) $ be a timed quantitative word.
We say $\word$
is \emph{associated} with $\run$
 if, for each $j \in \{1,2,\dots,m\}$, we have either of the following two. Here $\loc_{i}, \varval_{i}$ are so that  $s_{i} = (\loc_{i}, \varval_{i})$ for each $i \in\{0,1,\dots,n\}$. 
 \begin{enumerate}
 \item\label{def:monitoredLang:1} There exists $i \in \{0,1,2,\dots, n\}$ such that $\Duration(\run[i]) = \tau_j$ and $\wordVal_j = \varval_{i}$; or
 \item There exists $i \in \{0,1,2,\dots, n-1\}$ such that $\Duration(\run[i]) < \tau_j < \Duration(\run[i+1])$ and for any $\variable \in \Variables$, $\wordVal_j(\variable) = \varval_{i}(\variable) + (\tau_j - \Duration(\run[i]))f_{i}(\dot{\variable})$ holds, where ${\flecheRel_i} = {\longueflecheRel{d_i,\flow_i}}$. 
 \end{enumerate}
 Finally, the \emph{monitored language} $\Lg(\A)$ of an LHA $\A$ is the set of timed quantitative words associated with some accepting run of~$\A$.
\end{definition}
In the above definition of association of $\word$  to $\run$, note that the lengths of $\run$ and $\word$ can differ ($n\neq m$). 
Condition~\ref{def:monitoredLang:1} is when a sample in $\word$ happens to be simultaneous with some transition in~$\run$.
This special case is not required to happen at all, for $\word$ to be associated with~$\run$. 

For example, in \cref{figure:examples:acc:signal}, mode switches (\ie{} discrete transitions) in the LHA in \cref{figure:LHAsystemmodelExample} can occur at times other than $t=0,10,20$. 
This is in contrast to the language of hybrid automata in~\cite{AKV98}, where (observable) discrete transitions are always \emph{synchronous} with the word, much like the condition~\ref{def:monitoredLang:1}.

\begin{example}
 \label{example:timed_quantitative_word}
 Let $\A$ be the LHA in \cref{figure:LHAsystemmodelExample},
 $\run$ be the run of $\A$ in \cref{example:run}.
 The timed quantitative word $\word$ in \cref{figure:automotivePlatooningLog}
 is associated with $\run$. %
 We note that the sampling and the discrete transitions are \emph{asynchronous}:
 the sampling is at 0, 10, and 20, 
 while the discrete transitions are at $\frac{34}{3}$ and 12.
 This is in contrast to the \emph{synchronous} language in~\cite{AKV98}:
 the accepted words represent the discrete transitions, \eg{} at $\frac{34}{3}$ and 12.
\end{example}
\section{The Model-Bounded Monitoring Scheme}\label{section:monitoring}
Based on the technical definitions in \cref{section:monitoredLang}, we formally introduce the scheme that we sketched in \cref{figure:modelAwareMonitoring}. (Partial) algorithms for computing if $w\in \Lmon(\mathcal{M}_{\lnot\varphi})$ are introduced in later sections. 
Recall that we focus on safety specifications that are global and linear.
\begin{definition}
 [the LHA $\A_{\lnot\varphi}$]
 \label{def:Mlnotvarphi}
 Let $\A$ be an LHA, and $\varphi\in \GuardP$ (\cref{section:preliminaries}). The LHA $\A_{\lnot\varphi}$ is defined by 
\begin{itemize}
 \item making a copy $\A^{\circ}$ of $\A$,
 \item making each location $\loc^{\circ}$ of $\A^{\circ}$  non-initial, \ie{} $\Init(\loc^{\circ})=\bot$, 
 \item letting $\LocFinal$ consist of all the states $\loc^{\circ}$ of $\A^{\circ}$, and
 \item for each location $\loc\in \A$, creating an edge $(\loc,\neg \varphi, \emptyset,\loc^{\circ})$ from $\loc$ to its copy $\loc^{\circ}$, labeling the edge with the violation of safety specification $\varphi$ as a guard and no update.
\end{itemize}
\end{definition}
\noindent \cref{figure:collision:product} shows an example of  $\A_{\lnot\varphi}$.
In $\A_{\lnot\varphi}$,  having a single accepting sink state for a violation of $\varphi$ is not enough. After detecting violation, we are still obliged to check if the rest of a word $\word$ conforms with the bounding model $\A$. Thus, we maintain a copy of $\A$. 

\begin{lemma}\label{lem:Mlnotvarphi}
The following are equivalent for each sequence $\run$.
\begin{itemize}
 \item Both of the following hold:  
       \begin{ienumeration}
        \item $\run$ is a (non-necessarily accepting) run of  $\A$ and
        \item $\run$ violates $\varphi$ at a certain time instant.
       \end{ienumeration}
 \item There is an accepting run $\run'$ of $\A_{\lnot\varphi}$ such that  \\
 \begin{math}
  \{ \word \mid \text{$\word$ is associated with $\run$}\}
  =
  \{ \word \mid \text{$\word$ is associated with $\run'$}\}
 \end{math}. \qed{}
\end{itemize}
\end{lemma}
The proof is easy by definition. The runs $\run$ and $\run'$ can differ only in the locations they visit---in an LHA, an enabled transition is not always taken. Note, however, that violation of $\varphi$ and $\word$'s association are two properties that are insensitive to locations. 

We are ready to state the correctness of  our scheme (\cref{figure:modelAwareMonitoring}). The proof is straightforward by \cref{lem:Mlnotvarphi} and \cref{def:monitoredLang}. 
\begin{theorem}[correctness]\label{thm:soundness}
  In the setting of \cref{def:Mlnotvarphi}, let $\word$ be a timed quantitative word. We have $w\in \Lmon(\A_{\lnot\varphi})$  if and only if there is a (non-necessarily accepting) run $\run$ of $\A$ such that: \begin{ienumeration}
 \item $\word$ is associated with $\run$ and
 \item $\run$ violates $\varphi$ at some time instant.
\end{ienumeration}
\qed{}
\end{theorem}

Identifying a run $\run$ with a behavior $\signal$, and  association of $\word$ to $\run$ with sampling, the theorem establishes the feature~\ref{item:featureCorrectness} of our scheme (\cref{section:introduction}). 

The consequence in the safety analysis of  the real SUM (instead of its bounding model $\A$) is discussed in the ``usage scenario'' paragraph of \cref{section:introduction}. In particular, due to potential gaps between the SUM and the bounding model $\A$, an alert of our monitor can be false, while the absence of an alert proves safety. Overall, \emph{our model-bounded monitoring scheme is sound}.

\begin{example}
  \label{example:model_bounded_monitoring}
 We show how the illustration in \cref{example:automotivePlatooningMLnotVarphi} is formalized by the monitored language $\Lg(\A_{\lnot \varphi})$ of the LHA $\A_{\lnot \varphi}$.
 Let $\A_{\neg\varphi}$ be the LHA in \cref{figure:collision:product} and $\word$ be the timed quantitative word in \cref{figure:automotivePlatooningLog}.
 We have $\word[2] \not\in \Lg(\A_{\lnot\varphi})$ since all the runs with which  $\word[2]$ is associated are not accepting. That is, the log $\word$ is safe until time $t=10$. 

 However, for the full log, we have $\word \in \Lg(\A_{\neg\varphi})$, because of the following accepting run $\run$ with which $\word$ is associated:
 \begin{math}
  \run = \bigl[\,(\loc_0,\wordVal_1) \longueflecheRel{10,(8.3, 8.2)}  (\loc_0,\wordVal_2) \longueflecheRel{4,(7.5, 9.0)} (\loc_0,\varval) \longueflecheRel{\edge} (\loc_2,\varval) \longueflecheRel{6,\left(\frac{25}{3},8.0\right)} (\loc_2,\wordVal_2) \,\bigr]
 \end{math},  where $\varval = (153, 153)$.
 \LongVersionActuallyNot{The above discussions result in the illustration in \cref{figure:examples:acc:signal}: 
 there is no collision in $t \in [0,10]$ (guaranteed modulo the model soundness assumption, \cref{section:introduction}), while there is a potential collision in $t \in (10,20]$.}
\end{example}
\paragraph{Model-bounded monitoring with an imprecise bounding model}%
\label{newtext:imprecise}
In our model-bounded monitoring scheme, the bounding model~$\A$ must overapproximate the real SUM at least locally.
When the bounding model does not include the actual behavior~$\sigma$, model-bounded monitoring may fail to detect a violation of~$\varphi$.
We have the following two cases of the failure.

The first case is such that an observed sample $(\wordVal_i,\tau_i)$ is unreachable from the previous sample $(\wordVal_{i-1},\tau_{i-1})$ according to $\A$.
In this case, since there is no run $\run$ of the LHA $\A_{\neg \varphi}$ with which $\word$ is associated, model-bounded monitoring cannot detect any safety violation.

The second case is such that $(\wordVal_i,\tau_i)$ is reachable from $(\wordVal_{i-1},\tau_{i-1})$ but $\A$ does not capture the behavior in-between them.
In this case, even if there is a run $\run$ of $\A_{\neg \varphi}$ with which $\word$ is associated,
the guard of the transition to the unsafe locations $\loc^{\circ}$ may not be satisfied.
Therefore, model-bounded monitoring may fail to detect a violation of~$\varphi$.
Nevertheless, when the deviation of the bounding model~$\A$ from the real SUM is smaller than the distance between the actual behavior~$\sigma$ and the threshold of~$\varphi$,
since the guard of the transition to the unsafe locations remains satisfied, model-bounded monitoring can detect a violation of~$\varphi$.
This is much like the discussion on the robustness degree in~\cite{FP09}.
More formal discussion is a future work.

In order to ensure correctness of our approach,
it is therefore preferable to use even a very rough over-approximated bounding model, rather than a tighter---but possibly incorrect---bounding model. We believe that building a rough over-approximated bounding model is often relatively easy, as soon as one knows at least a little from the system domain.

\section{Membership for Monitored Languages: Symbolic Interpolation}\label{sec:symbInt}
The rest of the paper is devoted to solving the membership problem of $\Lmon(\A)$, a core computation task in \cref{figure:modelAwareMonitoring}. We will present two (partial) algorithms: they are symbolic algorithms that iteratively update polyhedra.

\smallskip

\defProblem{The $\Lmon$ membership}{
An LHA $\A$ and
a timed quantitative word $\word$.}{
Return the set $\matching(\word, \A)$ of indices $i$ satisfying
$\wordVal[i] \in \Lg(\A)$. In particular, $\word\in \Lg(\A)$ iff $|\word|\in \matching(\word, \A)$. 
}%
\begin{example}
 \label{example:monitoring}
 Let $\A_{\neg\varphi}$ be the LHA in \cref{figure:collision:product} and $\word$ be the timed quantitative word in \cref{figure:automotivePlatooningLog}.
 We have $\matching(\word, \A_{\neg\varphi}) = \{3\}$,
 meaning $\word[1], \word[2] \not\in \Lmon(\A_{\neg\varphi})$, and $\word[3] (= \word) \in \Lmon(\A_{\neg\varphi})$. %
This result corresponds to the illustration in  \cref{figure:examples:acc:signal}.
\end{example}

The following ``no-go'' theorem is not very surprising, given previous results from  the hybrid automata literature. 
\ACMVersion{The proof is by a reduction from the bounded-time reachability of LHAs.}
\begin{theorem}
 [undecidability]
 \label{theorem:undecidability}
 For an LHA $\A$ and a timed quantitative word $\word$,
 it is undecidable to decide the emptiness of $\matching(\word, \A)$.
\end{theorem}
\begin{proof} 
 We prove the claim by a reduction from the bounded-time reachability of the LHA $\A = (\Loc, \LocFinal, \Variables, \VarValInit, \Flow, \invariant, \Edges)$.
 Let $\tau \in \Rp$ be the time bound of the reachability checking.
 Let $\wordVal_0 \colon \Variables \to \R$ be the valuation satisfying $\wordVal_0(\variable) = 0$ for any $\variable \in \Variables$,
 $\word = (\wordVal_0, \tau)$, and
 $\A' = (\Loc \disjointUnion \{\loc_{f}\}, \{\loc_{f}\}, \Variables \disjointUnion \{t\}, \VarValInit, \Flow', \invariant', \Edges \disjointUnion \Edges')$, where:
 \begin{itemize}
  \item $\Flow'$ is $\Flow'(\loc) = \{\Flow(\loc) \land \dot{t} = 1\}$ for any $\loc \in \Loc$, and
        $\Flow'(\loc_{f})$ is such that $\dot{\variable} = 0 $ for any $\variable \in \Variables$ and $\dot{t} = 1$;
  \item $\invariant'$ is $\invariant'(\loc) = \invariant(\loc)$ for $\loc \in \Loc$ and $\invariant'(\loc_{f}) = \top$; and
  \item $\Edges' = \{(\loc, \{t \leq \tau \}, \updates, \loc_{f}) \mid \loc \in \LocFinal, \updates(\variable) = 0 \text{ for any } \variable \in \Variables\}$.
 \end{itemize}
 We have $(\wordVal_0, \tau) \in \Lg(\A')$ if and only if
 $\LocFinal$ is reachable in $\A$ within $\tau$.
 Since we have $|\word| = 1$, the bounded-time reachability of the LHA $\A$, which is undecidable~\cite{BDGORW11}, is reduced to the emptiness checking of $\matching(\word, \A')$.
\end{proof}
Given the undecidability result, we can think of 
 restricting the class of the models. 
For example, the problem becomes decidable under some monotonicity constraints~\cite{BDGORW13}, or if the number of discrete transitions within a time unit is bounded~\cite{BWWL19}. %

Nevertheless, in practice, as we observe in \cref{section:experiments}, our partial algorithms below perform effectively for many benchmarks on the full LHA class---especially our latter, direct algorithm.

Our two partial algorithms  have the following features. 
\begin{itemize}
 \item 
       Instead of solving one-way reachability (forward or backward) as many existing algorithms do, they solve \emph{interpolation} between two points, \ie{} samples (see \cref{figure:examples:acc:signal}).
 \item 
       They work in a \emph{one-shot manner}, collecting the linear constraints representing the interpolations from the given origin to the end with only bounded-time forward reachability analysis. 
       Instead, a naive method  would iterate between forward and backward reachability analysis.
\end{itemize}
\section{Method I:  via Reduction to LHA Reachability Analysis}\label{section:algorithm1}

In our first solution, we reduce the  $\Lmon$ membership problem to reachability analysis of LHAs.
In practice, we will use \phaverlite{}, one of the most efficient tools for reachability analysis of hybrid systems according to~\cite{BZ19}.

The idea of reducing monitoring to reachability analysis of extensions of finite-state automata is not new and was already proposed in the literature, \eg{}~\cite{AHW18}.
While both~\cite{AHW18} and the method we introduce in this section are symbolic; the differences are in the formalism and problem.
On the one hand,~\cite{AHW18} uses parametric timed automata as a parametric specification and performs \emph{parametric} timed pattern matching (which can be seen as parametric monitoring).
On the other hand, we use LHAs for the bounding model and we perform symbolic monitoring.
An extension for a parametric setting is future work, which is technically not very demanding.

Our workflow is as follows:
\begin{enumerate}
	\item We transform the input timed quantitative word $\word$ into an LHA $\A_{\word}$ (that is in fact only timed, \ie{} it only uses clocks), that uses two extra variables:
	\begin{enumerate}
		\item $\clockabs$ measures the absolute time since the beginning of the word; and
		\item $\clockrel$ measures the (relative) time since the last sampled timed quantitative word.
	\end{enumerate}
	
	\item We perform the synchronized product $\A \product \A_{\word}$ of the given LHA $\A$ with the transformed LHA $\A_{\word}$.
	
	\item We run the reachability analysis procedure for the product LHA $\A \product \A_{\word}$, to derive all possible locations $\wloc_i$ of $\A_{\word}$ such that $(\loc, \wloc_i)$ is reachable in $\A \product \A_{\word}$ with $\clockrel = 0$, where $\loc$ is an accepting location of the given LHA $\A$. %
\end{enumerate}

We explain these steps in the following.

\subsection{Transforming the Timed Quantitative Word into an LHA}

First; we transform the input timed quantitative word $\word$ into an LHA.
The resulting LHA $\A_{\word}$ is a simple sequence of locations with guarded transitions in between, also resetting~$\clockrel$.

The LHA $\A_{\word}$ features an absolute time clock $\clockabs$ (initially~0, of rate~1 and never reset), and can test all variables of the system in guards (these are not reset in this LHA, though).
In more detail, we simply convert each sample $(\wordVal_i, \tau_i)$ of the timed quantitative word $\word$ into a guard of the LHA testing for the timestamp using the absolute time clock~$\clockabs$, and for the value of the variables.
The invariant of the location preceding a timestamp $\tau_i$ also features the clock constraint $\clockabs \leq \tau_i$ (this is not crucial for correctness but limits the state space explosion).
The transitions are all labeled with a fresh action~\styleact{sample} (which could be replaced with an unobservable action, but such actions are not accepted by the \phaverlite{} model checker).
Each transition resets~$\clockrel$.
	Overall, this procedure shares similarities with the one that transforms a timed word into a timed automaton, proposed in~\cite{AHW18}.
Let \TransWord{} denote this procedure.
	For example, consider the timed quantitative word~$\word$ in \cref{figure:automotivePlatooningLog}.  %
	The result $\A_{\word}$ of $\TransWord(\word)$ is given in \cref{figure:translation:word}.
\begin{figure}[tb]
	\centering
	\small
		\scalebox{1.0}{\begin{tikzpicture}[shorten >=1pt,node distance=2.3cm,on grid,auto]
		\node[location] (w0) {$\wloc_0$};
		\node[location] (w1) [right of=w0] {$\wloc_1$};
		\node[location] (w2) [right of=w1] {$\wloc_2$};
		\node[location] (w3) [right of=w2]{$\wloc_3$};
		
		\node[invariant, below] at (w0.south) {$\stylevar{\clockabs} \leq 0$};
		\node[invariant, below] at (w1.south) {$\stylevar{\clockabs} \leq 10$};
		\node[invariant, below] at (w2.south) {$\stylevar{\clockabs} \leq 20$};

		\draw[<-] (w0) -- node[above, align=center] {$\variablei{1} = 40$ \\ $\variablei{2} = 35$ \\ $\clockabs = 0$ \\ $\clockrel = 0$} ++(-1.2cm,0);

		\path[->]
			(w0) edge
				node[above][align=center] {$\stylevar{\clockabs} = 0$ \\ $\land \stylevar{\variablei{1}} = 40$ \\ $\land \stylevar{\variablei{2}} = 35$ \\ $\styleact{sample}$}
				node[below]{$\clockrel \assign 0$}
				(w1)
			
			(w1) edge
				node[above] [align=center] {$\stylevar{\clockabs} = 10$ \\ $\land \stylevar{\variablei{1}} = 123$ \\ $\land \stylevar{\variablei{2}} = 117$ \\$\styleact{sample}$}
				node[below]{$\clockrel \assign 0$}
				(w2)
			
			(w2) edge
				node[above] [align=center] {$\stylevar{\clockabs} = 20$ \\ $\land \stylevar{\variablei{1}} = 203$ \\ $\land \stylevar{\variablei{2}} = 201$ \\ $\styleact{sample}$} 
				node[below]{$\clockrel \assign 0$}
				(w3)
		;
		\end{tikzpicture}}
		
	\caption{$\TransWord$ applied to the timed quantitative word in \cref{figure:automotivePlatooningLog}. Here, \emph{i}) $\Variables = \{ \variablei{1} , \variablei{2}, \clockabs, \clockrel \}$, \emph{ii})
		 $\VarValInit$ is such that $\VarValInit(\loc_0) = \{ \variablei{1} = 40 \land \variablei{2} = 35 \land \clockabs = 0 \land \clockrel = 0 \}$ and $\VarValInit(\wloc_i) = \bot$ for $1 \leq i \leq 3$, and \emph{iii}) $\dot{\clockabs} = \dot{\clockrel} = 1$ and $\variabledoti{1}, \variabledoti{2} \in \R$ in all locations. %
Invariants are boxed under the location.}\label{figure:translation:word}
\end{figure}
\subsection{Reachability Analysis Using \phaverlite{}}

We perform the synchronized product $\A_{\word} \product \A$ (``parallel composition'') of the LHA $\A_{\word}$ constructed from $\word$ together with the given LHA $\A$.

Then, we run the reachability analysis, setting as target the states for which both of the following conditions hold:
\begin{ienumeration}
	\item the monitor is in an accepting location and
	\item $\clockrel = 0$.
\end{ienumeration}
The latter condition ensures that only the states such that we just sampled a word are accepting.
Thanks to the latter condition, we can take into account of the next sample without any explicit backward reachability analysis.

This is intuitively because of the following.
While $\clockrel > 0$ and we are at $\wloc_i$ in $\A_{\word}$, we compute all the reachable valuations from $\word[i]$ by the forward reachability analysis. Here, we non-deterministically make an assumption of the next sample including the ones incompatible with the actual sample $\word[i+1]$.
When we take the transition from $\wloc_i$ to $\wloc_{i+1}$ and $\clockrel = 0$, we require that the assumption of the next sample must be compatible with the actual sample $\wloc_{i+1}$ and accept only if the accepting locations are reachable by an interpolation between $\word[i]$ and $\word[i+1]$.

\begin{example}%
 \label{example:indirect_method}
	Let us exemplify the need for the latter condition.
	Consider the LHA $\A_{\neg \varphi}$ in \cref{figure:collision:product} and the timed quantitative word $\word$ in \cref{figure:automotivePlatooningLog} transformed into the LHA $\A_{\word}$ in \cref{figure:translation:word} (only the time frame in $[0,10]$ is of interest in this example).
	Clearly, this log is safe \wrt{} the LHA $\A_{\neg \varphi}$ in \cref{figure:collision:product}, that is neither $\loc_2$ nor~$\loc_3$ are reachable, since the distance between both vehicles cannot be $\leq 0$ in the $[0,10]$ time frame.
	
        If we simply run the reachability procedure looking for~$\loc_2$ or~$\loc_3$ as target (without condition on~$\clockrel$), the procedure will output that at least~$\loc_2$ is reachable.
	Indeed, it is possible that vehicle~1 runs at the minimal rate of 7.5 while vehicle~2 runs at the maximal rate of~9.
	In that case, after 10 time units, vehicle~1 (resp.~2) reaches $x$-coordinate 115 (resp.~125), and thus their distance is $\leq 0$, making~$\loc_2$ reachable.
	While this behavior is indeed possible from the knowledge we have of the first sample, it is actually impossible knowing the full log and in particular the second sample.
	This phenomenon is illustrated in the part of \cref{figure:examples:acc:signal} restricted to the $[0,10]$ time frame: the blue part depicts all possible valuations knowing the first and second sample.
\end{example}

Hence, adding the condition $\clockrel = 0$ forces the model checker to take into consideration the next sample before making a decision concerning the reachability of a possible target location.

\section{Method II: Direct Method by Polyhedra Computation}\label{section:algorithm2}

In our second solution, we directly solve the  $\Lmon$ membership problem.
We iteratively compute the runs of the LHA $\A$ associated with the prefixes of the timed quantitative word $\word$ utilizing bounded reachability analysis.
This is our main contribution.

\begin{algorithm}[tbp]
 \caption{Outline of our incremental procedure for the $\Lmon$ membership problem}%
 \label{algorithm:incremental:outline}
 \DontPrintSemicolon{}
 \newcommand{\myCommentFont}[1]{\texttt{\footnotesize{#1}}}
 \SetCommentSty{myCommentFont}
 \KwIn{A timed quantitative word $\word = (\wordVal_1, \tau_1), (\wordVal_2, \tau_2),\dots, (\wordVal_n, \tau_n)$ and
 an LHA $\A = (\Loc, \LocFinal, \Variables, \VarValInit, \Flow, \invariant, \Edges)$.}
 \KwOut{The Boolean sequence $\Resulti{1},\dots,\Resulti{n}$, where
 $\Resulti{i} = \top \iff \word[i] \in \Lg(\A)$}
 $\Conf_{0} \gets \big\{(\loc, \varval) \mid \loc \in \Loc, \varval \in \VarValInit(\loc)\big\}$\;
 \label{algorithm:incremental:outline:init}
 \For{$i \gets 1$ \KwTo $n$}{
 \label{algorithm:incremental:outline:loop-begin}
 $\Conf'_{i} \gets$ reachable states from $\Conf_{i-1}$ in duration $\tau_i - \tau_{i-1}$ \;
 \label{algorithm:incremental:outline:reachability}
 $\Conf_{i} \gets \big\{(\loc,\varval) \in \Conf'_{i} \mid \varval = \wordVal_i \big\}$ \;
 \label{algorithm:incremental:outline:restrict}
 $\Resulti{i} \gets \exists (\loc, \varval) \in \Conf_{i}.\, \loc\in\LocFinal$ \;\nllabel{algo:accepting}
 \label{algorithm:incremental:outline:result}
 \label{algorithm:incremental:outline:loop-end}
 }
\end{algorithm}
\cref{algorithm:incremental:outline} shows an outline of our incremental procedure for the  $\Lmon$ membership problem.
\cref{algorithm:incremental:outline} incrementally constructs the intermediate states $\Conf_{i}$ and outputs the partial result $\Resulti{i}$ showing if $\word[i] \in \Lg(\A)$.
In \cref{algorithm:incremental:outline:init}, we construct the initial states $\Conf_{0}$.
We note that although $\Init(\loc)$ is, in general, an infinite set, it is given as a convex polyhedron, and we can represent $\Conf_{0}$ as a finite list of pairs of a location and a convex polyhedron.

From \cref{algorithm:incremental:outline:loop-begin} to \cref{algorithm:incremental:outline:loop-end} is the main part of  \cref{algorithm:incremental:outline}: we incrementally compute $\Conf_{i}$ and $\Resulti{i}$. 
In \cref{algorithm:incremental:outline:reachability}, we compute the reachable states $\Conf_{i}$ from $\Conf_{i-1}$ after the executions of duration $\tau_{i} - \tau_{i-1}$.
This part is essentially the same as the bounded-time reachability analysis, and thus, it is undecidable for LHAs in general~\cite{BDGORW11}.
Nevertheless, in practice, the reachable states $\Conf'_{i}$ are usually effectively computable as a finite union of convex polyhedra.

In \cref{algorithm:incremental:outline:restrict}, we require $\Conf_{i}$ to be the subset of $\Conf'_{i}$ compatible with the current observation $\wordVal_i$.
Thanks to this requirement, we can take into account of the next sample just by the forward reachability analysis.
Finally, in \cref{algorithm:incremental:outline:result}, we determine the partial result $\Resulti{i}$ by checking the reachability to the accepting locations.

\begin{example}%
 \label{ex:direct}
 Let $\word$ and $\A$ be the ones in \cref{example:monitoring}.
 In \cref{algorithm:incremental:outline:init} of \cref{algorithm:incremental:outline}, we let $\Conf_{0} = \{(\loc_0, (40, 35))\}$.
 In \cref{algorithm:incremental:outline:reachability}, we conduct the time-bounded reachability analysis of duration 10. 
 The result is as follows.
 
 {\small
 \begin{align*}
   \Conf'_{1} &= \{(\loc_0, \varval_1) \mid \varval_1 (\variable_1) \in [115,125], \varval_1(\variable_2) \in [115, 125]\}\\
  \cup &\{(\loc_0, \varval_1) \mid -3 \varval_1(\variable_1) + 11\varval_1(\variable_2) \geq 876, -2\varval_1(\variable_1) + 9\varval_1(\variable_2) \geq 789,\\
  & \qquad\qquad\,\, \varval_1(\variable_2) \leq 431/3, \varval_1(\variable_1) \leq 499/3, \varval_1(\variable_2) \geq 115, %
  \varval_1(\variable_1) \geq 115, 4\varval_1(\variable_1) - 7\varval_1(\variable_2) \geq - 415 \}\\
  \cup &\{(\loc_1, \varval_1) \mid 
  -3\varval_1(\variable_1) + 11\varval_1(\variable_2) \geq 876, -\varval_1(\variable_1) + 5\varval_1(\variable_2) \geq 456,\\
  &\qquad\qquad\,\, -3\varval_1(\variable_2) \geq -431, -3\varval_1(\variable_1) \geq -499, \varval_1(\variable_2) \geq 115, %
  \varval_1(\variable_1) \geq 115, 4\varval_1(\variable_1) - 7\varval_1(\variable_2) \geq -415
  \}\\
  \cup &\{(\loc_2, \varval_1) \mid 
  -18\varval_1(\variable_1) + 15\varval_1(\variable_2) \geq -415, -\varval_1(\variable_1) + 2\varval_1(\variable_2) \geq 115, %
  4\varval_1(\variable_1) - 7\varval_1(\variable_2) \geq -415, \varval_1(\variable_1) \geq 115
  \}\\
  \cup &\{(\loc_3, \varval_1) \mid
  \varval_1(\variable_1) \in [115,455/3],
  -3\varval_1(\variable_1) + 11\varval_1(\variable_2) \geq 920, %
  \varval_1(\variable_2) \leq 415/3, 4\varval_1(\variable_1) - 7\varval_1(\variable_2) \geq -415
  \}
\end{align*}
}

 In \cref{algorithm:incremental:outline:restrict}, we require $\variable_1 = 123$ and $\variable_2 = 117$, and we have $\Conf_{1} = \{(\loc_0, (123,127)), (\loc_1, (123,127))\}$.
 Since $\loc_0 \not\in \LocFinal$ and $\loc_1 \not\in \LocFinal$, we have $\Resulti{1} = \emptyset$.

 After incrementing $i$ in \cref{algorithm:incremental:outline:loop-begin},
 in \cref{algorithm:incremental:outline:reachability}, we again conduct the time-bounded reachability analysis. The result is as follows.
 
 {\small
 \begin{align*}
  \Conf'_{2} &= \{(\loc_0, \varval_2) \mid
  \varval_2(\variable_1) \in [198,253], \varval_2(\variable_2) \in [197, 227], %
  -2\varval_2(\variable_1) + 9\varval_2(\variable_2) \geq 1357, 4\varval_2(\variable_1) - 7\varval_2(\variable_2) \geq -657
  \}\\
  \cup &\{(\loc_1, \varval_2) \mid
  \varval_2(\variable_1) \in [233,253], \varval_2(\variable_2) \in [207, 227]
  \}\\
  \cup &\{(\loc_1, \varval_2) \mid
  -3\varval_2(\variable_1) + 11\varval_2(\variable_2) \geq 1540, -\varval_2(\variable_1) + 5\varval_2(\variable_2) \geq 784, \\
  &\qquad\qquad\,\, -3\varval_2(\variable_2) \geq -673, -3\varval_2(\variable_1) \geq -737, \varval_2(\variable_2) \geq 197, %
  \varval_2(\variable_1) \geq 198, 4\varval_2(\variable_1) - 7\varval_2(\variable_2) \geq -657
  \}\\
  \cup &\{(\loc_2, \varval_2) \mid
  -6\varval_2(\variable_1) + 5\varval_2(\variable_2) \geq -219, -\varval_2(\variable_1) + 2\varval_2(\variable_2) \geq 198, %
  4\varval_2(\variable_1) - 7\varval_2(\variable_2) \geq -657, \varval_2(\variable_1) \geq 198
  \}\\
  \cup &\{(\loc_3, \varval_2) \mid
  \varval_2(\variable_1) \in [198,231], \varval_2(\variable_2) \leq 219, %
  -3\varval_2(\variable_1) + 11\varval_2(\variable_2) \geq 1584, 4\varval_2(\variable_1) - 7\varval_2(\variable_2) \geq -657
  \}
 \end{align*}
 }
 
 In \cref{algorithm:incremental:outline:restrict}, we require $\variable_1 = 203$ and $\variable_2 = 201$.
 This time, we have $\Conf_{1} = \{(\loc, (203,201)) \mid \loc \in \Loc\}$.
 Thus, we have $\Resulti{2} = \{(203,201)\}$.
\end{example}

The intermediate states set $\Conf_{i}$ is the set of the last states of the runs of $\A$ associated with $\word[i]$ and of duration $\tau_i$.
Therefore, we have the following correctness theorem.

\begin{theorem}
 [correctness of \cref{algorithm:incremental:outline}]%
 \label{theorem:correctness:incremental}
 Given a timed quantitative word $\word$ and an LHA $\A$,
 \cref{algorithm:incremental:outline} returns the sequence
 $\Resulti{1},\dots,\Resulti{n}$ satisfying
 $\Resulti{i} = \top \iff \word[i] \in \Lg(\A)$ if it terminates.
\end{theorem}

To prove \cref{theorem:correctness:incremental}, we show the following lemma.

\begin{lemma}%
 \label{lemma:correctness:run_and_state}
 Let
 $\word$ be a timed quantitative word $\word = (\wordVal_1, \tau_1), \dots, (\wordVal_m, \tau_m)$,
 $\A$ be  an LHA $\A = (\Loc, \LocFinal, \Variables, \VarValInit, \Flow, \invariant, \Edges)$, and
 $\Conf_{i} \in \Loc \times \Val$ be the $\Conf_{i}$ in \cref{algorithm:incremental:outline:restrict} of \cref{algorithm:incremental:outline},
 where $i \in \{1,2,\dots,m\}$.
 For any $i \in \{1,2,\dots,m\}$, $\loc \in \Loc$, and $\varval \in \Val$,
 we have $(\loc, \varval) \in \Conf_{i}$ if and only if 
 there is a run $\run_i = s_0, \flecheRel_1, s_1, \dots, \flecheRel_n, s_n = (\loc, \varval)$ of $\A$ associated with $\word[i]$ and satisfying $\Duration(\run_i) = \tau_i$.
\end{lemma}
\begin{proof}
 We prove \cref{lemma:correctness:run_and_state} by induction on $i$.
 When $i = 1$, $\Conf_{1}$ is the set of states reachable from $\{(\loc_0,\varval_0) \mid \loc_0 \in \Loc, \varval_0 \in \VarValInit(\loc)\}$ in $\tau_1$ and 
 satisfying $\varval = \wordVal_1$.
 Therefore, for any $(\loc_1,\varval_1) \in \Conf_{1}$, there are $\loc_0 \in \Loc$, $\varval_0 \in \VarValInit(\loc)$, and
 a run $\run_1 = (\loc_0,\varval_0), \flecheRel_1, s_1, \dots, \flecheRel_n, (\loc_1, \varval_1)$ of $\A$ satisfying $\Duration(\run_1) = \tau_1$.
 By \cref{def:monitoredLang}, such $\run_1$ is associated to $\word[1]$, and thus, 
 when we have $(\loc_1, \varval_1) \in \Conf_{1}$,
 there is a run $\run_1 = s_0, \flecheRel_1, s_1, \dots, \flecheRel_n, s_n = (\loc_1, \varval_1)$ of $\A$ associated with $\word[1]$ and satisfying $\Duration(\run_1) = \tau_1$.
 When there is a run $\run_1 = s_0, \flecheRel_1, s_1, \dots, \flecheRel_n, s_n = (\loc_1, \varval_1)$ of $\A$ associated with $\word[1]$ and satisfying $\Duration(\run_1) = \tau_1$,
 since $s_0 \in \Conf_{0}$ holds, we have $(\loc_1, \varval_1) \in \Conf_{1}$.

 When $i > 1$, by definition in \cref{algorithm:incremental:outline}, $\Conf_{i}$ is the set of states $(\loc_{i}, \varval_{i}) \in \Loc \times \Val$ satisfying the following.
 \begin{itemize}
  \item There is a state $(\loc_{i-1}, \varval_{i-1}) \in \Conf_{i-1}$ such that $(\loc_{i}, \varval_{i})$ is reachable from $(\loc_{i-1}, \varval_{i-1})$ in $\tau_i - \tau_{i-1}$.
  \item We have $\varval_i = \wordVal_i$.
 \end{itemize}
 By definition of the reachability in $\A$, $\Conf_{i}$ is the set of states $(\loc_{i}, \varval_{i})$ such that there is an alternating sequence 
 $\tilde{\run}_i = s_0, \flecheRel_1, s_1, \dots, \flecheRel_n, s_n = (\loc_{i}, \varval_{i})$ 
 of concrete states $s_i \in \States$ and transitions $\flecheRel_i \in \flecheRel$ satisfying the following.
 \begin{itemize}
  \item $s_0 \to_1 s_1 \dots \to_n s_n$
  \item $s_0 \in \Conf_{i-1}$
  \item Sum of the delays in $\tilde{\run}_i$ is $\tau_i - \tau_{i-1}$
  \item $\varval_i = \wordVal_i$
 \end{itemize}
 By induction hypothesis, for any state $(\loc_{i-1}, \varval_{i-1}) \in \Conf_{i-1}$, 
 there is a run $\run_{i-1} = s'_0, \flecheRel'_1, s'_1, \dots, \flecheRel'_n, s'_n = (\loc_{i-1}, \varval_{i-1})$ of $\A$ associated with $\word[i-1]$ and satisfying $\Duration(\run_{i-1}) = \tau_{i-1}$.
 Therefore, for any $\loc \in \Loc$ and $\varval \in \Val$ satisfying  $(\loc, \varval) \in \Conf_{i}$,
 there is a run 
 \begin{align*}
  \run_i = \run_{i-1} \cdot \tilde{\run}_{i} = s'_0, \flecheRel'_1, s'_1, \dots, \flecheRel'_n, s'_n, \flecheRel_1, s_1, \dots, \flecheRel_n, s_n = (\loc_{i}, \varval_{i})
 \end{align*}
 of $\A$. 
 Moreover, since $\Duration(\run_{i-1}) = \tau_{i-1}$ holds, $\run_{i-1}$ is associated with $\word[i-1]$, and the sum of the delays in $\tilde{\run}_i$ is $\tau_i - \tau_{i-1}$,
 such $\run_{i}$ satisfies $\Duration(\run_{i}) = \tau_{i}$ and $\run_{i}$ is associated with $\word[i]$.

 For any run $\run_i = s_0, \flecheRel_1, s_1, \dots, \flecheRel_n, s_n = (\loc_i, \varval_i)$ of $\A$ associated with $\word[i]$ and satisfying $\Duration(\run_i) = \tau_i$,
 there is a run $\run_{i-1} = s_0, \flecheRel_1, s_1, \dots, \flecheRel_k, s_k = (\loc_{i-1}, \varval_{i-1})$ of $\A$ associated with $\word[i-1]$, and 
 an alternating sequence $\tilde{\run}_i = s_k, \flecheRel_{k+1}, s_{k+1}, \dots, \flecheRel_n, s_n$ 
 satisfying the following.
 \begin{itemize}
  \item $\Duration(\run_{i-1}) = \tau_{i-1}$
  \item sum of the delays in $\tilde{\run}_i$ is $\tau_i - \tau_{i-1}$
 \end{itemize}
 By induction hypothesis, we have $(\loc_{i-1}, \varval_{i-1}) \in \Conf_{i-1}$.
 Moreover, by the existence of $\tilde{\run}_i$, $(\loc_i, \varval_i)$ is reachable from $(\loc_{i-1}, \varval_{i-1}) \in \Conf_{i-1}$, and thus, we have
 $(\loc_{i}, \varval_{i}) \in \Conf_{i}$.
\end{proof}

Let us now go back to the proof of \cref{theorem:correctness:incremental}.

\begin{proof}%
 [Proof of \cref{theorem:correctness:incremental}]
 For each $i\in \{1,2,\dots,m\}$, by \cref{algorithm:incremental:outline:result} of \cref{algorithm:incremental:outline}, 
 we have $\Resulti{i} = \top$ if and only if there is $(\loc,\varval) \in \Conf_{i}$ satisfying $\loc \in \LocFinal$.
 By \cref{lemma:correctness:run_and_state}, for each $(\loc,\varval) \in \Conf_{i}$, 
 there is a run $\run_i = s_0, \flecheRel_1, s_1, \dots, \flecheRel_n, s_n = (\loc, \varval)$ of $\A$ associated with $\word[i]$.
 Therefore, we have $\Resulti{i} = \top$ if and only if there is an accepting run $\run_i$ of $\A$ associated with $\word[i]$, which is equivalent to $\word[i] \in \Lg(\A)$.
\end{proof}
\section{Experimental Evaluation}\label{section:experiments}

We experimentally evaluated our model-bounded monitoring scheme using the two procedures for $\Lmon$ membership.
For the first %
 procedure via reachability analysis
 (%
 in \cref{section:algorithm1}), we used \phaverlite{}~\cite{BZ19} for reachability analysis.
For the second direct procedure (in \cref{section:algorithm2}), 
we implemented a prototypical tool \masakiTool{}.

We pose the following research questions.
\begin{description}
 \item[RQ1] Is our dedicated implementation \masakiTool{} worthwhile, performance-wise?
 \item[RQ2] Is \masakiTool{} scalable \wrt{} the length of the input log?
 \item[RQ3] How is the scalability of \phaverlite{} and \masakiTool{} \wrt{} the dimension of the bounding model?
 \item[RQ4] Is there any relationship between the robustness of the log and the precision of model-bounded monitoring? Moreover, can we reduce the sampling interval in the model-bounded monitoring scheme in \cref{figure:modelAwareMonitoring} without causing significant false alarms?
\end{description}

\subsection{Benchmarks}
\begin{table}[tbp]
 \centering
 \caption{Summary of the benchmarks}%
 \label{table:benchmarks_summary}
 \begin{tabular}[t]{c c c r}
  \toprule
  Name & Dimension ($= d$) & \# of locs.& max.\ length of \ACMVersion{the }logs\\
  \midrule
  \Accc{} & 5,10,15& $d+1$& 1,000 \\
  \Acci{} & 2& 4& 100,000 \\
  \AccdSafety{} & $2,3,\dots,7$& $2^{d}$& 1,000 \\
  \NAV{} & $4$& $18$& 150 \\
  \UnsafeNAV{} & $4$& $18$& 150 \\
  \SharedGasBurner{} & $2$& $6$& 5,000 \\
  \bottomrule
\end{tabular}
\end{table}
\cref{table:benchmarks_summary} summarizes the benchmarks for both scalability and precision experiments.

\subsubsection{Benchmarks for RQ1--3 %
   }
In the scalability experiments to answer RQ1--3, we used the following three benchmarks on adaptive cruise controller: Piecewise-Constant ACC (\Accc); Interval ACC (\Acci{}); and Diagonal ACC (\Accd). %
The bounding models for the benchmarks are mostly taken from the literature (see below); they express keeping the inter-vehicular distance by switching between the normal cruise mode and the recovery mode. This is much like \cref{figure:LHAsystemmodelExample}. 

 The input logs $w$ were randomly generated by following the flows and the transitions of the bounding model $\A$. This means that our SUM is the bounding model itself (\cref{figure:modelAwareMonitoring}). We note that this coincidence is not mandatory. %
\begin{figure}[tbp]
 \begin{minipage}[b]{\ACMVersion{0.6\linewidth}\LongVersion{0.95\linewidth}}
  \centering
  \scalebox{.5}{
  \begin{tikzpicture}[shorten >=1pt, scale=\ACMVersion{2}\LongVersion{3}, xscale=1, auto] %
   \node[location] (cruise) at (0,5) [align=center]{
   $\variabledoti{1} = 8$, $\variabledoti{2} = 8.5$,\\ $\variabledoti{3} = 9,$
   $\variabledoti{4} = 9.5$,\\ $\variabledoti{5} = 10$,\\ $2 \leq x_{i} - x_{i+1} \leq 10$};
   \draw[<-] (cruise) -- node[above, align=center] {
   $\stylevar{\variablei{1}} = 40$,
   $\stylevar{\variablei{2}} = 35$,\\
   $\stylevar{\variablei{3}} = 30$,
   $\stylevar{\variablei{4}} = 25$,\\
   $\stylevar{\variablei{5}} = 20$
   } ++(-2.25cm,0);
   \node[location] (rec1) at (-3.3,3) [align=center]{
   $\variabledoti{1} = 12$, $\variabledoti{2} = 10$,\\ $\variabledoti{3} = 8$,
   $\variabledoti{4} = 9$,\\ $\variabledoti{5} = 10$,\\ $0 \leq x_{i} - x_{i+1} \leq 10$};
   \node[location] (rec2) at (-1.3,3) [align=center]{
   $\variabledoti{1} = 12$, $\variabledoti{2} = 12$,\\ $\variabledoti{3} = 10$,
   $\variabledoti{4} = 8.5$,\\ $\variabledoti{5} = 9.5$,\\ $0 \leq x_{i} - x_{i+1} \leq 10$};
   \node[location] (rec3) at (1.3,3) [align=center]{
   $\variabledoti{1} = 12$, $\variabledoti{2} = 12$,\\
   $\variabledoti{3} = 12$, $\variabledoti{4} = 10$,\\ $\variabledoti{5} = 9$,\\
   $0 \leq x_{i} - x_{i+1} \leq 10$};
   \node[location] (rec4) at (3.3,3) [align=center]{
   $\variabledoti{1} = 12$, $\variabledoti{2} = 12$,\\
   $\variabledoti{3} = 12$, $\variabledoti{4} = 12$,\\
   $\variabledoti{5} = 10$,\\ 
   $0 \leq x_{i} - x_{i+1} \leq 10$};
   \node[location,accepting] (unsafe) at (0,1.5) {unsafe};

   \path[->] 
   (cruise) edge [bend right=23] node[above left] {$\stylevar{\variablei{1}} - \stylevar{\variablei{2}} \leq 4$} (rec1)
   (rec1) edge [bend left=20] node[below right, pos=0.2] {$\stylevar{\variablei{1}} - \stylevar{\variablei{2}} \geq 4$} (cruise)
   (cruise) edge [bend right=20] node[left] {$\stylevar{\variablei{2}} - \stylevar{\variablei{3}} \leq 4$} (rec2)
   (rec2) edge [bend left=15] node[right, pos=0.2] {$\stylevar{\variablei{2}} - \stylevar{\variablei{3}} \geq 4$} (cruise)
   (cruise) edge [bend left=20] node[right] {$\stylevar{\variablei{3}} - \stylevar{\variablei{4}} \leq 4$} (rec3)
   (rec3) edge [bend right=10] node[left, pos=0.2] {$\stylevar{\variablei{3}} - \stylevar{\variablei{4}} \geq 4$} (cruise)
   (cruise) edge [bend left=20] node[below left, pos=0.8] {$\stylevar{\variablei{4}} - \stylevar{\variablei{5}} \leq 4$} (rec4)
   (rec4) edge [bend right=23] node[above right] {$\stylevar{\variablei{4}} - \stylevar{\variablei{5}} \geq 4$} (cruise)
   (rec1) edge [bend right=10] node[below left] {$\stylevar{\variablei{1}} - \stylevar{\variablei{2}} \geq 1$} (unsafe)
   (rec2) edge [bend left=0] node[left=10pt] {$\stylevar{\variablei{2}} - \stylevar{\variablei{3}} \geq 1$} (unsafe)
   (rec3) edge [bend left=0] node[right] {$\stylevar{\variablei{3}} - \stylevar{\variablei{4}} \geq 1$} (unsafe)
   (rec4) edge [bend left=10] node[below right] {$\stylevar{\variablei{4}} - \stylevar{\variablei{5}} \geq 1$} (unsafe)
   ;
  \end{tikzpicture}
  }
  \caption{The LHA of dimension~5 in \Accc{}, where $i \in \{1,2,3,4\}$}%
  \label{figure:benchmark:ACCC}
 \end{minipage}
\ACMVersion{\hfill}%
\LongVersion{%

}
 \begin{minipage}[b]{\ACMVersion{0.39\linewidth}\LongVersion{0.95\linewidth}}
  \centering
  \small
  \scalebox{0.75}{\begin{tikzpicture}[shorten >=1pt, scale=\ACMVersion{1.5}\LongVersion{2.5}, yscale=1, auto] %
                  \node[location] (crs) at (0,0) [align=center]{
                  $\variabledoti{1} = 36, 0 \leq \variabledoti{2},$ \\
                  $|\variabledoti{1} - \variabledoti{2}| \leq 1$\\
                  $1 \leq \stylevar{\variablei{1}} - \stylevar{\variablei{2}}$};
                  \draw[<-] (crs) -- node[above, align=center] {$\stylevar{\variablei{1}} = 3,$ \\ $\stylevar{\variablei{2}} = 0$} ++(-1.5cm,0);
                  \draw[<-] (crs) -- node[above, align=center] {$\variablei{1} = 3,$ \\ $\variablei{2} = 0$} ++(-1.5cm,0);
                  \node[location] (rcv) at (2.5,0) [align=center]{
                  $\variabledoti{1} = 36, 0 \leq \variabledoti{2},$ \\
                  $|\variabledoti{1} - \variabledoti{2} - \varepsilon| \leq 1$\\
                  $\variablei{1} - \variablei{2} \leq 3$};

                  \node[location,accepting] (crs_crash) at (0,-1.75) [align=center]{
                  $\variabledoti{1} = 36, 0 \leq \variabledoti{2},$ \\
                  $|\variabledoti{1} - \variabledoti{2}| \leq 1$\\
                  $1 \leq \variablei{1} - \variablei{2}$};
                  \node[location,accepting] (rcv_crash) at (2.5,-1.75) [align=center]{
                  $\variabledoti{1} = 36, 0 \leq \variabledoti{2},$ \\
                  $|\variabledoti{1} - \variabledoti{2} - \varepsilon| \leq 1$\\
                  $\variablei{1} - \variablei{2} \leq 3$};

                  \path[->] 
                  (crs) edge [bend left=10] node {$\stylevar{\variablei{1}} - \stylevar{\variablei{2}} \leq 2$} (rcv)
                  (rcv) edge [bend left=10] node {$\stylevar{\variablei{1}} - \stylevar{\variablei{2}} \geq 2$} (crs)
                  (crs) edge node {$\stylevar{\variablei{1}} \leq \stylevar{\variablei{2}}$} (crs_crash)
                  (rcv) edge node {$\stylevar{\variablei{1}} \leq \stylevar{\variablei{2}}$} (rcv_crash)
                  (crs_crash) edge [bend left=10] node {$\stylevar{\variablei{1}} - \stylevar{\variablei{2}} \leq 2$} (rcv_crash)
                  (rcv_crash) edge [bend left=10] node {$\stylevar{\variablei{1}} - \stylevar{\variablei{2}} \geq 2$} (crs_crash)
                  ;
                 \end{tikzpicture}}
  \caption{The LHA of dimension~2 in \AccdSafety{}}%
  \label{figure:benchmark:ACCD_RP11}
 \end{minipage}
\end{figure}

\myparagraph{Piecewise-Constant ACC (\Accc{})}
The bounding models for \Accc{} are
taken from~\cite{DBLP:conf/cpsweek/Bu0S19}.
This model keeps the inter-vehicular distance by switching between the normal cruise behavior and the recovery behavior.
Namely, when the distance between the cars $x_i$ and $x_{i+1}$ becomes shorter than or equal to 4, the model can start the recovery behavior to make the distance longer.
When the distance between the cars $x_i$ and $x_{i+1}$ becomes larger than or equal to 4, the model can go back to the normal cruise behavior.
The accepting locations of \Accc{} happen to be unreachable, thus there will be no alerts. This is no problem for the scalability evaluation. %
In \Accc{}, the velocities of the cars at each location are constant.
\Accc{} contains three LHAs of dimensions 5, 10, and 15. \cref{figure:benchmark:ACCC} %
 is the LHA of dimension~5.

\myparagraph{Interval ACC  (\Acci{})}
 \Acci{} is a variant of the \Accc{} benchmark.
In the bounding model for \Acci{}, the velocities of the cars at each location are nondeterministically chosen from the given interval. It is shown in 
\cref{figure:collision:product}. 

\myparagraph{Diagonal ACC (\Accd{})}
The bounding models for \Accd{} are taken from~\cite{DBLP:conf/cpsweek/FrehseAABBCGGMM19}.
In \Accd{}, the velocities of the cars at each location are constrained by the following diagonal constraints (\ie{} constraints of the form $\variablei{i} - \variablei{j} \compOp n$, $n \in \setZ$): when recovering the distance between $x_i$ and $x_{i+1}$, 
we have $|\dot{x}_i - \variabledoti{i+1} - \varepsilon| < 1$, where $\varepsilon$ is the slow-down parameter; otherwise, we have $|\dot{x}_i - \variabledoti{i+1}| < 1$.
We used $\varepsilon = 0.9$ and $\varepsilon = 2.0$.
The safety specification in \AccdSafety{} is $x_i > x_{i-1}$ for each $i$.
\AccdSafety{} contains six LHAs of dimensions from 2 to 7.
The LHA of dimension~2 is shown in \cref{figure:benchmark:ACCD_RP11}.

\subsubsection{Benchmarks for RQ4}
\begin{figure}[tbp]
 \centering
 \begin{minipage}[b]{\ACMVersion{.48\linewidth}\LongVersion{.95\linewidth}}
  \begin{tikzpicture}[shorten >=1pt, scale=\ACMVersion{0.57}\LongVersion{1}, auto, every node/.style={transform shape}, node distance=0.5cm]
   \node[location] (loc1_1) [align=center]{
   $x \in [0,1]$,\\ $y \in [0,1]$,\\
   $v_x = \dot{x}$, $v_y = \dot{y}$\\
   $(\dot{v}_x, \dot{v}_y)^\top = \dot{\mathbf{v}}_2$};
   \node[location] (loc2_1) [right=of loc1_1] [align=center]{
   $x \in [1,2]$,\\ $y \in [0,1]$,\\
   $v_x = \dot{x}$, $v_y = \dot{y}$\\
   $(\dot{v}_x, \dot{v}_y)^\top = \dot{\mathbf{v}}_2$};
   \node[location] (loc3_1) [right=of loc2_1] [align=center]{
   $x \in [2,3]$,\\ $y \in [0,1]$,\\ 
   $v_x = \dot{x}$, $v_y = \dot{y}$\\
   $(\dot{v}_x, \dot{v}_y)^\top = \dot{\mathbf{v}}_0$};

   \node[location] (loc1_2) [above=of loc1_1] [align=center]{
   $x \in [0,1]$,\\ $y \in [1,2]$,\\ 
   $v_x = \dot{x}$, $v_y = \dot{y}$\\
   $(\dot{v}_x, \dot{v}_y)^\top = \dot{\mathbf{v}}_4$};
   \node[location] (loc2_2) [above=of loc2_1] [align=center]{
   $x \in [1,2]$,\\ $y \in [1,2]$,\\ 
   $v_x = \dot{x}$, $v_y = \dot{y}$\\
   $(\dot{v}_x, \dot{v}_y)^\top = \dot{\mathbf{v}}_3$};
   \node[location] (loc3_2) [above=of loc3_1] [align=center]{
   $x \in [2,3]$,\\ $y \in [1,2]$,\\ 
   $v_x = \dot{x}$, $v_y = \dot{y}$\\
   $(\dot{v}_x, \dot{v}_y)^\top = \dot{\mathbf{v}}_4$};

   \node[location] (loc1_3) [above=of loc1_2] [align=center]{
   $x \in [0,1]$,\\ $y \in [2,3]$,\\ 
   $v_x = \dot{x}$, $v_y = \dot{y}$\\
   $(\dot{v}_x, \dot{v}_y)^\top = \dot{\mathbf{v}}_3$};
   \node[location] (loc2_3) [above=of loc2_2] [align=center]{
   $x \in [1,2]$,\\ $y \in [2,3]$,\\ 
   $v_x = \dot{x}$, $v_y = \dot{y}$\\
   $(\dot{v}_x, \dot{v}_y)^\top = \dot{\mathbf{v}}_2$};
   \node[location] (loc3_3) [above=of loc3_2] [align=center]{
   $x \in [2,3]$,\\ $y \in [2,3]$,\\ 
   $v_x = \dot{x}$, $v_y = \dot{y}$\\
   $(\dot{v}_x, \dot{v}_y)^\top = \dot{\mathbf{v}}_3$};

   \draw[<-] (loc1_3) -- node[above, align=center] {
   $x \in [0,1]$,\\
   $y \in [2,3]$,\\
   $v_x \in [-1,1]$,\\
   $v_y \in [-1,1]$
   } ++(-3.25cm,0);

   \path[->] 
   (loc1_1) edge[bend left] node[left] {$y \geq 1$} (loc1_2)
   (loc1_2) edge[bend left] node[left] {$y \leq 1$} (loc1_1)
   (loc2_1) edge[bend left] node[left] {$y \geq 1$} (loc2_2)
   (loc2_2) edge[bend left] node[left] {$y \leq 1$} (loc2_1)
   (loc3_1) edge[bend left] node[left] {$y \geq 1$} (loc3_2)
   (loc3_2) edge[bend left] node[left] {$y \leq 1$} (loc3_1)

   (loc1_2) edge[bend left] node[left] {$y \geq 2$} (loc1_3)
   (loc1_3) edge[bend left] node[left] {$y \leq 2$} (loc1_2)
   (loc2_2) edge[bend left] node[left] {$y \geq 2$} (loc2_3)
   (loc2_3) edge[bend left] node[left] {$y \leq 2$} (loc2_2)
   (loc3_2) edge[bend left] node[left] {$y \geq 2$} (loc3_3)
   (loc3_3) edge[bend left] node[left] {$y \leq 2$} (loc3_2)

   (loc1_1) edge[bend left] node[above] {$x \geq 1$} (loc2_1)
   (loc2_1) edge[bend left] node[below] {$x \leq 1$} (loc1_1)
   (loc1_2) edge[bend left] node[above] {$x \geq 1$} (loc2_2)
   (loc2_2) edge[bend left] node[below] {$x \leq 1$} (loc1_2)
   (loc1_3) edge[bend left] node[above] {$x \geq 1$} (loc2_3)
   (loc2_3) edge[bend left] node[below] {$x \leq 1$} (loc1_3)

   (loc2_1) edge[bend left] node[above] {$x \geq 2$} (loc3_1)
   (loc3_1) edge[bend left] node[below] {$x \leq 2$} (loc2_1)
   (loc2_2) edge[bend left] node[above] {$x \geq 2$} (loc3_2)
   (loc3_2) edge[bend left] node[below] {$x \leq 2$} (loc2_2)
   (loc2_3) edge[bend left] node[above] {$x \geq 2$} (loc3_3)
   (loc3_3) edge[bend left] node[below] {$x \leq 2$} (loc2_3)
   ;
  \end{tikzpicture}
                                                                                                              \caption[dummy caption text]{The affine hybrid automaton for the original model in \NAV{}, where $A = \left(\begin{array}{c c} -1.2& 0.1\\ 0.1& -1.2\end{array}\right)$, $v_{d,i} = (\sin(i\pi/4), \cos(i\pi/4))^\top$, $\dot{\mathbf{v}}_i = A \bigl((v_x, v_y)^\top - v_{d,i}\bigr)$ for $i \in \{2,3,4\}$, and $\dot{\mathbf{v}}_0 = A (v_x, v_y)^\top$}%
  \label{figure:benchmark:NAV}
 \end{minipage}
\ACMVersion{\hfill}%
\LongVersion{%

}
 \begin{minipage}[b]{\ACMVersion{.5\linewidth}\LongVersion{.95\linewidth}}
  \centering
  \begin{tikzpicture}[shorten >=1pt, scale=\ACMVersion{0.6}\LongVersion{1}, auto, every node/.style={transform shape}, node distance=1.3cm]
   \node[location] (loc1) [align=center]{
   $\dot{x}_1 = \mathrm{ON}_1$,\\
   $\dot{x}_2 = \mathrm{OFF}_2$,\\
   $x_1 \in [0,100]$,\\
   $x_2 \in [0,100]$};
   \node[location] (loc0) [right=of loc1] [align=center]{
   $\dot{x}_1 = \mathrm{OFF}_1$,\\
   $\dot{x}_2 = \mathrm{OFF}_2$,\\
   $x_1 \in [80,100]$,\\
   $x_2 \in [80,100]$};
   \node[location] (loc2) [right=of loc0] [align=center]{
   $\dot{x}_1 = \mathrm{OFF}_1$,\\
   $\dot{x}_2 = \mathrm{ON}_2$,\\
   $x_1 \in [0,100]$,\\
   $x_2 \in [0,100]$};

   \draw[<-] (loc1) -- node[above, align=center] {
   $x_1 = 0$,\\
   $x_2 = 50$
   } ++(-2.5cm,0);

   \path[->] 
   (loc1) edge[bend left=40] node[above] {$(x_1 \geq 100 \land x_2 \leq 80) \lor (x_1 \geq 20 \land x_2 \leq 0)$} (loc2)
   (loc1) edge[bend left] node[below] {$x_1 \geq 100$} (loc0)
   (loc0) edge[bend left] node[above] {$x_1 \leq 80$} (loc1)
   (loc2) edge[bend left] node[above] {$x_2 \geq 100$} (loc0)
   (loc0) edge[bend left] node[below] {$x_2 \leq 80$} (loc2)
   (loc2) edge[bend left=40] node[below] {$(x_2 \geq 100 \land x_1 \leq 80) \lor (x_2 \geq 20 \land x_1 \leq 0)$} (loc1)
   ;
  \end{tikzpicture}
  \caption{The affine hybrid automaton for the original model in \SharedGasBurner{}, where $h = 2$, $a = 0.01$, $b = 0.005$, $\mathrm{ON}_1 = h - ax_1 + bx_2$, $\mathrm{ON}_2 = h - ax_2 + bx_1$, $\mathrm{OFF}_1 = - ax_1 + bx_2$, and $\mathrm{OFF}_2 = - ax_2 + bx_1$.}%
  \label{figure:benchmark:SharedGasBurner}
 \end{minipage}
\end{figure}
In the precision experiments to answer RQ4, we used the robot navigation benchmark (\NAV), its variant with unsafe behavior (\UnsafeNAV), and the shared gas-burners benchmark (\SharedGasBurner{}).
The original system models are \emph{affine} hybrid automata~\cite{DHR05}, and 
the bounding models are constructed by the \emph{projection} overapproximation in~\cite[Section~3.2]{Frehse2008}.
We use \cref{def:Mlnotvarphi} to construct $\A_{\lnot\varphi}$.
The input logs $w$ were generated by a simulation of a Simulink implementation of the original system model.

\myparagraph{Navigation (\NAV{})}
The original model  for \NAV{} is taken from~\cite{FI04}.
To obtain an LHA $\A$, we used the projection overapproximation in~\cite{Frehse2008} with an additional invariant $v_x \in [-1,1], v_y \in [-1,1]$. 
We constructed the bounding model $\A_{\lnot\varphi}$ by the construction in \cref{def:Mlnotvarphi}, where
the monitored safety specification $\varphi$ is $v_x \in [-1,1] \land v_y \in [-1,1]$. 
We note that the bounding model $\A_{\neg \varphi}$ does not contain the invariant $v_x \in [-1,1], v_y \in [-1,1]$, and the accepting locations for the safety violation are reachable.
The affine hybrid automaton in \cref{figure:benchmark:NAV} represents the original model of \NAV{}. 

\myparagraph{Unsafe Navigation (\UnsafeNAV) Benchmark}
 \UnsafeNAV{} is a variant of the \NAV{} benchmark with potential safety violation.
 The differences are as follows: 
 \begin{inparaenum}
  \item All the acceleration $\dot{\mathbf{v}}_i$ are multiplied by~1.1;
  \item The flow in the center location is $\dot{\mathbf{v}}_2$ instead of~$\dot{\mathbf{v}}_3$.
 \end{inparaenum}

\myparagraph{Shared Gas-Burner (\SharedGasBurner) Benchmark} 
The original model  for \SharedGasBurner{} is taken from~\cite{DHR05}.
In order to obtain an LHA $\A$, we used the projection overapproximation in~\cite{Frehse2008}.
We constructed the bounding model $\A_{\lnot\varphi}$ by the construction in \cref{def:Mlnotvarphi}, where
the monitored safety specification $\varphi$ is $x_1 \in [0,100.1) \land x_2 \in [0,100.1)$. %
We note that all the upper bounds $100$ of the invariants in bounding model $\A_{\neg \varphi}$ are replaced with $150$, and thus, the accepting locations for the safety violation are reachable.
The affine hybrid automaton in \cref{figure:benchmark:SharedGasBurner} represents the original model of \SharedGasBurner{}.
\subsection{Experiments}\label{subsec:experiments}

\subsubsection{Scalability experiments}\label{subsubsec:scalability_experiments}
For the reachability analysis in the procedure presented in \cref{section:algorithm1}, we used \phaverlite{} 0.2.1.
\phaverlite{} relies on PPLite~\cite{BZ19} to compute symbolic states.
We implemented an OCaml program and a Python script to construct a \phaverlite{} model from an LHA $\A$ and a timed quantitative word~$\word$. %
For the procedure presented in \cref{section:algorithm2}, 
we implemented \masakiTool{} in C++ with Parma Polyhedra Library (PPL)~\cite{BHZ08} and compiled using GCC 7.4.0.
In both \phaverlite{} and \masakiTool{}, closed convex polyhedra are used to analyze the reachability~\cite{BZ19}.

Since the difficulty of the $\Lmon$ membership problem depends on the given timed quantitative word $\word$,
we randomly generated 30 logs for each experiment setting and measured the average of the execution time.
The sampling interval of $\word$ is from 1 to 5 seconds, uniformly distributed.
The timeout is 10 minutes.
We conducted the experiments on an Amazon EC2 c4.large instance (2.9\,GHz Intel Xeon E5-2666 v3, 2 vCPUs, and 3.75\,GiB RAM) that runs Ubuntu 18.04 LTS (64\,bit).
\cref{table:result:accc,table:result:acci,table:result:accd} summarize the experiment results.

\begin{table}[tbp]
 \begin{minipage}[c]{.6\linewidth}
  \caption{Experiment result on \Accc{} [sec.]}%
  \label{table:result:accc}
  \begin{minipage}[c]{.47\linewidth}
   \centering
   \scriptsize
   \begin{tabular}[t]{c c c c}
\toprule
dim.&len.&\phaverlite{}&\masakiTool{}\\
\midrule
5&10&500.02&\tbcolor0.46\\
5&25&540.02&\tbcolor0.86\\
5&50&480.06&\tbcolor0.78\\
5&75&580.01&\tbcolor0.79\\
5&100&520.07&\tbcolor1.21\\
5&200&560.06&\tbcolor1.00\\
5&300&540.13&\tbcolor1.37\\
5&400&\TO{}&\tbcolor1.77\\
\bottomrule
\end{tabular}

  \end{minipage}
  \ACMVersion{\hfill}\LongVersion{%

  }
  \begin{minipage}[c]{.47\linewidth}
   \centering
   \scriptsize
   \begin{tabular}[t]{c c c c}
\toprule
dim.&len.&\phaverlite{}&\masakiTool{}\\
\midrule
5&500&480.41&\tbcolor2.27\\
5&600&500.40&\tbcolor2.12\\
5&700&520.40&\tbcolor2.37\\
5&800&540.33&\tbcolor2.58\\
5&900&580.12&\tbcolor2.76\\
5&1000&560.26&\tbcolor3.26\\
10&10&267.60&\tbcolor204.79\\
15&10&\TO{}&\TO{}\\
\bottomrule
\end{tabular}

  \end{minipage}
 \end{minipage}
 \hfill
 \begin{minipage}[c]{.39\linewidth}
  \caption{Experiment result on \Acci{} [sec.]}%
  \label{table:result:acci}
   \centering
   \scriptsize
   \begin{tabular}[t]{c c c}
\toprule
len.&\phaverlite{}&\masakiTool{}\\
\midrule
10,000&13.42&\tbcolor2.18\\
20,000&40.29&\tbcolor4.29\\
30,000&83.39&\tbcolor6.43\\
40,000&141.55&\tbcolor8.56\\
50,000&225.23&\tbcolor10.76\\
60,000&116.67&\tbcolor12.86\\
70,000&26.60&\tbcolor15.00\\
80,000&227.29&\tbcolor17.14\\
90,000&259.17&\tbcolor19.28\\
10,0000&227.12&\tbcolor21.45\\
\bottomrule
\end{tabular}

 \end{minipage}
\end{table}
\begin{table}[tbp]
 \caption{Experiment result on \Accd{} [sec.]}%
 \label{table:result:accd}
 \begin{minipage}{.49\linewidth}
  \centering
  \LongVersion{\setlength{\tabcolsep}{2pt}}
  \scalebox{\ACMVersion{0.71}\LongVersion{.65}}{
  \begin{tabular}[t]{c c c c c c}
   \toprule
   &  & \multicolumn{2}{c}{$\varepsilon = 0.9$}& \multicolumn{2}{c}{$\varepsilon = 2.0$} \\ 
   \cmidrule{3-6}
   dim.&len. & \phaverlite{} &\masakiTool{} & \phaverlite{} &\masakiTool{}\\
   \midrule
   2&25&0.03&\tbcolor0.00&0.03&\tbcolor0.00\\
2&50&0.04&\tbcolor0.00&0.04&\tbcolor0.00\\
2&75&0.05&\tbcolor0.01&0.05&\tbcolor0.00\\
2&100&0.06&\tbcolor0.01&0.06&\tbcolor0.01\\
2&200&0.12&\tbcolor0.02&0.11&\tbcolor0.01\\
2&300&0.16&\tbcolor0.02&0.16&\tbcolor0.02\\
2&400&0.21&\tbcolor0.03&0.20&\tbcolor0.02\\
2&500&0.27&\tbcolor0.04&0.26&\tbcolor0.03\\
2&600&0.32&\tbcolor0.04&0.30&\tbcolor0.03\\
2&700&0.35&\tbcolor0.05&0.36&\tbcolor0.04\\
2&800&0.40&\tbcolor0.05&0.40&\tbcolor0.05\\
2&900&0.46&\tbcolor0.06&0.45&\tbcolor0.05\\
2&1000&0.51&\tbcolor0.07&0.50&\tbcolor0.06\\
\midrule
3&10&0.05&\tbcolor0.02&0.04&\tbcolor0.01\\
3&25&0.08&\tbcolor0.03&0.08&\tbcolor0.02\\
3&50&0.14&\tbcolor0.04&0.12&\tbcolor0.03\\
3&75&0.19&\tbcolor0.05&0.18&\tbcolor0.03\\
3&100&0.23&\tbcolor0.05&0.22&\tbcolor0.04\\
3&200&0.44&\tbcolor0.07&0.42&\tbcolor0.05\\
3&300&0.65&\tbcolor0.09&0.62&\tbcolor0.06\\
3&400&0.84&\tbcolor0.12&0.82&\tbcolor0.09\\
3&500&1.07&\tbcolor0.15&1.00&\tbcolor0.10\\
3&600&1.25&\tbcolor0.16&1.24&\tbcolor0.11\\
3&700&1.46&\tbcolor0.16&1.42&\tbcolor0.13\\
3&800&1.73&\tbcolor0.20&1.61&\tbcolor0.13\\
3&900&1.84&\tbcolor0.19&1.83&\tbcolor0.16\\
3&1000&2.04&\tbcolor0.20&2.00&\tbcolor0.17\\
\midrule
4&10&\tbcolor0.28&0.35&\tbcolor0.22&0.42\\
4&25&0.46&0.46&0.34&\tbcolor0.22\\
4&50&0.66&\tbcolor0.42&0.59&\tbcolor0.38\\
4&75&0.92&\tbcolor0.47&0.82&\tbcolor0.36\\
4&100&1.21&\tbcolor0.64&1.05&\tbcolor0.35\\
4&200&2.13&\tbcolor0.82&2.02&\tbcolor0.87\\
\bottomrule
  \end{tabular}}
 \end{minipage}
 \hfill
 \begin{minipage}{.49\linewidth}
  \centering
  \LongVersion{\setlength{\tabcolsep}{2pt}}
  \scalebox{\ACMVersion{0.71}\LongVersion{.65}}{
  \begin{tabular}[t]{c c c c c c}
   \toprule
   &  & \multicolumn{2}{c}{$\varepsilon = 0.9$}& \multicolumn{2}{c}{$\varepsilon = 2.0$} \\ 
   \cmidrule{3-6}
   dim.&len. & \phaverlite{} &\masakiTool{} & \phaverlite{} &\masakiTool{}\\
   \midrule
   4&300&2.98&\tbcolor0.74&2.81&\tbcolor0.40\\
4&400&3.98&\tbcolor0.92&3.76&\tbcolor0.62\\
4&500&4.79&\tbcolor0.85&4.69&\tbcolor0.65\\
4&600&5.71&\tbcolor0.86&5.63&\tbcolor0.68\\
4&700&6.60&\tbcolor0.81&6.53&\tbcolor0.93\\
4&800&7.67&\tbcolor1.15&7.28&\tbcolor0.97\\
4&900&8.39&\tbcolor1.04&\tbcolor0.02&0.87\\
4&1000&9.26&\tbcolor1.18&9.15&\tbcolor0.89\\
\midrule
5&10&\tbcolor2.65&7.51&\tbcolor22.38&29.33\\
5&25&\tbcolor3.84&10.13&\tbcolor3.05&9.22\\
5&50&\tbcolor5.25&9.68&\tbcolor12.93&31.06\\
5&75&\tbcolor7.13&14.41&\tbcolor18.32&29.93\\
5&100&\tbcolor8.14&12.87&\tbcolor7.58&21.13\\
5&200&11.39&\tbcolor8.94&\tbcolor29.52&36.17\\
5&300&17.33&\tbcolor15.03&15.38&\tbcolor14.56\\
5&400&20.44&\tbcolor9.81&19.47&\tbcolor15.34\\
5&500&24.87&\tbcolor11.72&36.83&\tbcolor34.12\\
5&600&28.78&\tbcolor8.79&28.47&\tbcolor16.37\\
5&700&35.14&\tbcolor17.05&32.86&\tbcolor16.10\\
5&800&37.76&\tbcolor12.41&41.97&\tbcolor39.07\\
5&900&42.22&\tbcolor9.83&48.26&\tbcolor46.68\\
5&1000&47.11&\tbcolor11.55&57.49&\tbcolor54.44\\
\midrule
6&10&\tbcolor47.42&221.55&\tbcolor80.26&428.66\\
6&25&\tbcolor41.18&165.63&\tbcolor76.40&510.55\\
6&50&\tbcolor65.36&243.97&\tbcolor120.77&518.01\\
6&75&\tbcolor58.21&173.46&\tbcolor125.61&480.19\\
6&100&\tbcolor87.47&209.53&\tbcolor131.88&421.59\\
\midrule
7&10&\tbcolor525.07&560.69&\TO{}&\tbcolor594.05\\
7&25&\tbcolor489.35&559.56&\TO{}&\tbcolor526.14\\
7&50&\tbcolor514.10&562.68&\TO{}&\tbcolor566.01\\
7&75&\TO{}&\tbcolor588.23&\TO{}&\tbcolor583.49\\
7&100&\TO{}&\tbcolor577.29&\TO{}&\tbcolor578.39\\
\bottomrule
  \end{tabular}}
 \end{minipage}
\end{table}

\myparagraph{RQ1: worthwhileness of a dedicated implementation}
In \cref{table:result:accc,table:result:acci,table:result:accd}, we observe that \masakiTool{} tends to outperform \phaverlite{}.
Especially, in \cref{table:result:accc}, %
we observe that for \Accc{}, %
\masakiTool{} performed drastically faster than \phaverlite{} for dimension~5. %
This is because \phaverlite{} is not a specific tool for the  $\Lmon$ membership problem but a tool for reachability analysis in general.
Thus, despite the engineering cost for the implementation that is not small, a dedicated solver, \ie{} \masakiTool{}, is worthwhile.

However, in \cref{table:result:accd},
we also observe that in \Accd{}, 
when the dimension of the LHA is relatively large and the timed quantitative word is not too long,
\phaverlite{} often outperformed \masakiTool{}.
This is because of the %
 optimized reachability analysis algorithm in \phaverlite{}.
Optimization of the reachability analysis algorithm in \masakiTool{}, \eg{} utilizing the techniques in~\cite{BZ19,CAF11}, is future work.

\begin{figure}[tbp]
 \begin{minipage}[b]{0.45\linewidth}
  \centering 
  \scalebox{0.37}{\begin{tikzpicture}[gnuplot]
\tikzset{every node/.append style={font={\fontsize{22.0pt}{26.4pt}\selectfont}}}
\path (0.000,0.000) rectangle (12.500,8.750);
\gpcolor{color=gp lt color border}
\gpsetlinetype{gp lt border}
\gpsetdashtype{gp dt solid}
\gpsetlinewidth{1.00}
\draw[gp path] (2.905,2.169)--(3.085,2.169);
\draw[gp path] (11.284,2.169)--(11.104,2.169);
\node[gp node right,font={\fontsize{32.0pt}{38.4pt}\selectfont}] at (2.500,2.169) {$0$};
\draw[gp path] (2.905,3.012)--(3.085,3.012);
\draw[gp path] (11.284,3.012)--(11.104,3.012);
\node[gp node right,font={\fontsize{32.0pt}{38.4pt}\selectfont}] at (2.500,3.012) {$0.5$};
\draw[gp path] (2.905,3.855)--(3.085,3.855);
\draw[gp path] (11.284,3.855)--(11.104,3.855);
\node[gp node right,font={\fontsize{32.0pt}{38.4pt}\selectfont}] at (2.500,3.855) {$1$};
\draw[gp path] (2.905,4.698)--(3.085,4.698);
\draw[gp path] (11.284,4.698)--(11.104,4.698);
\node[gp node right,font={\fontsize{32.0pt}{38.4pt}\selectfont}] at (2.500,4.698) {$1.5$};
\draw[gp path] (2.905,5.542)--(3.085,5.542);
\draw[gp path] (11.284,5.542)--(11.104,5.542);
\node[gp node right,font={\fontsize{32.0pt}{38.4pt}\selectfont}] at (2.500,5.542) {$2$};
\draw[gp path] (2.905,6.385)--(3.085,6.385);
\draw[gp path] (11.284,6.385)--(11.104,6.385);
\node[gp node right,font={\fontsize{32.0pt}{38.4pt}\selectfont}] at (2.500,6.385) {$2.5$};
\draw[gp path] (2.905,7.228)--(3.085,7.228);
\draw[gp path] (11.284,7.228)--(11.104,7.228);
\node[gp node right,font={\fontsize{32.0pt}{38.4pt}\selectfont}] at (2.500,7.228) {$3$};
\draw[gp path] (2.905,8.071)--(3.085,8.071);
\draw[gp path] (11.284,8.071)--(11.104,8.071);
\node[gp node right,font={\fontsize{32.0pt}{38.4pt}\selectfont}] at (2.500,8.071) {$3.5$};
\draw[gp path] (2.905,2.169)--(2.905,2.349);
\draw[gp path] (2.905,8.071)--(2.905,7.891);
\node[gp node center,font={\fontsize{32.0pt}{38.4pt}\selectfont}] at (2.905,1.491) {$0$};
\draw[gp path] (3.743,2.169)--(3.743,2.349);
\draw[gp path] (3.743,8.071)--(3.743,7.891);
\node[gp node center,font={\fontsize{32.0pt}{38.4pt}\selectfont}] at (3.743,1.491) {$1$};
\draw[gp path] (4.581,2.169)--(4.581,2.349);
\draw[gp path] (4.581,8.071)--(4.581,7.891);
\node[gp node center,font={\fontsize{32.0pt}{38.4pt}\selectfont}] at (4.581,1.491) {$2$};
\draw[gp path] (5.419,2.169)--(5.419,2.349);
\draw[gp path] (5.419,8.071)--(5.419,7.891);
\node[gp node center,font={\fontsize{32.0pt}{38.4pt}\selectfont}] at (5.419,1.491) {$3$};
\draw[gp path] (6.257,2.169)--(6.257,2.349);
\draw[gp path] (6.257,8.071)--(6.257,7.891);
\node[gp node center,font={\fontsize{32.0pt}{38.4pt}\selectfont}] at (6.257,1.491) {$4$};
\draw[gp path] (7.095,2.169)--(7.095,2.349);
\draw[gp path] (7.095,8.071)--(7.095,7.891);
\node[gp node center,font={\fontsize{32.0pt}{38.4pt}\selectfont}] at (7.095,1.491) {$5$};
\draw[gp path] (7.932,2.169)--(7.932,2.349);
\draw[gp path] (7.932,8.071)--(7.932,7.891);
\node[gp node center,font={\fontsize{32.0pt}{38.4pt}\selectfont}] at (7.932,1.491) {$6$};
\draw[gp path] (8.770,2.169)--(8.770,2.349);
\draw[gp path] (8.770,8.071)--(8.770,7.891);
\node[gp node center,font={\fontsize{32.0pt}{38.4pt}\selectfont}] at (8.770,1.491) {$7$};
\draw[gp path] (9.608,2.169)--(9.608,2.349);
\draw[gp path] (9.608,8.071)--(9.608,7.891);
\node[gp node center,font={\fontsize{32.0pt}{38.4pt}\selectfont}] at (9.608,1.491) {$8$};
\draw[gp path] (10.446,2.169)--(10.446,2.349);
\draw[gp path] (10.446,8.071)--(10.446,7.891);
\node[gp node center,font={\fontsize{32.0pt}{38.4pt}\selectfont}] at (10.446,1.491) {$9$};
\draw[gp path] (11.284,2.169)--(11.284,2.349);
\draw[gp path] (11.284,8.071)--(11.284,7.891);
\node[gp node center,font={\fontsize{32.0pt}{38.4pt}\selectfont}] at (11.284,1.491) {$10$};
\draw[gp path] (2.905,8.071)--(2.905,2.169)--(11.284,2.169)--(11.284,8.071)--cycle;
\node[gp node center,rotate=-270,font={\fontsize{30.0pt}{36.0pt}\selectfont}] at (0.642,5.120) {Exec. time [sec.]};
\node[gp node center,font={\fontsize{30.0pt}{36.0pt}\selectfont}] at (7.094,0.474) {Length [$ \times 100$]};
\node[gp node right,font={\fontsize{22.0pt}{26.4pt}\selectfont}] at (8.269,7.552) {\Accc{}, dim 5};
\gpcolor{rgb color={0.580,0.000,0.827}}
\gpsetlinewidth{10.00}
\draw[gp path] (8.674,7.552)--(10.474,7.552);
\draw[gp path] (2.989,2.945)--(3.114,3.619)--(3.324,3.484)--(3.533,3.501)--(3.743,4.209)%
  --(4.581,3.855)--(5.419,4.479)--(6.257,5.154)--(7.095,5.997)--(7.932,5.744)--(8.770,6.165)%
  --(9.608,6.520)--(10.446,6.823)--(11.284,7.666);
\gpsetpointsize{12.00}
\gppoint{gp mark 7}{(2.989,2.945)}
\gppoint{gp mark 7}{(3.114,3.619)}
\gppoint{gp mark 7}{(3.324,3.484)}
\gppoint{gp mark 7}{(3.533,3.501)}
\gppoint{gp mark 7}{(3.743,4.209)}
\gppoint{gp mark 7}{(4.581,3.855)}
\gppoint{gp mark 7}{(5.419,4.479)}
\gppoint{gp mark 7}{(6.257,5.154)}
\gppoint{gp mark 7}{(7.095,5.997)}
\gppoint{gp mark 7}{(7.932,5.744)}
\gppoint{gp mark 7}{(8.770,6.165)}
\gppoint{gp mark 7}{(9.608,6.520)}
\gppoint{gp mark 7}{(10.446,6.823)}
\gppoint{gp mark 7}{(11.284,7.666)}
\gppoint{gp mark 7}{(9.574,7.552)}
\gpcolor{color=gp lt color border}
\gpsetlinewidth{1.00}
\draw[gp path] (2.905,8.071)--(2.905,2.169)--(11.284,2.169)--(11.284,8.071)--cycle;
\gpdefrectangularnode{gp plot 1}{\pgfpoint{2.905cm}{2.169cm}}{\pgfpoint{11.284cm}{8.071cm}}
\end{tikzpicture}}
 \end{minipage}
 \hfill
 \begin{minipage}[b]{0.45\linewidth}
  \centering 
  \scalebox{0.37}{\begin{tikzpicture}[gnuplot]
\tikzset{every node/.append style={font={\fontsize{22.0pt}{26.4pt}\selectfont}}}
\path (0.000,0.000) rectangle (12.500,8.750);
\gpcolor{color=gp lt color border}
\gpsetlinetype{gp lt border}
\gpsetdashtype{gp dt solid}
\gpsetlinewidth{1.00}
\draw[gp path] (2.500,2.169)--(2.680,2.169);
\draw[gp path] (11.284,2.169)--(11.104,2.169);
\node[gp node right,font={\fontsize{32.0pt}{38.4pt}\selectfont}] at (2.095,2.169) {$0$};
\draw[gp path] (2.500,3.153)--(2.680,3.153);
\draw[gp path] (11.284,3.153)--(11.104,3.153);
\node[gp node right,font={\fontsize{32.0pt}{38.4pt}\selectfont}] at (2.095,3.153) {$4$};
\draw[gp path] (2.500,4.136)--(2.680,4.136);
\draw[gp path] (11.284,4.136)--(11.104,4.136);
\node[gp node right,font={\fontsize{32.0pt}{38.4pt}\selectfont}] at (2.095,4.136) {$8$};
\draw[gp path] (2.500,5.120)--(2.680,5.120);
\draw[gp path] (11.284,5.120)--(11.104,5.120);
\node[gp node right,font={\fontsize{32.0pt}{38.4pt}\selectfont}] at (2.095,5.120) {$12$};
\draw[gp path] (2.500,6.104)--(2.680,6.104);
\draw[gp path] (11.284,6.104)--(11.104,6.104);
\node[gp node right,font={\fontsize{32.0pt}{38.4pt}\selectfont}] at (2.095,6.104) {$16$};
\draw[gp path] (2.500,7.087)--(2.680,7.087);
\draw[gp path] (11.284,7.087)--(11.104,7.087);
\node[gp node right,font={\fontsize{32.0pt}{38.4pt}\selectfont}] at (2.095,7.087) {$20$};
\draw[gp path] (2.500,8.071)--(2.680,8.071);
\draw[gp path] (11.284,8.071)--(11.104,8.071);
\node[gp node right,font={\fontsize{32.0pt}{38.4pt}\selectfont}] at (2.095,8.071) {$24$};
\draw[gp path] (2.500,2.169)--(2.500,2.349);
\draw[gp path] (2.500,8.071)--(2.500,7.891);
\node[gp node center,font={\fontsize{32.0pt}{38.4pt}\selectfont}] at (2.500,1.491) {$0$};
\draw[gp path] (4.257,2.169)--(4.257,2.349);
\draw[gp path] (4.257,8.071)--(4.257,7.891);
\node[gp node center,font={\fontsize{32.0pt}{38.4pt}\selectfont}] at (4.257,1.491) {$20$};
\draw[gp path] (6.014,2.169)--(6.014,2.349);
\draw[gp path] (6.014,8.071)--(6.014,7.891);
\node[gp node center,font={\fontsize{32.0pt}{38.4pt}\selectfont}] at (6.014,1.491) {$40$};
\draw[gp path] (7.770,2.169)--(7.770,2.349);
\draw[gp path] (7.770,8.071)--(7.770,7.891);
\node[gp node center,font={\fontsize{32.0pt}{38.4pt}\selectfont}] at (7.770,1.491) {$60$};
\draw[gp path] (9.527,2.169)--(9.527,2.349);
\draw[gp path] (9.527,8.071)--(9.527,7.891);
\node[gp node center,font={\fontsize{32.0pt}{38.4pt}\selectfont}] at (9.527,1.491) {$80$};
\draw[gp path] (11.284,2.169)--(11.284,2.349);
\draw[gp path] (11.284,8.071)--(11.284,7.891);
\node[gp node center,font={\fontsize{32.0pt}{38.4pt}\selectfont}] at (11.284,1.491) {$100$};
\draw[gp path] (2.500,8.071)--(2.500,2.169)--(11.284,2.169)--(11.284,8.071)--cycle;
\node[gp node center,rotate=-270,font={\fontsize{30.0pt}{36.0pt}\selectfont}] at (0.642,5.120) {Exec. time [sec.]};
\node[gp node center,font={\fontsize{30.0pt}{36.0pt}\selectfont}] at (6.892,0.474) {Length [$\times 1,000$]};
\node[gp node right,font={\fontsize{30.0pt}{36.0pt}\selectfont}] at (7.387,7.429) {\Acci{}};
\gpcolor{rgb color={0.580,0.000,0.827}}
\gpsetlinewidth{10.00}
\draw[gp path] (7.939,7.429)--(10.327,7.429);
\draw[gp path] (3.378,2.705)--(4.257,3.224)--(5.135,3.750)--(6.014,4.274)--(6.892,4.815)%
  --(7.770,5.331)--(8.649,5.858)--(9.527,6.384)--(10.406,6.910)--(11.284,7.444);
\gpsetpointsize{12.00}
\gppoint{gp mark 7}{(3.378,2.705)}
\gppoint{gp mark 7}{(4.257,3.224)}
\gppoint{gp mark 7}{(5.135,3.750)}
\gppoint{gp mark 7}{(6.014,4.274)}
\gppoint{gp mark 7}{(6.892,4.815)}
\gppoint{gp mark 7}{(7.770,5.331)}
\gppoint{gp mark 7}{(8.649,5.858)}
\gppoint{gp mark 7}{(9.527,6.384)}
\gppoint{gp mark 7}{(10.406,6.910)}
\gppoint{gp mark 7}{(11.284,7.444)}
\gppoint{gp mark 7}{(9.133,7.429)}
\gpcolor{color=gp lt color border}
\gpsetlinewidth{1.00}
\draw[gp path] (2.500,8.071)--(2.500,2.169)--(11.284,2.169)--(11.284,8.071)--cycle;
\gpdefrectangularnode{gp plot 1}{\pgfpoint{2.500cm}{2.169cm}}{\pgfpoint{11.284cm}{8.071cm}}
\end{tikzpicture}}
 \end{minipage}
  \caption{Execution time of \masakiTool{} for \Accc{} dimension 5 (left) and \Acci{} (right)%
}%
 \label{figure:result:rp11_accc_acci:dimension}
\end{figure}
\begin{figure}[tbp]
 \begin{subfigure}[b]{0.32\linewidth}
  \centering
  \scalebox{0.37}{\begin{tikzpicture}[gnuplot]
\tikzset{every node/.append style={font={\fontsize{24.0pt}{24.0pt}\selectfont}}}
\path (0.000,0.000) rectangle (12.500,8.750);
\gpcolor{color=gp lt color border}
\gpsetlinetype{gp lt border}
\gpsetdashtype{gp dt solid}
\gpsetlinewidth{1.00}
\draw[gp path] (2.705,1.772)--(2.885,1.772);
\draw[gp path] (11.506,1.772)--(11.326,1.772);
\node[gp node right,font={\fontsize{24.0pt}{38.4pt}\selectfont}] at (2.374,1.772) {$0$};
\draw[gp path] (2.705,2.835)--(2.885,2.835);
\draw[gp path] (11.506,2.835)--(11.326,2.835);
\node[gp node right,font={\fontsize{24.0pt}{38.4pt}\selectfont}] at (2.374,2.835) {$0.04$};
\draw[gp path] (2.705,3.898)--(2.885,3.898);
\draw[gp path] (11.506,3.898)--(11.326,3.898);
\node[gp node right,font={\fontsize{24.0pt}{38.4pt}\selectfont}] at (2.374,3.898) {$0.08$};
\draw[gp path] (2.705,4.961)--(2.885,4.961);
\draw[gp path] (11.506,4.961)--(11.326,4.961);
\node[gp node right,font={\fontsize{24.0pt}{38.4pt}\selectfont}] at (2.374,4.961) {$0.12$};
\draw[gp path] (2.705,6.024)--(2.885,6.024);
\draw[gp path] (11.506,6.024)--(11.326,6.024);
\node[gp node right,font={\fontsize{24.0pt}{38.4pt}\selectfont}] at (2.374,6.024) {$0.16$};
\draw[gp path] (2.705,7.087)--(2.885,7.087);
\draw[gp path] (11.506,7.087)--(11.326,7.087);
\node[gp node right,font={\fontsize{24.0pt}{38.4pt}\selectfont}] at (2.374,7.087) {$0.2$};
\draw[gp path] (2.705,1.772)--(2.705,1.952);
\draw[gp path] (2.705,7.087)--(2.705,6.907);
\node[gp node center,font={\fontsize{24.0pt}{38.4pt}\selectfont}] at (2.705,1.218) {$0$};
\draw[gp path] (3.585,1.772)--(3.585,1.952);
\draw[gp path] (3.585,7.087)--(3.585,6.907);
\node[gp node center,font={\fontsize{24.0pt}{38.4pt}\selectfont}] at (3.585,1.218) {$1$};
\draw[gp path] (4.465,1.772)--(4.465,1.952);
\draw[gp path] (4.465,7.087)--(4.465,6.907);
\node[gp node center,font={\fontsize{24.0pt}{38.4pt}\selectfont}] at (4.465,1.218) {$2$};
\draw[gp path] (5.345,1.772)--(5.345,1.952);
\draw[gp path] (5.345,7.087)--(5.345,6.907);
\node[gp node center,font={\fontsize{24.0pt}{38.4pt}\selectfont}] at (5.345,1.218) {$3$};
\draw[gp path] (6.225,1.772)--(6.225,1.952);
\draw[gp path] (6.225,7.087)--(6.225,6.907);
\node[gp node center,font={\fontsize{24.0pt}{38.4pt}\selectfont}] at (6.225,1.218) {$4$};
\draw[gp path] (7.106,1.772)--(7.106,1.952);
\draw[gp path] (7.106,7.087)--(7.106,6.907);
\node[gp node center,font={\fontsize{24.0pt}{38.4pt}\selectfont}] at (7.106,1.218) {$5$};
\draw[gp path] (7.986,1.772)--(7.986,1.952);
\draw[gp path] (7.986,7.087)--(7.986,6.907);
\node[gp node center,font={\fontsize{24.0pt}{38.4pt}\selectfont}] at (7.986,1.218) {$6$};
\draw[gp path] (8.866,1.772)--(8.866,1.952);
\draw[gp path] (8.866,7.087)--(8.866,6.907);
\node[gp node center,font={\fontsize{24.0pt}{38.4pt}\selectfont}] at (8.866,1.218) {$7$};
\draw[gp path] (9.746,1.772)--(9.746,1.952);
\draw[gp path] (9.746,7.087)--(9.746,6.907);
\node[gp node center,font={\fontsize{24.0pt}{38.4pt}\selectfont}] at (9.746,1.218) {$8$};
\draw[gp path] (10.626,1.772)--(10.626,1.952);
\draw[gp path] (10.626,7.087)--(10.626,6.907);
\node[gp node center,font={\fontsize{24.0pt}{38.4pt}\selectfont}] at (10.626,1.218) {$9$};
\draw[gp path] (11.506,1.772)--(11.506,1.952);
\draw[gp path] (11.506,7.087)--(11.506,6.907);
\node[gp node center,font={\fontsize{24.0pt}{38.4pt}\selectfont}] at (11.506,1.218) {$10$};
\draw[gp path] (2.705,7.087)--(2.705,1.772)--(11.506,1.772)--(11.506,7.087)--cycle;
\node[gp node center,rotate=-270,font={\fontsize{24.0pt}{36.0pt}\selectfont}] at (0.524,4.429) {Exec. time [sec.]};
\node[gp node center,font={\fontsize{24.0pt}{36.0pt}\selectfont}] at (7.105,0.387) {Length [$\times 100$]};
\node[gp node right,font={\fontsize{21.0pt}{24.0pt}\selectfont}] at (9.001,8.293) {\AccdSafety, dim. 3, $\varepsilon = 2.0$};
\gpcolor{rgb color={0.580,0.000,0.827}}
\gpsetlinewidth{5.00}
\draw[gp path] (9.332,8.293)--(10.836,8.293);
\draw[gp path] (2.793,2.038)--(2.925,2.304)--(3.145,2.569)--(3.365,2.569)--(3.585,2.835)%
  --(4.465,3.101)--(5.345,3.367)--(6.225,4.164)--(7.106,4.430)--(7.986,4.695)--(8.866,5.227)%
  --(9.746,5.227)--(10.626,6.024)--(11.506,6.290);
\gpsetpointsize{6.00}
\gppoint{gp mark 7}{(2.793,2.038)}
\gppoint{gp mark 7}{(2.925,2.304)}
\gppoint{gp mark 7}{(3.145,2.569)}
\gppoint{gp mark 7}{(3.365,2.569)}
\gppoint{gp mark 7}{(3.585,2.835)}
\gppoint{gp mark 7}{(4.465,3.101)}
\gppoint{gp mark 7}{(5.345,3.367)}
\gppoint{gp mark 7}{(6.225,4.164)}
\gppoint{gp mark 7}{(7.106,4.430)}
\gppoint{gp mark 7}{(7.986,4.695)}
\gppoint{gp mark 7}{(8.866,5.227)}
\gppoint{gp mark 7}{(9.746,5.227)}
\gppoint{gp mark 7}{(10.626,6.024)}
\gppoint{gp mark 7}{(11.506,6.290)}
\gppoint{gp mark 7}{(10.084,8.293)}
\gpcolor{color=gp lt color border}
\node[gp node right,font={\fontsize{21.0pt}{24.0pt}\selectfont}] at (9.001,7.739) {\AccdSafety, dim. 3, $\varepsilon = 0.9$};
\gpcolor{rgb color={0.000,0.620,0.451}}
\draw[gp path] (9.332,7.739)--(10.836,7.739);
\draw[gp path] (2.793,2.304)--(2.925,2.569)--(3.145,2.835)--(3.365,3.101)--(3.585,3.101)%
  --(4.465,3.632)--(5.345,4.164)--(6.225,4.961)--(7.106,5.758)--(7.986,6.024)--(8.866,6.024)%
  --(9.746,7.087)--(10.626,6.821)--(11.506,7.087);
\gppoint{gp mark 7}{(2.793,2.304)}
\gppoint{gp mark 7}{(2.925,2.569)}
\gppoint{gp mark 7}{(3.145,2.835)}
\gppoint{gp mark 7}{(3.365,3.101)}
\gppoint{gp mark 7}{(3.585,3.101)}
\gppoint{gp mark 7}{(4.465,3.632)}
\gppoint{gp mark 7}{(5.345,4.164)}
\gppoint{gp mark 7}{(6.225,4.961)}
\gppoint{gp mark 7}{(7.106,5.758)}
\gppoint{gp mark 7}{(7.986,6.024)}
\gppoint{gp mark 7}{(8.866,6.024)}
\gppoint{gp mark 7}{(9.746,7.087)}
\gppoint{gp mark 7}{(10.626,6.821)}
\gppoint{gp mark 7}{(11.506,7.087)}
\gppoint{gp mark 7}{(10.084,7.739)}
\gpcolor{color=gp lt color border}
\gpsetlinewidth{1.00}
\draw[gp path] (2.705,7.087)--(2.705,1.772)--(11.506,1.772)--(11.506,7.087)--cycle;
\gpdefrectangularnode{gp plot 1}{\pgfpoint{2.705cm}{1.772cm}}{\pgfpoint{11.506cm}{7.087cm}}
\end{tikzpicture}}
  \caption{\Accd{} of dimension~3}%
  \label{figure:result:accd:dimension3}
 \end{subfigure}
 \hfill
 \begin{subfigure}[b]{0.32\linewidth}
  \centering
  \scalebox{0.37}{\begin{tikzpicture}[gnuplot]
\tikzset{every node/.append style={font={\fontsize{18.0pt}{21.6pt}\selectfont}}}
\path (0.000,0.000) rectangle (12.500,8.750);
\gpcolor{color=gp lt color border}
\gpsetlinetype{gp lt border}
\gpsetdashtype{gp dt solid}
\gpsetlinewidth{1.00}
\draw[gp path] (2.374,1.772)--(2.554,1.772);
\draw[gp path] (11.506,1.772)--(11.326,1.772);
\node[gp node right,font={\fontsize{24.0pt}{38.4pt}\selectfont}] at (2.043,1.772) {$0.2$};
\draw[gp path] (2.374,2.835)--(2.554,2.835);
\draw[gp path] (11.506,2.835)--(11.326,2.835);
\node[gp node right,font={\fontsize{24.0pt}{38.4pt}\selectfont}] at (2.043,2.835) {$0.4$};
\draw[gp path] (2.374,3.898)--(2.554,3.898);
\draw[gp path] (11.506,3.898)--(11.326,3.898);
\node[gp node right,font={\fontsize{24.0pt}{38.4pt}\selectfont}] at (2.043,3.898) {$0.6$};
\draw[gp path] (2.374,4.961)--(2.554,4.961);
\draw[gp path] (11.506,4.961)--(11.326,4.961);
\node[gp node right,font={\fontsize{24.0pt}{38.4pt}\selectfont}] at (2.043,4.961) {$0.8$};
\draw[gp path] (2.374,6.024)--(2.554,6.024);
\draw[gp path] (11.506,6.024)--(11.326,6.024);
\node[gp node right,font={\fontsize{24.0pt}{38.4pt}\selectfont}] at (2.043,6.024) {$1$};
\draw[gp path] (2.374,7.087)--(2.554,7.087);
\draw[gp path] (11.506,7.087)--(11.326,7.087);
\node[gp node right,font={\fontsize{24.0pt}{38.4pt}\selectfont}] at (2.043,7.087) {$1.2$};
\draw[gp path] (2.374,1.772)--(2.374,1.952);
\draw[gp path] (2.374,7.087)--(2.374,6.907);
\node[gp node center,font={\fontsize{24.0pt}{38.4pt}\selectfont}] at (2.374,1.218) {$0$};
\draw[gp path] (3.287,1.772)--(3.287,1.952);
\draw[gp path] (3.287,7.087)--(3.287,6.907);
\node[gp node center,font={\fontsize{24.0pt}{38.4pt}\selectfont}] at (3.287,1.218) {$1$};
\draw[gp path] (4.200,1.772)--(4.200,1.952);
\draw[gp path] (4.200,7.087)--(4.200,6.907);
\node[gp node center,font={\fontsize{24.0pt}{38.4pt}\selectfont}] at (4.200,1.218) {$2$};
\draw[gp path] (5.114,1.772)--(5.114,1.952);
\draw[gp path] (5.114,7.087)--(5.114,6.907);
\node[gp node center,font={\fontsize{24.0pt}{38.4pt}\selectfont}] at (5.114,1.218) {$3$};
\draw[gp path] (6.027,1.772)--(6.027,1.952);
\draw[gp path] (6.027,7.087)--(6.027,6.907);
\node[gp node center,font={\fontsize{24.0pt}{38.4pt}\selectfont}] at (6.027,1.218) {$4$};
\draw[gp path] (6.940,1.772)--(6.940,1.952);
\draw[gp path] (6.940,7.087)--(6.940,6.907);
\node[gp node center,font={\fontsize{24.0pt}{38.4pt}\selectfont}] at (6.940,1.218) {$5$};
\draw[gp path] (7.853,1.772)--(7.853,1.952);
\draw[gp path] (7.853,7.087)--(7.853,6.907);
\node[gp node center,font={\fontsize{24.0pt}{38.4pt}\selectfont}] at (7.853,1.218) {$6$};
\draw[gp path] (8.766,1.772)--(8.766,1.952);
\draw[gp path] (8.766,7.087)--(8.766,6.907);
\node[gp node center,font={\fontsize{24.0pt}{38.4pt}\selectfont}] at (8.766,1.218) {$7$};
\draw[gp path] (9.680,1.772)--(9.680,1.952);
\draw[gp path] (9.680,7.087)--(9.680,6.907);
\node[gp node center,font={\fontsize{24.0pt}{38.4pt}\selectfont}] at (9.680,1.218) {$8$};
\draw[gp path] (10.593,1.772)--(10.593,1.952);
\draw[gp path] (10.593,7.087)--(10.593,6.907);
\node[gp node center,font={\fontsize{24.0pt}{38.4pt}\selectfont}] at (10.593,1.218) {$9$};
\draw[gp path] (11.506,1.772)--(11.506,1.952);
\draw[gp path] (11.506,7.087)--(11.506,6.907);
\node[gp node center,font={\fontsize{24.0pt}{38.4pt}\selectfont}] at (11.506,1.218) {$10$};
\draw[gp path] (2.374,7.087)--(2.374,1.772)--(11.506,1.772)--(11.506,7.087)--cycle;
\node[gp node center,rotate=-270,font={\fontsize{24.0pt}{36.0pt}\selectfont}] at (0.524,4.429) {Exec. time [sec.]};
\node[gp node center,font={\fontsize{24.0pt}{36.0pt}\selectfont}] at (6.940,0.387) {Length [$\times 100$]};
\node[gp node right,font={\fontsize{18.0pt}{21.6pt}\selectfont}] at (8.836,8.293) {\AccdSafety, dim. 4, $\varepsilon = 2.0$};
\gpcolor{rgb color={0.580,0.000,0.827}}
\gpsetlinewidth{5.00}
\draw[gp path] (9.167,8.293)--(10.671,8.293);
\draw[gp path] (2.465,2.941)--(2.602,1.878)--(2.831,2.729)--(3.059,2.622)--(3.287,2.569)%
  --(4.200,5.333)--(5.114,2.835)--(6.027,4.004)--(6.940,4.164)--(7.853,4.323)--(8.766,5.652)%
  --(9.680,5.865)--(10.593,5.333)--(11.506,5.439);
\gpsetpointsize{6.00}
\gppoint{gp mark 7}{(2.465,2.941)}
\gppoint{gp mark 7}{(2.602,1.878)}
\gppoint{gp mark 7}{(2.831,2.729)}
\gppoint{gp mark 7}{(3.059,2.622)}
\gppoint{gp mark 7}{(3.287,2.569)}
\gppoint{gp mark 7}{(4.200,5.333)}
\gppoint{gp mark 7}{(5.114,2.835)}
\gppoint{gp mark 7}{(6.027,4.004)}
\gppoint{gp mark 7}{(6.940,4.164)}
\gppoint{gp mark 7}{(7.853,4.323)}
\gppoint{gp mark 7}{(8.766,5.652)}
\gppoint{gp mark 7}{(9.680,5.865)}
\gppoint{gp mark 7}{(10.593,5.333)}
\gppoint{gp mark 7}{(11.506,5.439)}
\gppoint{gp mark 7}{(9.919,8.293)}
\gpcolor{color=gp lt color border}
\node[gp node right,font={\fontsize{18.0pt}{21.6pt}\selectfont}] at (8.836,7.739) {\AccdSafety, dim. 4, $\varepsilon = 0.9$};
\gpcolor{rgb color={0.000,0.620,0.451}}
\draw[gp path] (9.167,7.739)--(10.671,7.739);
\draw[gp path] (2.465,2.569)--(2.602,3.154)--(2.831,2.941)--(3.059,3.207)--(3.287,4.111)%
  --(4.200,5.067)--(5.114,4.642)--(6.027,5.599)--(6.940,5.227)--(7.853,5.280)--(8.766,5.014)%
  --(9.680,6.821)--(10.593,6.237)--(11.506,6.981);
\gppoint{gp mark 7}{(2.465,2.569)}
\gppoint{gp mark 7}{(2.602,3.154)}
\gppoint{gp mark 7}{(2.831,2.941)}
\gppoint{gp mark 7}{(3.059,3.207)}
\gppoint{gp mark 7}{(3.287,4.111)}
\gppoint{gp mark 7}{(4.200,5.067)}
\gppoint{gp mark 7}{(5.114,4.642)}
\gppoint{gp mark 7}{(6.027,5.599)}
\gppoint{gp mark 7}{(6.940,5.227)}
\gppoint{gp mark 7}{(7.853,5.280)}
\gppoint{gp mark 7}{(8.766,5.014)}
\gppoint{gp mark 7}{(9.680,6.821)}
\gppoint{gp mark 7}{(10.593,6.237)}
\gppoint{gp mark 7}{(11.506,6.981)}
\gppoint{gp mark 7}{(9.919,7.739)}
\gpcolor{color=gp lt color border}
\gpsetlinewidth{1.00}
\draw[gp path] (2.374,7.087)--(2.374,1.772)--(11.506,1.772)--(11.506,7.087)--cycle;
\gpdefrectangularnode{gp plot 1}{\pgfpoint{2.374cm}{1.772cm}}{\pgfpoint{11.506cm}{7.087cm}}
\end{tikzpicture}}
  \caption{\Accd{} of dimension~4}%
  \label{figure:result:accd:dimension4}
 \end{subfigure}
 \hfill
 \begin{subfigure}[b]{0.32\linewidth}
  \centering 
  \scalebox{0.37}{\begin{tikzpicture}[gnuplot]
\tikzset{every node/.append style={font={\fontsize{24.0pt}{21.6pt}\selectfont}}}
\path (0.000,0.000) rectangle (12.500,8.750);
\gpcolor{color=gp lt color border}
\gpsetlinetype{gp lt border}
\gpsetdashtype{gp dt solid}
\gpsetlinewidth{1.00}
\draw[gp path] (2.043,1.772)--(2.223,1.772);
\draw[gp path] (11.506,1.772)--(11.326,1.772);
\node[gp node right,font={\fontsize{24.0pt}{38.4pt}\selectfont}] at (1.712,1.772) {$0$};
\draw[gp path] (2.043,2.658)--(2.223,2.658);
\draw[gp path] (11.506,2.658)--(11.326,2.658);
\node[gp node right,font={\fontsize{24.0pt}{38.4pt}\selectfont}] at (1.712,2.658) {$10$};
\draw[gp path] (2.043,3.544)--(2.223,3.544);
\draw[gp path] (11.506,3.544)--(11.326,3.544);
\node[gp node right,font={\fontsize{24.0pt}{38.4pt}\selectfont}] at (1.712,3.544) {$20$};
\draw[gp path] (2.043,4.430)--(2.223,4.430);
\draw[gp path] (11.506,4.430)--(11.326,4.430);
\node[gp node right,font={\fontsize{24.0pt}{38.4pt}\selectfont}] at (1.712,4.430) {$30$};
\draw[gp path] (2.043,5.315)--(2.223,5.315);
\draw[gp path] (11.506,5.315)--(11.326,5.315);
\node[gp node right,font={\fontsize{24.0pt}{38.4pt}\selectfont}] at (1.712,5.315) {$40$};
\draw[gp path] (2.043,6.201)--(2.223,6.201);
\draw[gp path] (11.506,6.201)--(11.326,6.201);
\node[gp node right,font={\fontsize{24.0pt}{38.4pt}\selectfont}] at (1.712,6.201) {$50$};
\draw[gp path] (2.043,7.087)--(2.223,7.087);
\draw[gp path] (11.506,7.087)--(11.326,7.087);
\node[gp node right,font={\fontsize{24.0pt}{38.4pt}\selectfont}] at (1.712,7.087) {$60$};
\draw[gp path] (2.043,1.772)--(2.043,1.952);
\draw[gp path] (2.043,7.087)--(2.043,6.907);
\node[gp node center,font={\fontsize{24.0pt}{38.4pt}\selectfont}] at (2.043,1.218) {$0$};
\draw[gp path] (2.989,1.772)--(2.989,1.952);
\draw[gp path] (2.989,7.087)--(2.989,6.907);
\node[gp node center,font={\fontsize{24.0pt}{38.4pt}\selectfont}] at (2.989,1.218) {$1$};
\draw[gp path] (3.936,1.772)--(3.936,1.952);
\draw[gp path] (3.936,7.087)--(3.936,6.907);
\node[gp node center,font={\fontsize{24.0pt}{38.4pt}\selectfont}] at (3.936,1.218) {$2$};
\draw[gp path] (4.882,1.772)--(4.882,1.952);
\draw[gp path] (4.882,7.087)--(4.882,6.907);
\node[gp node center,font={\fontsize{24.0pt}{38.4pt}\selectfont}] at (4.882,1.218) {$3$};
\draw[gp path] (5.828,1.772)--(5.828,1.952);
\draw[gp path] (5.828,7.087)--(5.828,6.907);
\node[gp node center,font={\fontsize{24.0pt}{38.4pt}\selectfont}] at (5.828,1.218) {$4$};
\draw[gp path] (6.775,1.772)--(6.775,1.952);
\draw[gp path] (6.775,7.087)--(6.775,6.907);
\node[gp node center,font={\fontsize{24.0pt}{38.4pt}\selectfont}] at (6.775,1.218) {$5$};
\draw[gp path] (7.721,1.772)--(7.721,1.952);
\draw[gp path] (7.721,7.087)--(7.721,6.907);
\node[gp node center,font={\fontsize{24.0pt}{38.4pt}\selectfont}] at (7.721,1.218) {$6$};
\draw[gp path] (8.667,1.772)--(8.667,1.952);
\draw[gp path] (8.667,7.087)--(8.667,6.907);
\node[gp node center,font={\fontsize{24.0pt}{38.4pt}\selectfont}] at (8.667,1.218) {$7$};
\draw[gp path] (9.613,1.772)--(9.613,1.952);
\draw[gp path] (9.613,7.087)--(9.613,6.907);
\node[gp node center,font={\fontsize{24.0pt}{38.4pt}\selectfont}] at (9.613,1.218) {$8$};
\draw[gp path] (10.560,1.772)--(10.560,1.952);
\draw[gp path] (10.560,7.087)--(10.560,6.907);
\node[gp node center,font={\fontsize{24.0pt}{38.4pt}\selectfont}] at (10.560,1.218) {$9$};
\draw[gp path] (11.506,1.772)--(11.506,1.952);
\draw[gp path] (11.506,7.087)--(11.506,6.907);
\node[gp node center,font={\fontsize{24.0pt}{38.4pt}\selectfont}] at (11.506,1.218) {$10$};
\draw[gp path] (2.043,7.087)--(2.043,1.772)--(11.506,1.772)--(11.506,7.087)--cycle;
\node[gp node center,rotate=-270,font={\fontsize{24.0pt}{36.0pt}\selectfont}] at (0.524,4.429) {Exec. time [sec.]};
\node[gp node center,font={\fontsize{24.0pt}{36.0pt}\selectfont}] at (6.774,0.387) {Length [$\times 100$]};
\node[gp node right,font={\fontsize{21.0pt}{21.6pt}\selectfont}] at (8.670,8.293) {\AccdSafety, dim. 5, $\varepsilon = 2.0$};
\gpcolor{rgb color={0.580,0.000,0.827}}
\gpsetlinewidth{5.00}
\draw[gp path] (9.001,8.293)--(10.505,8.293);
\draw[gp path] (2.138,4.370)--(2.280,2.589)--(2.516,4.523)--(2.753,4.423)--(2.989,3.644)%
  --(3.936,4.976)--(4.882,3.062)--(5.828,3.131)--(6.775,4.794)--(7.721,3.222)--(8.667,3.198)%
  --(9.613,5.233)--(10.560,5.907)--(11.506,6.594);
\gpsetpointsize{6.00}
\gppoint{gp mark 7}{(2.138,4.370)}
\gppoint{gp mark 7}{(2.280,2.589)}
\gppoint{gp mark 7}{(2.516,4.523)}
\gppoint{gp mark 7}{(2.753,4.423)}
\gppoint{gp mark 7}{(2.989,3.644)}
\gppoint{gp mark 7}{(3.936,4.976)}
\gppoint{gp mark 7}{(4.882,3.062)}
\gppoint{gp mark 7}{(5.828,3.131)}
\gppoint{gp mark 7}{(6.775,4.794)}
\gppoint{gp mark 7}{(7.721,3.222)}
\gppoint{gp mark 7}{(8.667,3.198)}
\gppoint{gp mark 7}{(9.613,5.233)}
\gppoint{gp mark 7}{(10.560,5.907)}
\gppoint{gp mark 7}{(11.506,6.594)}
\gppoint{gp mark 7}{(9.753,8.293)}
\gpcolor{color=gp lt color border}
\node[gp node right,font={\fontsize{21.0pt}{21.6pt}\selectfont}] at (8.670,7.739) {\AccdSafety, dim. 5, $\varepsilon = 0.9$};
\gpcolor{rgb color={0.000,0.620,0.451}}
\draw[gp path] (9.001,7.739)--(10.505,7.739);
\draw[gp path] (2.138,2.437)--(2.280,2.669)--(2.516,2.629)--(2.753,3.048)--(2.989,2.912)%
  --(3.936,2.564)--(4.882,3.103)--(5.828,2.641)--(6.775,2.810)--(7.721,2.551)--(8.667,3.282)%
  --(9.613,2.871)--(10.560,2.643)--(11.506,2.795);
\gppoint{gp mark 7}{(2.138,2.437)}
\gppoint{gp mark 7}{(2.280,2.669)}
\gppoint{gp mark 7}{(2.516,2.629)}
\gppoint{gp mark 7}{(2.753,3.048)}
\gppoint{gp mark 7}{(2.989,2.912)}
\gppoint{gp mark 7}{(3.936,2.564)}
\gppoint{gp mark 7}{(4.882,3.103)}
\gppoint{gp mark 7}{(5.828,2.641)}
\gppoint{gp mark 7}{(6.775,2.810)}
\gppoint{gp mark 7}{(7.721,2.551)}
\gppoint{gp mark 7}{(8.667,3.282)}
\gppoint{gp mark 7}{(9.613,2.871)}
\gppoint{gp mark 7}{(10.560,2.643)}
\gppoint{gp mark 7}{(11.506,2.795)}
\gppoint{gp mark 7}{(9.753,7.739)}
\gpcolor{color=gp lt color border}
\gpsetlinewidth{1.00}
\draw[gp path] (2.043,7.087)--(2.043,1.772)--(11.506,1.772)--(11.506,7.087)--cycle;
\gpdefrectangularnode{gp plot 1}{\pgfpoint{2.043cm}{1.772cm}}{\pgfpoint{11.506cm}{7.087cm}}
\end{tikzpicture}}
  \caption{\Accd{} of dimension~5}%
  \label{figure:result:accd:dimension5}
 \end{subfigure}
 \caption{Execution time of \masakiTool{} for \Accd{}}%
 \label{figure:result:accd:dimension}
\end{figure}
\myparagraph{RQ2: scalability \wrt{} the word length}
\cref{figure:result:rp11_accc_acci:dimension,figure:result:accd:dimension} show the execution time of \masakiTool{} \wrt{} the length of the timed quantitative word for selected experiment settings.

In \cref{figure:result:rp11_accc_acci:dimension,figure:result:accd:dimension3}, we observe that when the dimension of the LHA is not large, the execution time was more or less linear to the word length.
This is because, when the number of the intermediate states ($\Conf_{i}$ in \cref{algorithm:incremental:outline}) is constant, the execution time of the bounded-time reachability analysis (\cref{algorithm:incremental:outline:reachability} of \cref{algorithm:incremental:outline}) is constant for each iteration, and
the execution time of \cref{algorithm:incremental:outline} is linear to the word length.
Thanks to the merging of the convex polyhedra, such saturation often happens when the word length is long enough for the complexity of the LHA.\@

In \cref{figure:result:accd:dimension5}, we observe that for \Acci{} of dimension~5, the execution time was more or less constant \wrt{} the word length.
Such behavior also happens for other benchmarks when the word length is short, \eg{} when the word length is less than 200 for \Accc{} of dimension~5.

In \cref{figure:result:accd:dimension4}, we observe the behavior between them. Namely, for \Acci{} of dimension~4, the execution time was more or less constant \wrt{} the word length for short logs and more or less linear to the word length.

Overall, in our experiments, the execution time was at most linear, and  we conclude that at least for many benchmarks, \masakiTool{} is scalable to the word length.
\begin{figure}[tbp]
 \begin{minipage}{.45\linewidth}
  \centering  
  \scalebox{0.38}{\begin{tikzpicture}[gnuplot]
\tikzset{every node/.append style={font={\fontsize{20.0pt}{24.0pt}\selectfont}}}
\path (0.000,0.000) rectangle (17.000,11.000);
\gpcolor{color=gp lt color border}
\gpsetlinetype{gp lt border}
\gpsetdashtype{gp dt solid}
\gpsetlinewidth{1.00}
\draw[gp path] (3.376,1.971)--(3.556,1.971);
\draw[gp path] (14.195,1.971)--(14.015,1.971);
\node[gp node right] at (3.008,1.971) {$0.001$};
\draw[gp path] (3.376,2.259)--(3.466,2.259);
\draw[gp path] (14.195,2.259)--(14.105,2.259);
\draw[gp path] (3.376,2.428)--(3.466,2.428);
\draw[gp path] (14.195,2.428)--(14.105,2.428);
\draw[gp path] (3.376,2.548)--(3.466,2.548);
\draw[gp path] (14.195,2.548)--(14.105,2.548);
\draw[gp path] (3.376,2.641)--(3.466,2.641);
\draw[gp path] (14.195,2.641)--(14.105,2.641);
\draw[gp path] (3.376,2.716)--(3.466,2.716);
\draw[gp path] (14.195,2.716)--(14.105,2.716);
\draw[gp path] (3.376,2.781)--(3.466,2.781);
\draw[gp path] (14.195,2.781)--(14.105,2.781);
\draw[gp path] (3.376,2.836)--(3.466,2.836);
\draw[gp path] (14.195,2.836)--(14.105,2.836);
\draw[gp path] (3.376,2.885)--(3.466,2.885);
\draw[gp path] (14.195,2.885)--(14.105,2.885);
\draw[gp path] (3.376,2.929)--(3.556,2.929);
\draw[gp path] (14.195,2.929)--(14.015,2.929);
\node[gp node right] at (3.008,2.929) {$0.01$};
\draw[gp path] (3.376,3.217)--(3.466,3.217);
\draw[gp path] (14.195,3.217)--(14.105,3.217);
\draw[gp path] (3.376,3.386)--(3.466,3.386);
\draw[gp path] (14.195,3.386)--(14.105,3.386);
\draw[gp path] (3.376,3.506)--(3.466,3.506);
\draw[gp path] (14.195,3.506)--(14.105,3.506);
\draw[gp path] (3.376,3.599)--(3.466,3.599);
\draw[gp path] (14.195,3.599)--(14.105,3.599);
\draw[gp path] (3.376,3.674)--(3.466,3.674);
\draw[gp path] (14.195,3.674)--(14.105,3.674);
\draw[gp path] (3.376,3.739)--(3.466,3.739);
\draw[gp path] (14.195,3.739)--(14.105,3.739);
\draw[gp path] (3.376,3.794)--(3.466,3.794);
\draw[gp path] (14.195,3.794)--(14.105,3.794);
\draw[gp path] (3.376,3.843)--(3.466,3.843);
\draw[gp path] (14.195,3.843)--(14.105,3.843);
\draw[gp path] (3.376,3.887)--(3.556,3.887);
\draw[gp path] (14.195,3.887)--(14.015,3.887);
\node[gp node right] at (3.008,3.887) {$0.1$};
\draw[gp path] (3.376,4.175)--(3.466,4.175);
\draw[gp path] (14.195,4.175)--(14.105,4.175);
\draw[gp path] (3.376,4.344)--(3.466,4.344);
\draw[gp path] (14.195,4.344)--(14.105,4.344);
\draw[gp path] (3.376,4.464)--(3.466,4.464);
\draw[gp path] (14.195,4.464)--(14.105,4.464);
\draw[gp path] (3.376,4.557)--(3.466,4.557);
\draw[gp path] (14.195,4.557)--(14.105,4.557);
\draw[gp path] (3.376,4.632)--(3.466,4.632);
\draw[gp path] (14.195,4.632)--(14.105,4.632);
\draw[gp path] (3.376,4.697)--(3.466,4.697);
\draw[gp path] (14.195,4.697)--(14.105,4.697);
\draw[gp path] (3.376,4.752)--(3.466,4.752);
\draw[gp path] (14.195,4.752)--(14.105,4.752);
\draw[gp path] (3.376,4.801)--(3.466,4.801);
\draw[gp path] (14.195,4.801)--(14.105,4.801);
\draw[gp path] (3.376,4.845)--(3.556,4.845);
\draw[gp path] (14.195,4.845)--(14.015,4.845);
\node[gp node right] at (3.008,4.845) {$1$};
\draw[gp path] (3.376,5.133)--(3.466,5.133);
\draw[gp path] (14.195,5.133)--(14.105,5.133);
\draw[gp path] (3.376,5.302)--(3.466,5.302);
\draw[gp path] (14.195,5.302)--(14.105,5.302);
\draw[gp path] (3.376,5.422)--(3.466,5.422);
\draw[gp path] (14.195,5.422)--(14.105,5.422);
\draw[gp path] (3.376,5.515)--(3.466,5.515);
\draw[gp path] (14.195,5.515)--(14.105,5.515);
\draw[gp path] (3.376,5.590)--(3.466,5.590);
\draw[gp path] (14.195,5.590)--(14.105,5.590);
\draw[gp path] (3.376,5.655)--(3.466,5.655);
\draw[gp path] (14.195,5.655)--(14.105,5.655);
\draw[gp path] (3.376,5.710)--(3.466,5.710);
\draw[gp path] (14.195,5.710)--(14.105,5.710);
\draw[gp path] (3.376,5.759)--(3.466,5.759);
\draw[gp path] (14.195,5.759)--(14.105,5.759);
\draw[gp path] (3.376,5.803)--(3.556,5.803);
\draw[gp path] (14.195,5.803)--(14.015,5.803);
\node[gp node right] at (3.008,5.803) {$10$};
\draw[gp path] (3.376,6.091)--(3.466,6.091);
\draw[gp path] (14.195,6.091)--(14.105,6.091);
\draw[gp path] (3.376,6.260)--(3.466,6.260);
\draw[gp path] (14.195,6.260)--(14.105,6.260);
\draw[gp path] (3.376,6.380)--(3.466,6.380);
\draw[gp path] (14.195,6.380)--(14.105,6.380);
\draw[gp path] (3.376,6.473)--(3.466,6.473);
\draw[gp path] (14.195,6.473)--(14.105,6.473);
\draw[gp path] (3.376,6.548)--(3.466,6.548);
\draw[gp path] (14.195,6.548)--(14.105,6.548);
\draw[gp path] (3.376,6.613)--(3.466,6.613);
\draw[gp path] (14.195,6.613)--(14.105,6.613);
\draw[gp path] (3.376,6.668)--(3.466,6.668);
\draw[gp path] (14.195,6.668)--(14.105,6.668);
\draw[gp path] (3.376,6.717)--(3.466,6.717);
\draw[gp path] (14.195,6.717)--(14.105,6.717);
\draw[gp path] (3.376,6.761)--(3.556,6.761);
\draw[gp path] (14.195,6.761)--(14.015,6.761);
\node[gp node right] at (3.008,6.761) {$100$};
\draw[gp path] (3.376,7.049)--(3.466,7.049);
\draw[gp path] (14.195,7.049)--(14.105,7.049);
\draw[gp path] (3.376,7.218)--(3.466,7.218);
\draw[gp path] (14.195,7.218)--(14.105,7.218);
\draw[gp path] (3.376,7.338)--(3.466,7.338);
\draw[gp path] (14.195,7.338)--(14.105,7.338);
\draw[gp path] (3.376,7.431)--(3.466,7.431);
\draw[gp path] (14.195,7.431)--(14.105,7.431);
\draw[gp path] (3.376,7.506)--(3.466,7.506);
\draw[gp path] (14.195,7.506)--(14.105,7.506);
\draw[gp path] (3.376,7.571)--(3.466,7.571);
\draw[gp path] (14.195,7.571)--(14.105,7.571);
\draw[gp path] (3.376,7.626)--(3.466,7.626);
\draw[gp path] (14.195,7.626)--(14.105,7.626);
\draw[gp path] (3.376,7.675)--(3.466,7.675);
\draw[gp path] (14.195,7.675)--(14.105,7.675);
\draw[gp path] (3.376,7.719)--(3.556,7.719);
\draw[gp path] (14.195,7.719)--(14.015,7.719);
\node[gp node right] at (3.008,7.719) {$1000$};
\draw[gp path] (3.376,1.971)--(3.376,2.151);
\draw[gp path] (3.376,7.719)--(3.376,7.539);
\node[gp node center] at (3.376,1.355) {$2$};
\draw[gp path] (4.728,1.971)--(4.728,2.151);
\draw[gp path] (4.728,7.719)--(4.728,7.539);
\node[gp node center] at (4.728,1.355) {$2.5$};
\draw[gp path] (6.081,1.971)--(6.081,2.151);
\draw[gp path] (6.081,7.719)--(6.081,7.539);
\node[gp node center] at (6.081,1.355) {$3$};
\draw[gp path] (7.433,1.971)--(7.433,2.151);
\draw[gp path] (7.433,7.719)--(7.433,7.539);
\node[gp node center] at (7.433,1.355) {$3.5$};
\draw[gp path] (8.786,1.971)--(8.786,2.151);
\draw[gp path] (8.786,7.719)--(8.786,7.539);
\node[gp node center] at (8.786,1.355) {$4$};
\draw[gp path] (10.138,1.971)--(10.138,2.151);
\draw[gp path] (10.138,7.719)--(10.138,7.539);
\node[gp node center] at (10.138,1.355) {$4.5$};
\draw[gp path] (11.490,1.971)--(11.490,2.151);
\draw[gp path] (11.490,7.719)--(11.490,7.539);
\node[gp node center] at (11.490,1.355) {$5$};
\draw[gp path] (12.843,1.971)--(12.843,2.151);
\draw[gp path] (12.843,7.719)--(12.843,7.539);
\node[gp node center] at (12.843,1.355) {$5.5$};
\draw[gp path] (14.195,1.971)--(14.195,2.151);
\draw[gp path] (14.195,7.719)--(14.195,7.539);
\node[gp node center] at (14.195,1.355) {$6$};
\draw[gp path] (3.376,7.719)--(3.376,1.971)--(14.195,1.971)--(14.195,7.719)--cycle;
\node[gp node center,rotate=-270] at (0.584,4.845) {Execution time [sec.]};
\node[gp node center] at (8.785,0.431) {Dimension};
\node[gp node right] at (11.823,10.312) {\phaverlite{}, \AccdSafety, len. 100, $\varepsilon = 2.0$};
\gpcolor{rgb color={0.580,0.000,0.827}}
\gpsetlinewidth{10.00}
\draw[gp path] (12.191,10.312)--(13.843,10.312);
\draw[gp path] (3.376,3.688)--(6.081,4.219)--(8.786,4.864)--(11.490,5.688)--(14.195,6.876);
\gpsetpointsize{12.00}
\gppoint{gp mark 9}{(3.376,3.688)}
\gppoint{gp mark 9}{(6.081,4.219)}
\gppoint{gp mark 9}{(8.786,4.864)}
\gppoint{gp mark 9}{(11.490,5.688)}
\gppoint{gp mark 9}{(14.195,6.876)}
\gppoint{gp mark 9}{(13.017,10.312)}
\gpcolor{color=gp lt color border}
\node[gp node right] at (11.823,9.696) {\phaverlite{}, \AccdSafety, len. 100, $\varepsilon = 0.9$};
\gpcolor{rgb color={0.000,0.620,0.451}}
\draw[gp path] (12.191,9.696)--(13.843,9.696);
\draw[gp path] (3.376,3.695)--(6.081,4.242)--(8.786,4.924)--(11.490,5.717)--(14.195,6.705);
\gppoint{gp mark 9}{(3.376,3.695)}
\gppoint{gp mark 9}{(6.081,4.242)}
\gppoint{gp mark 9}{(8.786,4.924)}
\gppoint{gp mark 9}{(11.490,5.717)}
\gppoint{gp mark 9}{(14.195,6.705)}
\gppoint{gp mark 9}{(13.017,9.696)}
\gpcolor{color=gp lt color border}
\node[gp node right] at (11.823,9.080) {\masakiTool{}, \AccdSafety, len. 100, $\varepsilon = 2.0$};
\gpcolor{rgb color={0.337,0.706,0.914}}
\draw[gp path] (12.191,9.080)--(13.843,9.080);
\draw[gp path] (3.376,2.818)--(6.081,3.458)--(8.786,4.408)--(11.490,6.114)--(14.195,7.360);
\gppoint{gp mark 7}{(3.376,2.818)}
\gppoint{gp mark 7}{(6.081,3.458)}
\gppoint{gp mark 7}{(8.786,4.408)}
\gppoint{gp mark 7}{(11.490,6.114)}
\gppoint{gp mark 7}{(14.195,7.360)}
\gppoint{gp mark 7}{(13.017,9.080)}
\gpcolor{color=gp lt color border}
\node[gp node right] at (11.823,8.464) {\masakiTool{}, \AccdSafety, len. 100, $\varepsilon = 0.9$};
\gpcolor{rgb color={0.902,0.624,0.000}}
\draw[gp path] (12.191,8.464)--(13.843,8.464);
\draw[gp path] (3.376,2.915)--(6.081,3.582)--(8.786,4.658)--(11.490,5.908)--(14.195,7.069);
\gppoint{gp mark 7}{(3.376,2.915)}
\gppoint{gp mark 7}{(6.081,3.582)}
\gppoint{gp mark 7}{(8.786,4.658)}
\gppoint{gp mark 7}{(11.490,5.908)}
\gppoint{gp mark 7}{(14.195,7.069)}
\gppoint{gp mark 7}{(13.017,8.464)}
\gpcolor{color=gp lt color border}
\gpsetlinewidth{1.00}
\draw[gp path] (3.376,7.719)--(3.376,1.971)--(14.195,1.971)--(14.195,7.719)--cycle;
\gpdefrectangularnode{gp plot 1}{\pgfpoint{3.376cm}{1.971cm}}{\pgfpoint{14.195cm}{7.719cm}}
\end{tikzpicture}}
 \caption{Execution time of \phaverlite{} and \masakiTool{} for \Accd{} fixing the word length to be 100}%
 \label{figure:result:ACCI-length} 
\end{minipage}
 \hfill
 \begin{minipage}{.52\linewidth}
 \centering
 \scalebox{0.47}{ \input{fig/shortest_interval_robustness.tikz}}
 \caption{Robustness (vertical) vs.\ shortest sampling interval with no false alarms [ms.] (horizontal)}%
 \label{figure:result:navigation}
\end{minipage}
\end{figure}
\begin{table}[tbp]
 \caption{The number and the ratio of the false and total alarms in \UnsafeNAV{} for each sampling interval}%
 \label{table:result:false_alarm}
 \centering
 \begin{tabular}{lrrrrrrrrrrrr}
  \toprule
  & 50\,ms & 100\,ms & 150\,ms & 200\,ms & 250\,ms & 300\,ms\\
  \midrule
  \# of false alarms / \# of safe behaviors & 0.19 & 0.28 & 0.40 & 0.50 & 0.59 & 0.72\\
  \# of false alarms / \# of total alarms & 0.20 & 0.28 & 0.35 & 0.41 & 0.45 & 0.50\\
  \# of the false alarms & 11 & 16 & 23 & 29 & 34 & 42\\
  \# of the total alarms & 53 & 58 & 65 & 71 & 76 & 84\\
  \bottomrule
 \end{tabular}
 \vspace{0.5em}
 \begin{tabular}{lrrrrrrrrrrrr}
  \toprule
  & 350\,ms & 400\,ms & 450\,ms & 500\,ms & 550\,ms & 600\,ms\\
  \midrule
  \# of false alarms / \# of safe behaviors & 0.79 & 0.81 & 0.91 & 0.95 & 0.98 & 1.0\\
  \# of false alarms / \# of total alarms & 0.52 & 0.53 & 0.56 & 0.57 & 0.58 & 0.58\\
  \# of the false alarms & 46 & 47 & 53 & 55 & 57 & 58\\
  \# of the total alarms & 88 & 89 & 95 & 97 & 99 & 100\\
  \bottomrule
 \end{tabular}
\end{table}

\myparagraph{RQ3: scalability \wrt{} the dimensionality of the LHA}
\cref{figure:result:ACCI-length} shows the execution time of \phaverlite{} and \masakiTool{} to the dimension of the LHA, where the length of the log is fixed to be 100.
As we mentioned in \cref{subsubsec:scalability_experiments}, the sampling interval is from 1 to 5 seconds, uniformly distributed; thus, each $\word$ spans 300 seconds on average.
Note that the $y$-axis of \cref{figure:result:ACCI-length} follows a logarithmic scale.

In \cref{figure:result:ACCI-length}, we observe that the execution time is more or less exponential to the dimension of the LHA.\@
This is due to the exponential complexity of the convex polyhedra operations.
However, in \cref{table:result:accc,table:result:acci,table:result:accd}, we observe that for any benchmark (excluding \Acci{}, which consists of only one LHA and is not suitable for this discussion), \masakiTool{} can effectively process a huge log up to around 5~dimensions.
This is important for monitoring where the log tends to be huge while the dimension of the bounding model may not be much large.

In contrast, in \cref{figure:result:ACCI-length}, we observe that \phaverlite{} is more scalable to the model dimension than \masakiTool{}.
This is thanks to the optimized convex polyhedra algorithms and implementations in \phaverlite{}.
Again, our future work will consist in optimizing \masakiTool{}, \eg{} by using the techniques from~\cite{BZ19,CAF11}.

Overall, our experiment results suggest that for signals of up to 5 dimensions, our monitoring works more or less in real-time because both implementations handle the logs spanning at least 100 seconds (a few paragraphs ago) in less than 30 seconds on average (\cref{figure:result:ACCI-length}).

\subsubsection{Precision Experiments}\label{subsec:precision_experiments}
In the precision experiments, since the result of the procedures presented in \cref{section:algorithm1,section:algorithm2} are the same, we only used \masakiTool{} to evaluate the precision of model-bounded monitoring.
Since \NAV{} and \UnsafeNAV{} have nondeterminism in the initial configuration, we randomly generated 100 logs for each of them.
We note that \NAV{} always satisfies the specification while \UnsafeNAV{} may violate it.
Since the initial configuration of \SharedGasBurner{} is unique and the simulation by Simulink is deterministic, we generated only one log for \SharedGasBurner{}.
From a dense enough raw log, we generated timed quantitative words $\word$ with constant sampling interval.
We tried the following 10 sampling intervals for \NAV{} and \SharedGasBurner{}: $100, 200, 300, \dots, 1000$ milliseconds.
We tried the following 12 sampling intervals for \UnsafeNAV{}: $50, 100, 150, \dots, 600$ milliseconds.
\myparagraph{RQ4: robustness vs.\ precision}
\cref{figure:result:navigation} shows the shortest sampling interval with no false alarms and the robustness for each log of \NAV{} and \UnsafeNAV{}. Since the monitored specification is $v_x \in [-1,1] \land v_y \in [-1,1]$, the robustness, \ie{} the satisfaction degree of the monitored specification, is the minimum value of $\min\{1 - |v_x|, 1 - |v_y|\}$ in the whole execution in the raw log.
For \SharedGasBurner{}, the shortest sampling interval without any false alarms is 600 ms and the robustness is 0. %

In \cref{figure:result:navigation}, we observe that when the robustness is small, \ie{} the execution is close to violation, model-bounded monitoring tends to cause false alarms with a short sampling interval. This is because when the sampled log is coarse, the overapproximation in the model-bounded monitoring is rough, and near unsafe behavior is deemed to be unsafe.
 This observation suggests the following process to decide a reasonable sampling interval:
\begin{ienumeration}
 \item assume false alarms with small robustness (\eg{} robustness $\leq 0.15$) are acceptable,
 \item we plot the relationship between the robustness and the shortest sampling interval without false alarms, \eg{} \cref{figure:result:navigation}, and
 \item we decide the sampling interval so that the amount of the false alarms with high robustness becomes reasonably small, \eg{} in \cref{figure:result:navigation}, if the false alarms with robustness $\leq 0.15$ are acceptable, 600 ms is a reasonable sampling interval.
\end{ienumeration}

\cref{table:result:false_alarm} shows the rate of the false alarms in \UnsafeNAV{} compared with the number of the safe behaviors and the number of the total alarms.
Although the monitored behaviors are very close to the violation (the mean robustness is $6.20 \times 10^{-3}$ and the standard deviation is $3.00 \times 10^{-2}$), when the sampling interval is~50\,ms, 80\,\% of the alarms are true alarms.
We believe that this amount of false alarms is acceptable in many applications.
If about 50\,\% of false alarms is acceptable, we can increase the sampling interval to~300\,ms.

For \SharedGasBurner{}, the shortest sampling interval without any false alarms is 600 ms and the robustness is 0. %
In \SharedGasBurner{}, even though the robustness is 0, we have no false alarms for a relatively large sampling interval (600\,ms). This is because:
\begin{ienumeration}
 \item the low robustness is due to the initial value $x_1 = 0$,
 \item initially, $\dot{x}_1 > 0$ holds in the bounding model, and we have no false alarm there;,
 \item we potentially have false alarm only around the switching, and 
 \item around the switching, the robustness is 0.099, which is not very small.
\end{ienumeration}

We note that our LHA construction does not utilize any runtime information, and the precision of the approximation is limited. It is a future work to reduce the false alarms by improving the approximation utilizing the runtime information by an online construction of the LHAs \eg{} the iterative refinement in~\cite[Section~3.2]{Frehse2008}.

\begin{PartialObservationVersion}
\section{Model-Bounded Monitoring with Partial Observations}\label{section:hamoni_partial}

 Here, we briefly show that it is straightforward to generalize model-bounded monitoring to monitor the signals with partial observations.
We note that by using partial observations, we can also use parameters both in the model and the specification because a parameter is equivalent to a variable $\variable$ with flow $\dot{\variable} = 0 $ and never observed in the log $\word$.
Moreover, we can conduct \emph{timed pattern matching}~\cite{UFAM14,WAH16} by constructing the \emph{matching automaton} in~\cite{BFNMA18,Waga19} with parameters $t$ and $t'$.

\subsection{Partial timed quantitative word}
 We introduce partial timed quantitative word to model the sampling with \emph{partial} observation. A partial timed quantitative word is a timed quantitative word where the valuation function may be partial.
 For partial timed quantitative words, we use the same notation as timed quantitative words.

 \begin{definition}
  [partial timed quantitative words]%
  \label{def:partial_timed_quantitative_word}
  A \emph{partial timed quantitative word} $\word$ is a sequence $(\wordVal_1, \tau_1), (\wordVal_2, \tau_2), \dots,(\wordVal_m, \tau_m)$ of pairs $(\wordVal_i, \tau_i)$ of a partial valuation $\wordVal_i\colon \Variables \partfun \R$ and a timestamp $\tau_i \in \Rnn$ satisfying $\tau_i \leq \tau_{i+1}$ for each $i \in \{1,2,\dots,m-1\}$.
 For partial valuations $\wordVal,\wordVal'\colon \Variables \partfun \R$,
 we denote $\wordVal \sqsubseteq \wordVal'$ if
 for any $\variable \in \domain{\wordVal}$, $\wordVal(\variable) = \wordVal'(\variable)$ holds.
 \end{definition}
\subsection{Partial monitored language}

We extend the monitored language in \cref{def:monitoredLang} to the set of partial timed quantitative words. For partial monitored languages, we use the same notation as monitored languages.

\begin{definition}[partial monitored language $\Lmon(\A)$]%
\label{def:partial_monitoredLang}
Let $\run = s_0 \to_1 s_1 \to_2 \dots \to_n s_n$ be a run of an LHA $\A$ (\cref{def:concreteSem}), and
 $\word = (\wordVal_1,\tau_1),(\wordVal_2,\tau_2), \dots, (\wordVal_m,\tau_m) $ be a partial timed quantitative word.
We say $\word$
is \emph{associated} to $\run$
 if, for each $j \in \{1,2,\dots,m\}$, we have either of the following two. Here $\loc_{i}, \varval_{i}$ are so that  $s_{i} = (\loc_{i}, \varval_{i})$ for each $i \in\{0,1,\dots,n\}$. 
 \begin{enumerate}
 \item There exists $i \in \{0,1,2,\dots, n\}$ such that $\Duration(\run[i]) = \tau_j$ and $\wordVal_j \sqsubseteq \varval_{i}$; or
 \item There exists $i \in \{0,1,2,\dots, n-1\}$ such that $\Duration(\run[i]) < \tau_j < \Duration(\run[i+1])$ and for any $\variable \in \Variables$, $\wordVal_j(\variable) = \varval_{i}(\variable) + (\tau_j - \Duration(\run[i]))f_{i}(\dot{\variable})$ holds if $\wordVal_{j}(\variable)$ is defined, where ${\flecheRel_i} = {\longueflecheRel{d_i,\flow_i}}$. 
 \end{enumerate}
 Finally, the \emph{partial monitored language} $\partialLg(\A)$ of an LHA $\A$ is the set of partial timed quantitative words associated with some accepting run of~$\A$.
\end{definition}
\subsection{Membership constraint problem for partial monitored languages}

Then, we formulate the monitoring problem of a partial timed quantitative word $\word$ against a specification in an LHA $\A$. 
In the monitoring problem,
the total valuations $\completeWordVal\colon\Variables \to \R$ compatible with the observation are synthesized as well as 
the indices $i$ satisfying $\word[i] \in \partialLg(\A)$ are returned.

\smallskip
\defProblem{The $\partialLg$ membership constraint} {
An LHA $\A$ and
a partial timed quantitative word $\word = (\wordVal_1, \tau_1), \dots, (\wordVal_m, \tau_m)$.}{
Returns  (constraints that express) the set $\matching(\word, \A)$, consisting of all the pairs $(i, \completeWordVal)$ of 
an index $i \in \{1,2,\dots,m\}$ and a total valuation $\completeWordVal\colon \Variables \to \R$ satisfying
$\wordVal_i \sqsubseteq \completeWordVal$ and
$(\wordVal_1,\tau_1),(\wordVal_2,\tau_2),\dots,(\wordVal_{i-1},\tau_{i-1}),(\completeWordVal,\tau_i) \in \partialLg(\A)$.
}
\medskip
\subsection{Algorithms for the membership constraint problem}

Here we show how to modify the algorithms in \cref{section:algorithm1,section:algorithm2} to solve the membership constraint problem. 

\myparagraph{Algorithm I in \cref{section:algorithm1}}

We can directly reuse the algorithm in \cref{section:algorithm1} because of the following reasons.

\begin{itemize}
 \item Although the valuation in the partial timed quantitative word $\word$ can be partial, the construction of the LHA $\A_{\word}$ is the same; we simply convert each partial valuation to a guard.
 \item The construction of the synchronized product $\A_{\word} \product \A$ is the same.
 \item We need to synthesize the constraint for the reachability rather than deciding the reachability. We can use \phaverlite{} to synthesize such a constraint.
\end{itemize}
\myparagraph{Algorithm II in \cref{section:algorithm2}}

Since \cref{algorithm:incremental:outline} does not conduct the constraint synthesis, we need the following slight  modification.

\begin{itemize}
 \item The constraint $v = \wordVal_i$ in \cref{algorithm:incremental:outline:restrict} of \cref{algorithm:incremental:outline} is replaced with $\wordVal_i \sqsubseteq v$.
 \item We return the valuation rather than the Boolean acceptance. Namely, \cref{algorithm:incremental:outline:result} of \cref{algorithm:incremental:outline} becomes  $\Resulti{i} \gets \{\varval \mid \exists (\loc, \varval) \in \Conf_{i}.\, \loc\in\LocFinal \}$.
\end{itemize}
\end{PartialObservationVersion}

\section{Conclusions}\label{section:conclusion}
Based on a novel language notion $\Lmon$ for LHAs, we formulated what we call the \emph{model-bounded monitoring} scheme for hybrid systems.
It features the use of a bounding model $\A$ of the system to bridge the gap between (continuous-time) system behaviors and (discrete-time) logs that a monitor can access.
While the $\Lmon$ membership problem is undecidable, our two partial algorithms (especially our dedicated \masakiTool{}) work well for benchmarks on automotive platooning, robot navigation, and gas burner controlling. 
Although it remains exponential in the dimension, mainly due to polyhedra operations, our technique is \emph{scalable}, in the sense that it is linear (and even less than linear) in the size of the input timed quantitative word.
When the bounding model $\A$ is imprecise, there can be false alarms, \ie{} safe behaviors can be detected as a safety violation.
In contrast, our experiment results show that when the sampling interval is short enough, false alarms occur only for the behavior with small robustness.
Overall, our results show the power of symbolic manipulation of polyhedra in the monitoring of cyber-physical systems. 

\section{Perspectives}\label{section:perspectives}
The following summarizes some of the future directions of model-bounded monitoring.

\paragraph{Optimizing \masakiTool{}}
So far, \masakiTool{} is a quickly developed prototype.
Optimization of the reachability analysis with the technique in~\cite{BZ19,CAF11} is a first future work.
In addition, importing the latest optimizations from \phaverlite{}~\cite{BZ19} into \masakiTool{} is on our agenda.
Further evaluation of the precision of model-bounded monitoring using a real-world industry example is also a future work.

\paragraph{Monitoring against temporal properties}
Another immediate future work is an extension for monitoring against a temporal property, \eg{} represented by a temporal logic formula or an automaton instead of a conjunction of linear (in)equalities.
Since many realistic properties are temporal (see \eg{}~\cite{KJDDYIKKS2016,ARCH21:ARCH_COMP_2021_Category_Report}),
the practical benefit by such an extension would be high.
We note that if the violation of the property is represented by an LHA $\A'$, such a temporal extension is straightforward---the LHA $\A_{\neg \varphi}$ is the synchronous product of the bounding model $\A$ and the property LHA $\A'$.

\paragraph{Further uncertainties in the observation}
By using a bounding model, we conduct monitoring taking into account of the uncertain behavior between the samples due to the lack of the observation.
In practice, there are many other uncertainties in the observation, \eg{} some values are unobservable in some samples, we have to predict the future behavior, and the obtained values have some measurement error.
One of the future directions is to utilize a bounding model to take into account of such uncertainties.

\paragraph{Uncertainty in the specification}
One of the practical obstacles in monitoring is to precisely decide the specification.
For instance, although it is natural to require that the distance between any two cars should be kept large enough in an ACC system, it is not always straightforward to define the exact threshold of the safe distance.
Utilizing \emph{parameters} in the specification is a solution for such an issue~\cite{AHW18,WA19,WAH19}.
Parametric monitoring consists in exhibiting parameter valuations for which a specification is violated (or correct). Parametrization of the current framework is future work.

\paragraph{Quantitative monitoring}
One potential  extension of the $\Lmon$ membership problem is to make it more quantitative: returning a \emph{distance} between the monitored language $\Lg(\A)$ and the observation $\word$ rather than checking the membership of the observation $\word$ to the monitored language $\Lg(\A)$.
This makes the result of the analysis even more useful.
We note that this future direction is related to the algorithm~\cite[Section 4]{FP09} to estimate the worst-case robustness, but the use of LHAs would give a finer estimation.

\paragraph{Bounding models for specifying signal frequencies}
When we specify the behavior of periodic signals, signal \emph{frequency} is one of the essential features.
One of the future directions is to use a bounding model to specify the frequencies of the signals, \eg{} using the LHA construction in~\cite{CSF11}.

\paragraph{Testing-based approach}
One of the limitations of our approach is the scalability, especially \wrt{} the dimension of the bounding model.
A testing-based approach, \eg{} in~\cite{DN09}, instead of the current \emph{guaranteed} approach by the convex polyhedra analysis, is future work for more efficiency at the expense of the safety guarantee.
\ifdefined\VersionForArXiV\else
\begin{acks}
	\ouracks{}
\end{acks}
\fi
	\newcommand{\CCIS}{Communications in Computer and Information Science}
	\newcommand{\ENTCS}{Electronic Notes in Theoretical Computer Science}
	\newcommand{\FMSD}{Formal Methods in System Design}
	\newcommand{\IJFCS}{International Journal of Foundations of Computer Science}
	\newcommand{\IJSSE}{International Journal of Secure Software Engineering}
	\newcommand{\JLAP}{Journal of Logic and Algebraic Programming}
	\newcommand{\JLC}{Journal of Logic and Computation}
	\newcommand{\LMCS}{Logical Methods in Computer Science}
	\newcommand{\LNCS}{LNCS}
	\newcommand{\RESS}{Reliability Engineering \& System Safety}
	\newcommand{\STTT}{International Journal on Software Tools for Technology Transfer}
	\newcommand{\TCS}{Theoretical Computer Science}
	\newcommand{\ToPNoC}{Transactions on Petri Nets and Other Models of Concurrency}
	\newcommand{\TSE}{IEEE Transactions on Software Engineering}
        \newcommand{\FGCS}{Future Generation Computing Systems}
        \newcommand{\SCP}{Science of Computer Programming}

\ifdefined\VersionForArXiV
	\renewcommand*{\bibfont}{\small}
	\printbibliography[title={References}]
\else
      \bibliographystyle{ACM-Reference-Format}
	\bibliography{ParamHAMon}
\fi

\end{document}